\documentclass{article}

\pdfoutput=1

\newcommand{\aka}{a.k.a.\ }
\newcommand{\cf}{cf.\ }
\newcommand{\st}{s.t.\ }

\usepackage{xspace}
\newcommand{\etc}{etc.\xspace}
\newcommand{\eg}{e.g.,\xspace}
\newcommand{\ie}{i.e.,\xspace}

\usepackage{array}
\usepackage[pdftex,dvipsnames]{color}
\usepackage{ifthen,calc,url}

\usepackage{amstext,amsmath,amssymb,stmaryrd,rotating}
\usepackage{amsthm}
\usepackage[all,cmtip]{xy}
\usepackage{extarrows}

\newcommand{\defeq}{\mathrel{:=}}
\newcommand{\defiff}{\mbox{:iff}}
\newcommand{\bnfeq}{\mathrel{::=}}
\newcommand{\bnfor}{\ \big| \ }

\newcommand{\set}[1]{\{#1\}}
\newcommand{\setst}[2]{\{\ #1\ \boldsymbol{|}\ #2\ \}}
\newcommand{\equivclass}[2]{[#1]_{#2}}

\newcommand{\powerset}[1]{2^{#1}}%{\mathrm{p}}
\newcommand{\powersetFinite}[1]{2_{\text{finite}}^{#1}}%{\mathrm{p_{f}}}

\newcommand{\nn}[1]{\makebox[2em][r]{#1}\ \ }
\newcommand{\true}{\text{true}}
\newcommand{\false}{\text{false}}

\newcommand{\Scomb}[1]{\mathtt{S}_{#1}}
\newcommand{\Kcomb}[1]{\mathtt{K}_{#1}}
\newcommand{\reduc}[1]{\xlongrightarrow{#1}}

\newcommand{\gscc}[2]{\mathtt{succ}_{#1}^{#2}}

\newcommand{\CM}{\mathtt{CM}}
\newcommand{\agents}{\mathcal{A}}
\newcommand{\community}{\mathcal{C}}

\newcommand{\messages}{\mathcal{M}}
\newcommand{\data}{\mathcal{D}}

\newcommand{\pair}[2]{(#1,#2)}
\newcommand{\sign}[2]{{\{\negmedspace[#1]\negmedspace\}}_{#2}}
\newcommand{\encrypt}[2]{[#1]_{#2}}
\newcommand{\hash}[1]{\lceil#1\rceil}
\newcommand{\quoted}[1]{\ulcorner#1\urcorner}

\newcommand{\derives}[3]{#3\vdash_{#1}#2}

\newcommand{\denotation}[4]{{_{#2}^{}{\llbracket#1\rrbracket}_{#3}^{#4}}}

\newcommand{\provesEq}[3]{\mathrel{_{#1}{\equiv}_{#2}^{#3}}}

\newcommand{\states}{\mathcal{S}}

\newcommand{\msgs}[1]{\mathrm{msgs}_{#1}}
\newcommand{\preorder}[1]{\leq_{#1}}
\newcommand{\indist}[3]{#2\equiv_{#1}#3}

\newcommand{\pAccess}[3]{\mathrel{_{#1}\negthinspace\mathrm{R}_{#2}^{#3}}}
\newcommand{\access}[3]{\mathrel{_{#1}\negthinspace\mathcal{R}_{#2}^{#3}}}

\newcommand{\clo}[2]{\mathrm{cl}_{#1}^{#2}}
\newcommand{\Clo}[2]{\mathrm{Cl}_{#1}^{#2}}
\newcommand{\LiP}{\mathrm{LiP}}
\newcommand{\pFormulas}{\mathcal{L}}

\renewcommand{\true}{\top}
\renewcommand{\false}{\bot}

\newcommand{\relsym}[1]{\thinspace{#1}\thinspace}
\newcommand{\knows}[2]{#1\relsym{\mathsf{k}}#2}

\newcommand{\limp}{\rightarrow}
\newcommand{\lequiv}{\leftrightarrow}

\newcommand{\K}[1]{\mathsf{K}_{#1}}

\newcommand{\CK}[1]{\mathsf{CK}_{#1}}

\newcommand{\proves}[4]{#1\relsym{:_{#3}^{#4}}#2}
\newcommand{\refutes}[4]{#1\relsym{{\div}_{#3}^{#4}}#2}

\newcommand{\proofdiamond}[4]{#1\relsym{{\diamond}_{#3}^{#4}}#2}%{#1\relsym{{\cdot\thinspace\cdot}_{#3}^{#4}}#2}
\newcommand{\decides}[4]{#1\relsym{{\veebar}_{#3}^{#4}}#2}
\newcommand{\undecides}[4]{#1\relsym{{\barwedge}_{#3}^{#4}}#2}

\newcommand{\Limp}{\Rightarrow}
\newcommand{\Lequiv}{\Leftrightarrow}

\newcommand{\aModalFrame}{\mathfrak{S}}

\newcommand{\LPded}{\vdash_{\mathrm{LP}}}
\newcommand{\LiPded}{\vdash_{\LiP}}
\newcommand{\LiPdedBis}{\mathrel{{\dashv}{\vdash}_{\LiP}}}

\theoremstyle{definition}
\newtheorem{definition}{Definition}
\newtheorem{remark}{Remark}
\newtheorem*{example}{Example}
\theoremstyle{theorem}
\newtheorem{fact}{Fact}
\newtheorem{proposition}{Proposition}
\newtheorem{lemma}{Lemma}
\newtheorem{theorem}{Theorem}
\newtheorem{corollary}{Corollary}
\theoremstyle{definition}
\newtheorem*{acknowledgements}{Acknowledgements}

\begin{document}

\title{A Logic of Interactive Proofs\\ 
		(Formal Theory of Knowledge Transfer)\thanks{%
			Work funded mostly with Grant P~08742 from the Japan Society for the Promotion of Science, and 
			to a lesser extent with Grant AFR~894328 from the National Research Fund Luxembourg  
				cofunded under the Marie-Curie Actions of the European Commission (FP7-COFUND).
		A preliminary version of  
			Section~\ref{section:MotivationGoalProblem} and \ref{section:SolutionMethodology}  
				appeared in \cite[not peer-reviewed]{iFoundationiComputation:CiE} and 
		another of Sections~\ref{section:MotivationGoalProblem}---\ref{section:IntroductionLP}
				in \cite[peer-reviewed]{iFoundationiComputation:AiML}.}}
\author{Simon Kramer\\[\jot]
		SK-R\&D Ltd liab.\ Co\\
		\texttt{simon.kramer@a3.epfl.ch}}

\maketitle

\begin{abstract}
	We propose a \emph{logic of interactive proofs} as a framework for 
		an intuitionistic foundation for interactive computation, which we construct via 
			an interactive analog 
				of the G\"odel-McKinsey-Tarski-Art\"{e}mov definition of 
					Intuitionistic Logic as embedded into 
					a classical modal logic of proofs, and
				of the Curry-Howard isomorphism between
					intuitionistic proofs and
					typed programs.
	Our interactive proofs effectuate
		a persistent epistemic \emph{impact} in their intended communities of peer reviewers that
			consists in the induction of the (propositional) knowledge of their proof goal by 
				means of the (individual) knowledge of the proof with the interpreting reviewer.
	That is,
		interactive proofs effectuate  
			a \emph{transfer} of propositional knowledge (knowable facts) via 
				the transmission of certain individual knowledge (knowable proofs) in 
					multi-agent distributed systems.	
	In other words, 
		we as a community can have the formal common knowledge that 
			a proof is that which 
			if known to one of our peer members 
			would induce the knowledge of its proof goal with that member.
	Last but not least, 
		we prove non-trivial interactive computation as definable within our 
			simply typed \emph{interactive Combinatory Logic}
				to be nonetheless equipotent to non-interactive computation as defined by 
					simply typed Combinatory Logic.

\medskip
\noindent
\textbf{Keywords\ }
	agents as 
		knowledge processors, 
		proof- and signature-checkers, and 
		reduction relations; 
	communication channels as variable bindings;
	constructive Kripke-semantics; 
	cryptographic and interpreted communication; 
	interactive and oracle computation; 
	proof equality;
	proofs as certificates, messages, and sufficient evidence; 
	propositions as types.
\end{abstract}

\tableofcontents

\section{Introduction}
The subject matter of this paper is a formal logic of interactive proofs serving as a framework for 
	an intuitionistic foundation of interactive computation between 
		multiple agents in distributed systems.	
Due to its ambitious goal, 
	it cannot be light-weight, and comprehending it may well 
		resemble the experience described by the Indian parable of the blind men and an elephant.
In order to help our readers perceive and find again the various parts of our elephant, 
	we have chosen to prefer serious structuring and (perhaps over-)emphasising over 
		more light-weight linear chatting, which would be more entertaining but 
			in our opinion also distracting (important points and structure could go unnoticed or get lost).

\subsection{Motivation, Goal \& Problem}\label{section:MotivationGoalProblem}
\subsubsection{Motivation}\label{section:Motivation}
In \cite{InteractiveComputation}, 
	interactive computation is proposed 
		as the new, \emph{to-be-defined} paradigm of computation,
		as opposed to the old paradigm of non-interactive computation in the sense of 
			the old sages like Turing and others.
The motivation for this paper is the consensus of the contributors to \cite{InteractiveComputation}, which is that 
	the purpose of interactive computation ultimately is 
		not the computation of result values, to which we consent, but 
		the possibly unending interaction \emph{itself}, from which we dissent. 
Interaction may well be unending, but 
	it cannot be a self-purpose because 
	if it were then all interactive programs would be quines---rhetorically exaggerated.
(A quine [re]produces itself and only itself.)

\subsubsection{Goal}\label{section:Goal}
Our goal is to reach consensus with the reader that
	values are only the means---not the ends---of interactive computation, and that
	\textbf{the purpose of interactive computation is \emph{interpreted communication} between 
		distributed man or machine agents interacting via message passing.}
Note that the communication medium can be modelled as a machine agent.
For example in communication security, which is an important application of interactive computation, 
	the communication medium \emph{is} an adversarial agent \cite{LIiP:ACMTOCL}.

\subsubsection{Problem}\label{section:Problem}
So what is interpreted communication?
According to Shannon \cite{ShannonCommunication}:\footnote{The standard typographic convention of 
		brackets occurring within in-lined or displayed quoted text indicates that
			the text within the brackets does not occur in the original text.
So,
	``[\emph{un}interpreted]'' indicates that ``\emph{un}interpreted'' does not occur in the original text.}
	\begin{quote}
		The fundamental problem of [\emph{un}interpreted] communication is that of reproducing at one point 
		either exactly or approximately a message selected at another point.
	\end{quote}
In analogy, we declare:
	\begin{quote}
		The fundamental problem of \emph{interpreted} communication is that of [re]pro\-ducing at one point 
		either exactly or approximately the \emph{intended meaning} of a message selected at another point.
	\end{quote}
Note that
	due to the distribution of the different agents in a communication network,
			who may have different views of the system, 
				the agents constitute different message \emph{interpretation contexts}.
Hence, 
	identical messages may well 
		be interpreted differently in different contexts, and thus 
		have different meanings to different agents.
As a matter of fact,
	message misinterpretations are ubiquitous in man or machine communications, \eg 
		in communication protocols \cite[Chapter~3]{SecurityEngineering}, and
	may have serious or even catastrophic consequences, \eg 
		in the context of nuclear command and control \cite[Chapter~13]{SecurityEngineering}.
Indeed, 
	[re]producing intended message meaning across interpretation contexts is a highly critical and non-trivial problem.
But what is message meaning?

In \cite{skramer2007FCS-ARSPA-LORI}, we argue that 
	the (denotational) meaning of a message in a given interpretation context is  
		the \emph{propositional knowledge} that the \emph{individual knowledge} of that message induces in that context
			(\cf Section~\ref{section:EpistemicExplication}).
(See \cite{ParikhRamanujam} for a related notion of message meaning.)
By individual knowledge we mean
		knowledge in the sense of the transitive use of the verb ``to know,''
			here to know a message, such as the plaintext of an encrypted message.
Notation: $\knows{a}{M}$ for ``agent $a$ knows message $M$'' 
	(\cf Definition~\ref{definition:LiPLanguage}).
This is the classic concept of knowledge \emph{de re} (``of a thing'') 
	made explicit for message things and 
	meaning here taking them apart (analysing) and putting them again together (synthesising).
Whereas by propositional knowledge we mean
	knowledge in the sense of the use of the verb ``to know'' with a clause,  
		here to know that a statement is true, 
			such as that the plaintext of an encrypted message is (individually) unknown to potential adversaries.
Notation: $\K{a}(\phi)$ for ``agent $a$ knows that $\phi$ [is true]''
	(\cf Section~\ref{section:EpistemicExplication}).
This is the classic concept of knowledge \emph{de dicto} (``of a fact'').\footnote{
	In a first-order setting, 
		knowledge \emph{de re} and \emph{de dicto} can
			be related with Barcan-laws.}
Notice that we make 
	the distinction between individual and propositional knowledge with respect to the \emph{``object''} of knowledge (the know\emph{n}),
		that is, with respect to a message and clause, respectively.
However, individual as well as propositional knowledge can both be individual with respect to the \emph{subject} of knowledge (the know\emph{er}), 
	that is, an (individual) agent.

Hence, an agent-centric paraphrase of our previous problem statement is:
	\begin{quote}\label{page:TheFundamentalProblemofCommunication}
		The fundamental problem of communication is that of inducing at one point 
		either an intended knowledge or an intended belief with
			a message selected at another point
				(\cf Section~\ref{section:EpistemicExplication} for formal meanings).
	\end{quote}
With this paper,
	we 
		intend to induce (necessarily true) knowledge, and
		leave induction of (possibly false) belief for further work.
(For our standard notions of belief and knowledge, see \cite{MultiAgents}.)
Here,
	interactive computations compute 
		propositional \emph{knowledge} 
			(\eg that the goal of this paper has been achieved), and
				they do so by passing as messages pieces of interactively or non-interactively computed individual knowledge 
					(\eg this paper).
Again, result values are only the means---not the ends---of interactive computations.

\subsection{Solution \& Methodology}\label{section:SolutionMethodology}
\subsubsection{Solution}\label{section:Solution}
Our problem statement contains 
	an inceptive solution and defining principle for interactive computation, namely  
		\emph{induction of knowledge} (\cf Section~\ref{section:EpistemicExplication}).
Our task is thus to make this principle precise.
This in turn 
	leads us to defining the concept of an \emph{interactive proof} (or \emph{certificate}) whose 
		effect is to induce the knowledge of its proof goal (or statement of certification) in 
			the intended interpretation context 
				(\cf Section~\ref{section:EpistemicExplication}).
Thereby, we identify proofs (and thus certificates) with messages 
	(for more views on these three concepts, 
		see \cite{Asperti:ProofMessageCertificate})):
	\begin{description}
		\item[\textbf{Messages as proofs}] Any (well-formed) message, 
				cryptographic or not, can prove something (cf.\ Theorem~\ref{theorem:SomeUsefulDeducibleLogicalLaws}.15) and thus is 
					a (potential) proof.
		\item[\textbf{Proofs as messages}] Any proof (a finite syntactic object) can obviously be transmitted and thus is 
				a (potential) message.
	\end{description}
The present paper is intended to be such an interactive proof: 
	its proof goal is the goal stated in Section~\ref{section:Goal}, and
	its intended interpretation context is the set of logically educated readers fluent in English.
Note that 
	our interactive proofs, just 
		like proof constructions defined in terms of interactive Turing-machines \cite{ZeroKnowledge}, 
			need not have a proof-tree structure in the sense of classical proof theory \cite{HBProofTheory}.
Our interactive proofs are also \emph{formal social} proofs in that 
	they partially reconcile 
		two distinct viewpoints on mathematical proofs \cite{ProofTheoryIntro}:
	\begin{quotation}
		The first view is that proofs are social conventions by which 
			mathematicians convince one another of the truth of theorems.
		That is to say, a proof is expressed in natural language plus
			possibly symbols and figures, and is sufficient to 
			convince an expert of the correctness of a theorem.
		Examples of social proofs include the kinds of proofs that
			are presented in conversations or published in articles.
		Of course, it is impossible to precisely define what
			constitutes a valid proof in this social sense; and,
			the standards for valid proofs may vary with the audience and over time.
		The second view of proofs is more narrow in scope: in this view,
			a proof consists of a string of symbols which satisfy 
			some precisely stated set of rules and which prove a theorem,
			which itself must also be expressed as a string of symbols.
		According to this view, mathematics can be regarded as
			a `game' played with strings of symbols according to
			some precisely defined rules. 
		Proofs of the latter kind are called ``formal'' proofs 
			to distinguish them from ``social'' proofs.
	\end{quotation}
Note that 
	a theorem known by one (say $a$) but not by another mathematician  (say $b$)
		is a local truth from the viewpoint of 
			an audience (say $\set{a,b}$).
An example of a social convention is a work contract 
	(\cf Lemma~\ref{lemma:GettierSigning} and 
		Corollary~\ref{corollary:GettierSigning}).

\subsubsection{Methodology}
Our methodology for defining interactive computation emerges as 
	an interactive variant of a classical construction that consists 
		in a ``horizontal'' transitive \emph{embedding} of programs into proofs and
		in a ``vertical'' homomorphing of each non-interactive structure into its interactive counterpart
			(\cf Figure~\ref{figure:Methodology}).
We will argue that 
	the lower right-most ``vertical'' homomorphism (without $\subset$-tail)  
	cannot be an embedding (with $\subset$-tail) and
		that this reflects the essential difference between  
			interactivity and non-interactivity here.
So here, we shall present:\label{page:Methodology}
\begin{enumerate}
	\item a classical modal logic (LiP) of \emph{interactive proofs} that
		\begin{enumerate}
			\item are agent-centric generalisations of non-interactive proofs such that
					the agents are resource-\emph{un}bounded with respect to 
						individual and thus also propositional knowledge 
							(\cf Section~\ref{section:LogicalLaws}), though 
								our agents here are still unable 
									to guess individual (and thus also propositional) knowledge;
			\item induce the knowledge of their proof goal with 
					their intended interpreting agent(s) such that
						the induced knowledge is propositional in the sense of 
							the standard modal logic of knowledge S5
								\cite{Epistemic_Logic,MultiAgents,EpistemicLogicFiveQuestions}. 
		\end{enumerate}
		(See \cite{skramerPhDthesis,skramerCPLDolevYao} and \cite{skramerIMLA2008} for 
			preliminary, non-axiomatic explorations within different, non-standard semantics, but 
				in \cite{skramerIMLA2008} already with pairing and signing as proof-term constructors.)
	\item a classical modal logic (iS4) of \emph{interactive provability} via 
			an embedding into an epistemically guarded first-order extension (egFOLiP) of LiP in 
				loose analogy with Art\"{e}mov's embedding of the standard modal logic of 
					non-arithmetic\footnote{\ie not internalising provability of a formal system that 
												includes Peano Arithmetic} 
						provability S4 into his Logic of Proofs LP 
							\cite{LP,ArtemovBSL,ModalLogicInMathematics}.
			(To make the analogy tighter, one could equip LiP with term variables as LP is equipped with,
				and attempt a corresponding embedding. However, we find our present embedding a bit more natural in our setting.)	
	\item \emph{interactive Intuitionistic Logic} (iIL) via an embedding into iS4 in analogy with 
				the G\"odel-McKinsey-Tarski embedding of Intuitionistic Logic IL into S4 
					\cite{ModalLogicInMathematics}, which will turn out to be the non-modal fragment of
						\cite{LIiP:ACMTOCL}.
	\item \emph{(untyped) interactive Combinatory Logic} (iCL) as a multi-agent distributed generalisation
			of the classic Combinatory Logic (CL) \cite{LambdaCalculusAndCombinators} and as a parent formalism for its following simply typed variant.
	\item \emph{typed interactive Combinatory Logic} (TiCL) via an isomorphism from iIL in analogy with 
				the Curry-Howard isomorphism between IL and typed Combinatory Logic (TCL) 
					\cite{CurryHoward,LecturesOnTheCurryHowardIsomorphism}.
\end{enumerate}
We will deploy our methodology from right to left.
		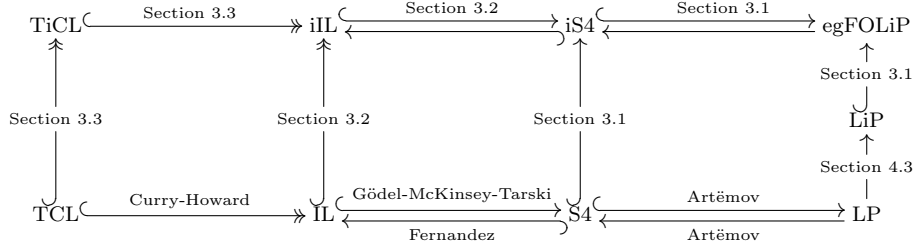
\begin{figure*}[t!]
			\centering
			\caption{Typed interactive programs from interactive proofs}
			\medskip
			\footnotesize
			$\xymatrixcolsep{80pt}\xymatrix{%
				\text{TiCL} 
					\ar@{^{(}->>}[r]^{\text{Section~\ref{section:TiCL}}} &
				\text{iIL} 
					\ar@<2pt>@{^{(}->}[r]^{\text{Section~\ref{section:iIL}}} &
				\text{iS4} 
					\ar@<2pt>@{^{(}->}[l] 
					\ar@<2pt>@{^{(}->}[r]^{\text{Section~\ref{section:iS4}}} &
				\text{egFOLiP} 
					\ar@<2pt>@{^{}->}[l]\\
				&&& 
				\text{LiP} 
					\ar@{^{(}->}[u]|-{\strut\text{Section~\ref{section:iS4}}}\\
					\text{TCL} \ar@{^{(}->>}[r]^{\text{Curry-Howard}} 
					\ar@{^{(}->>}[uu]|-{\strut\text{Section~\ref{section:TiCL}}} &
				\text{IL} 
					\ar@<2pt>@{^{(}->}[r]^{\text{G\"odel-McKinsey-Tarski}} 
					\ar@{^{(}->>}[uu]|-{\strut\text{Section~\ref{section:iIL}}} &
				\text{S4} 
					\ar@<2pt>@{^{(}->}[l]^{\text{Fernandez}} 
					\ar@<2pt>@{^{(}->}[r]^{\text{Art\"{e}mov}} 
					\ar@{^{(}->}[uu]|-{\strut\text{Section~\ref{section:iS4}}} &
				\text{LP} 
					\ar@<2pt>@{^{}->}[l]^{\text{Art\"{e}mov}} 
					\ar@{^{}->}[u]|-{\strut\text{Section~\ref{section:LPmapstoLiP}}}}$
		\label{figure:Methodology}
		\end{figure*}
egFOLiP (LP) is the richest among all the (non-)interactive structures in the sense that
	all other (non-)interactive structures embed into egFOLiP (LP).
In result, 
	terms viewed as proofs are descriptions of constructive deductions, 
	terms viewed as programs are prescriptions for interactive computations, 
	LiP-formulas viewed as propositions are proof goals, and
	LiP-formulas viewed as types are program properties.
To agents,
	interactive proofs are message terms that
		induce the propositional knowledge of their proof goal with their intended interpreters, and
	interactive computations are message communications between distributed interlocutors that
		compute that knowledge from the meaning of the communicated messages.
In sum,
	\textbf{\emph{the purpose of interactive proofs is 
		the transfer of propositional knowledge (knowable facts) via 
			the transmission of certain individual knowledge (knowable proofs) in 
				multi-agent distributed systems}} 
				(for example: editorial boards, 
					scientific communities, 
					social networks and
					other virtualised societies---even 
					the whole Internet).
(In that,
	LiP can also be viewed as an inductive but truthful and thus really a deductive logic, which 
		is about drawing conclusions [obtaining factual knowledge] from observed data [individual knowledge]
			\cite{sep-logic-inductive}.)
That is, \textbf{\emph{LiP is a formal theory of knowledge transfer.}}\footnote{Knowledge being power,
	knowledge transfer between agents becomes agent empowerment.}
In contrast, Shannon's theory is about the 
	(error-correcting) transmission of individual knowledge (data) only.

\subsection{The Logic of Proofs (LP)}\label{section:IntroductionLP}
The language of Art\"{e}mov's Logic of Proofs (LP) 
	(\cf \cite{LP,ArtemovBSL} and \cite[Section~5]{ModalLogicInMathematics}) is 
	the language of classical propositional logic enriched with
	formulas $p{:}F$, where
		$F$ denotes formulas and
		$p$ so-called proof polynomials.\footnote{%
			In that, LP (and LiP) is an example of Gabbay's labelled deductive systems (LDS), and
				its (and LiP's) labels can be appreciated as more concrete (binder-free) implementations of 
					abstract (Kripke-world-quantifying) labels \emph{\`{a} la}  
						\cite{TheFunctionalInterpretationOfLogicalDeduction}.}
Proof polynomials are terms built from 
	proof variables $x,y,z,\ldots$ and 
	proof constants $a,b,c,\ldots$ by means of 
		three operations: 
			application `${\cdot}$' (binary), 
			sum `$+$' (binary), and 
			proof checker `$!$' (unary).
According to Art\"{e}mov,
	proof polynomials represent the whole set of possible operations on (non-interactive) proofs for a propositional language.

Then, the following proof system defines (the non-normal modal logic) LP:
\begin{enumerate}\setcounter{enumi}{-1}
	\item all axioms of classical propositional logic
	\item $\LPded ((p{:}F)\lor q{:}F)\limp (p{+}q){:}F$\quad(sum)
	\item $\LPded (p{:}(F\limp G))\limp((q{:}F)\limp (p{\cdot}q){:}G)$\quad(application)
	\item $\LPded (p{:}F)\limp F$\quad(reflection)
	\item $\LPded (p{:}F)\limp (!p){:}(p{:}F)$\quad(proof checker)
	\item $\set{F\limp G, F}\LPded G$\quad(\emph{modus ponens})
	\item $\LPded c{:}A,$ for any axiom $A$ and proof constant $c$\quad(constant specification\footnote{%
		Constant specification is a somewhat flexible concept (\cf 
			\cite{JustificationLogic} for four variations).}),
\end{enumerate}
where 	
	$\set{F\limp G, F}\LPded G$ abbreviates ``if $\LPded F\limp G$ and $\LPded F$ then $\LPded G$'' in 
		horizontal Hilbert style, and
	`$!$' is interpreted as a primitive-recursive program for checking the correctness of proofs
		which given a proof of $p$ produces a proof that $p$ proves $F$.
The application axiom internalises the \emph{modus ponens} rule.

Note that
	LP does not explicate beyond the formula $p{:}F$ what it means for $p$ to prove $F$, but
		rather attempts to characterise axiomatically this relation.
Indeed, $p{:}F$ really stands for an atomic concept.
Arguably, 
	the standard semantics of LP does not fully explicate the concept either:
		that semantics actually merely re-stipulates each axiom of LP as 
			a corresponding condition on the model
				in set-theoretic language 
					(\cf \cite[Section~5.3]{ModalLogicInMathematics} and \cite{LPsemantics}).
In that, it is rather a convenient semantic \emph{interface} than a semantics proper.
Here it is.
Given 
	\begin{itemize}
		\item a frame $(W,R)$ with a reflexive and transitive relation $R\subseteq W\times W$
		\item an abstractly constrained \emph{evidence mapping} $\mathcal{E}$ 
			from 
				worlds $u$ and 
				proof polynomials $p$
			to sets of formulas $F$
	\end{itemize}
such that:
\begin{enumerate}
	\item if $uRv$ then $\mathcal{E}(u,p)\subseteq\mathcal{E}(v,p)$\quad(monotonicity)
	\item (closure)
		\begin{itemize}
			\item if $F\rightarrow G\in\mathcal{E}(u,p)$ and $F\in\mathcal{E}(u,q)$ then $G\in\mathcal{E}(u,p{\cdot}q)$\quad(application)
			\item if $F\in\mathcal{E}(u,p)$ then $p{:}F\in\mathcal{E}(u,!p)$\quad(proof checker)
			\item $\mathcal{E}(u,p)\cup\mathcal{E}(u,q)\subseteq\mathcal{E}(u,p{+}q)$\quad(sum),
		\end{itemize}
\end{enumerate}
and a usual valuation mapping $\mathcal{V}$ from atomic propositions to sets of worlds,
satisfaction for the LP-modality in a model $(W,R,\mathcal{E},\mathcal{V})$ at a world $u$ is so that
\begin{multline*}
	\text{$(W,R,\mathcal{E},\mathcal{V}),u\Vdash p{:}F$ iff}\\
			\text{$F\in\mathcal{E}(u,p)$ and for every $v\in W$,
					if $uRv$ then $(W,R,\mathcal{E},\mathcal{V}),v\Vdash F$.}
\end{multline*}
Notice the additional constraint $F\in\mathcal{E}(u,p)$, which
	a standard Kripke-seman\-tics format would not allow.
A more serious criticism than the one of not being a semantics proper is that
	in a truly interactive setting, 
		the reflection axiom is unsound (\cf Section~\ref{section:Concepts}).
By a truly interactive setting we mean a multi-agent distributed system 
	where not all proofs are known by all agents, that is,  
		a setting with a non-trivial distribution of information 
			(in the sense of Dana Scott, \cf Proposition~\ref{proposition:PropertiesOfDerivability}).

In contrast:
	\begin{enumerate}
		\item LiP will give an \emph{epistemic explication} of proofs, that is, 
			an explication of proofs in terms of the epistemic impact that 
					they effectuate with their intended interpreting agents 
						(\ie the knowledge of their proof goal).
			
			Technically,
				we will endow the proof modality with a \emph{standard} Kripke-semantics, whose
					accessibility relation we 
						\begin{enumerate}
							\item define \emph{constructively}, in terms of 
								elementary set-theoretic constructions 
									(in loose analogy with 
										the constructive rather than 
										the purely axiomatic definition of 
											ordered pairs [\eg Kuratowski's] or 
											numbers \cite{TheNumberSystems});
							\item match to an abstract semantic interface in standard form 
									(which abstractly stipulates the characteristic properties of 
										the accessibility relation \cite{ModalProofTheory}), which 
											entails 
										\begin{enumerate}
											\item the absorption of 
												the evidence mapping into the accessibility relation
													(and thus the absorption of the corresponding conjunctive constraint on
														the truth condition of the proof modality),  
											\item the elimination of the monotonicity constraint on 
													the evidence mapping (in the sense that 
														the constraint will become a property), which is
															a nice side-effect of the previous simplification.
										\end{enumerate}												
						\end{enumerate}
		\item LiP only validates a corresponding \emph{conditional reflection principle}, that is, 
				a reflection principle that 
					is conditioned on the (individual) knowledge of the proof mentioned by the principle 
						(\eg the above $p$ in LP).
		\item LiP is, technically speaking, a \emph{normal} modal logic, which 
				brings all the benefits of the existing standard techniques of normal modal logics to LiP.
	\end{enumerate}
Hence,
	we beg to differ with Art\"{e}mov and Nogina, who, like Aristotle and Plato, 
		define (propositional) knowledge as justified true belief, but unlike Aristotle and Plato,
			admit as admissible justifications for such knowledge only proofs in 
				the sense of at least LP \cite{JTB,JustificationLogic}.
As a counter-example to \emph{Art\"{e}mov and Nogina's provability explication of knowledge,} consider that
	an agent may know that a certain state of affairs is the case from the \emph{observation of a physical event} 
		(\eg a message input/output),
	yet not be able to prove her (propositional) knowledge to the non-observers (\eg an absent peer or judge)
	for lack of \emph{sufficient evidence} (\ie proof).
Whereas in \emph{our epistemic explication of provability,}
	provability possibly implies propositional knowledge, \eg with the (individual) knowledge of a proof, but
		propositional knowledge does not necessarily imply provability, \eg without such a proof.
The technical difference between the two philosophies may be subtle but nevertheless is serious 
	(\ie not a mere technicality)---especially for applications to 
		truly distributed computer systems.

\subsection{Contribution \& Roadmap}
\subsubsection{Contribution}\label{section:Contribution}
The contribution of this paper is, first, \emph{a formal theory of knowledge transfer,} that is, 
	our classical normal modal \emph{Logic of interactive Proofs (LiP),}  
	serving as a definitional framework for interactive computation %, or interpreted communication, 
	via a construction analogous to G\"odel-McKinsey-Tarski-Art\"{e}mov's, and second, 
		our resulting interactive structures iS4, iIL, and iCL and TiCL (the main computational structures), 
			as well as their interconnections.
More precisely, our main contributions are: 
	\begin{enumerate}
		\item a \emph{constructive Kripke-semantics} for LiP's proof modality 
				(\cf Section~\ref{section:Concretely});
		\item a stateful notion of transmittable \emph{interactive proofs} that
			\begin{enumerate}
				\item are \emph{agent-centric generalisations} of non-interactive proofs such that
						the agents are, as said, still resource-\emph{un}bounded with respect to 
							individual and thus also propositional knowledge,
				\item have intuitive \emph{epistemic explications} in that 
					\begin{enumerate}
						\item they effectuate (\cf Section~\ref{section:EpistemicExplication})
							\begin{enumerate}
								\item a persistent epistemic \emph{impact} in 
										their intended communities of peer reviewers that
											consists in the induction of the (propositional) knowledge of 
												their proof goal by means of the (individual) knowledge of 
													the proof with the interpreting reviewer,
								\item a \emph{transfer} of propositional knowledge (knowable facts) via 
										the transmission of certain individual knowledge (knowable proofs) in
											 multi-agent distributed systems,
							\end{enumerate}
						\item the individual proof knowledge can be thought of as being provided by an imaginary 
							\emph{computation oracle} 
								(\cf Section~\ref{section:OracleComputationalExplication}),
					\end{enumerate}
				\item are \emph{falsifiable} in a communal sense of Popper's 
						(\cf Theorem~\ref{theorem:Falsifiability}),
				\item can be 
						constructed with only two operations, namely 
							\emph{pairing} and  
							\emph{signing}, and 
						freely combined with other term operations (\eg \emph{encryption}),
				\item happen to have an 
					\emph{information-theoretic explication} in terms of 
						Scott's information systems (\cf Proposition~\ref{proposition:PropertiesOfDerivability});
			\end{enumerate}
		\item a stateful notion of 
				\emph{proof equality} in an idempotent commutative monoid capturing 
					\textbf{\emph{equality of epistemic impact}} (\cf Corollary~\ref{corollary:ProofEquality});
		\item a \emph{novel modal rule of logical modularity,} called \emph{epistemic antitonicity,} for 
			the class of justification logics \cite{JustificationLogic} including LP, which  
					allows the partial, or even total and thus modular generation of 
						the structural modal laws from 
							the laws of a separate (\eg application-specific) term theory 
					(\cf Page~\pageref{page:EpistemicAntitonicity} and Section~\ref{section:StructuralLaws});
		\item a \emph{sound and complete axiomatisation} for LiP (\cf Theorem~\ref{theorem:Adequacy});
		\item a proof of the \emph{finite-model property} (\cf Theorem~\ref{theorem:FMP}) and 
				\emph{decidability} (\cf Corollary~\ref{corollary:Decidability}) of LiP;
		\item the main interactive computational calculi of 
				Turing-powerful \emph{interactive Combinatory Logic (iCL)} and 
				its variant, \emph{typed interactive Combinatory Logic (TiCL),} 
					\begin{enumerate}
						\item which are distributed multi-agent generalisations of 
								their classic counterparts and enable us to view 
									\textbf{\emph{agents as reduction relations}} 
										(agent-centric reduction relation) and 
									\textbf{\emph{communication channels as variable bindings}} 
										(definable agent-centric lambda-operator),
						\item by means of which \textbf{\emph{interactive computation}} can be 
								defined and programmed, and 
								shown to be  
									\textbf{\emph{equipotent to non-interactive computation at 
										the level of simple types}} 
											(\cf Corollary~\ref{corollary:Equipotency})!
					\end{enumerate}
	\end{enumerate}
In sum, LiP 
	is a \emph{minimal modular extension} of propositional logic with 
	\begin{enumerate}
		\item an interactively generalised additional operator (the proof modality);
		\item a simplified and then interactively generalised 
			\begin{enumerate}
				\item proof-term language (only two instead of three constructors,
					\textbf{\emph{agents as proof- as well as signature-checkers}}),
				\item constructive Kripke-semantics 
					(including evidence-mapping absorption and monotonicity-condition elimination). 
			\end{enumerate}
		\item sufficient content to serve as a definitional framework for interactive computation.
	\end{enumerate}
With our contribution,
	we mean to concur with \cite[Page~viii]{NotesOnSetTheory}, where
	\begin{quote}  
	computation theory is viewed as part of the mathematics ``to be founded,''
	\end{quote} 
	since Kripke-models such as ours for LiP---conceived 
	as a foundation for \emph{interactive} computation theory---are 
	relational models of the meaning of modal languages in the language of set theory, which in turn \cite[Page~vii]{NotesOnSetTheory}
	\begin{quote} 
		is the official language of mathematics, just as mathematics is the official language of science.
	\end{quote}

\subsubsection{Roadmap}
In the next section,
	we introduce our Logic of interactive Proofs (LiP) axiomatically by means of 
		a compact closure operator 
			that induces the Hilbert-style proof system that we seek (\cf Proposition~\ref{proposition:Hilbert}) and
			that allows the simple generation of application-specific extensions of LiP (\cf Page~\pageref{page:DolevYao}).
We then prove some useful (further-used), deducible laws within the obtained system, 
	as well as a correspondence theorem (Theorem~\ref{theorem:EXIN}).
Next, we introduce the constructive semantics and the semantic interface for LiP.
For the construction of the semantics,
	we again make use of a closure operator, but this time on sets of messages to be used as interactive proofs.
In Section~\ref{section:EpistemicExplication},
	we present the promised 
		epistemic explication and 
in Section~\ref{section:OracleComputationalExplication} 
	the promised oracle-computational explication of our interactive proofs.
In Section~\ref{section:ImportantProperties}, 
	we demonstrate  
		the adequacy of our proof system and the decidability of its set of generated theorems (the logic), and
		present our notion of proof equality for LiP.
In Section~\ref{section:TiCLfromLiP},
	we present our resulting interactive structures iS4, iIL, iCL and TiCL, 
			as well as their interconnections.
Finally,
	we relate LiP to LP-like systems in Section~\ref{section:RelationToLP}.

\section{Basic Logic of interactive Proofs (LiP)}\label{section:LiP}
\subsection{Syntactically}\label{section:Syntactically}
The basic Logic of interactive Proofs (LiP) provides a modal \emph{formula language} over 
	a generic message or proof \emph{term language}.
The formula language offers 
	the classical propositional connectives, 
	a relational symbol `$\knows{}{}$' for 
		constructing atomic propositions about individual (term) knowledge, and
	a parameterised unary modal constructor `$\proves{}{}{}{}$' for 
		propositions about proofs.
The term language offers constructors for 
	\emph{pairing} and (possibly cryptographically implemented) \emph{signing}.
(Signature creation and verification is in polynomial time \cite{DigitalSignatures}.)
\begin{definition}[The language of LiP]\label{definition:LiPLanguage}
	Let
	\begin{itemize}
		\item $\agents\neq\emptyset$ designate a non-empty finite set of 
				\emph{agent names} $a$, $b$, $c$, \etc such that $\CM\in\agents$,
					where $\CM$ designates the communication medium;
		\item $\community\subseteq\agents$ denote 
				(finite and not necessarily disjoint) communities (sets) of 
					agents $a\in\agents$ (referred to by their name);
		\item $\messages \ni M\ \bnfeq\ 
				a \bnfor
				\Kcomb{a} \bnfor 
				\Scomb{a} \bnfor 
				\pair{M}{M} \bnfor 
				\sign{M}{a} \bnfor
				B$
			
			designate our language of \emph{message terms} $M$ \emph{over} $\agents$ with 
				constants $a\in\agents$, $\Kcomb{a}$, and $\Scomb{a}$, 
				message pairs $\pair{M}{M}$, 
				signed messages $\sign{M}{a}$, and
				application-specific data $B$ (left blank here);
			note that 
				messages must be grammatically well-formed, which yields an induction principle, and that 
				the meta-variable $B$ just signals the possibility of an extended term language $\messages$;
		\item $\mathcal{P}$ designate a set of atomic propositions $P$, 
				first, constrained so that for all $a\in\agents$ and $M\in\messages$, 
						$(\knows{a}{M})\in\mathcal{P}$\quad(for \emph{individual} knowledge ``$a$ knows $M$'') 
					are propositional constants (not substitutable as such), and 
				second, containing countably infinitely many propositional variables (substitutable as such);
		\item $\pFormulas\ni\phi \bnfeq P \bnfor 
				\neg\phi \bnfor 
				\phi\land\phi \bnfor 
				\proves{M}{\phi}{a}{\community}$ 
				
				designate our language of \emph{logical formulas} $\phi$, where
					$\proves{M}{\phi}{a}{\community}$ means that
						``$M$ is a $\community\cup\set{a}$-\emph{reviewable proof} of $\phi$'' in the sense that 
						``$M$ can prove $\phi$ to $a$ (\eg a designated verifying judge) and this fact  
							is commonly known in the (pointed) community $\community\cup\set{a}$ 
								(\eg for $\community$ being a jury).''
	\end{itemize}
\end{definition}
\noindent
The purpose of our term constants $\Kcomb{a}$ and $\Scomb{a}$ is to act as 
	basic combinators in concert with term pairing (explicit application) in the sense of 
		Combinatory Logic (CL) \cite{LambdaCalculusAndCombinators}.
Thus,
	our message language $\messages$ contains a copy of 
		the language of closed (variable-free) CL-terms (combinators) per agent $a$ 
			(\cf Section~\ref{section:TiCL}).
Of course, we could also conceive of \emph{term forms,}\label{page:TermForms} that is,  
	terms containing \emph{free variables} $x,x',x''\in\mathcal{X}$ with 
		$\mathcal{X}$ a countably infinite set, and 
			thus have a copy of the full language of CL per agent as sublanguages of $\messages$.
However, in order to keep the introduction of our logic 
		as simple as possible and 
		as complicated as necessary, 
	we do not 
		introduce term forms here.
For a technical discussion of our set of atomic propositions, 
	see Appendix~\ref{appendix:AtomicPropositions}.
Now note the following macro-definitions: 
	$\true \defeq \knows{a}{a}$, 
	$\false \defeq \neg \true$, 
	$\phi \lor \phi' \defeq \neg (\neg \phi \land \neg \phi')$,
	$\phi \limp \phi' \defeq \neg \phi \lor \phi'$, 
	$\phi \lequiv \phi' \defeq (\phi \limp \phi') \land (\phi' \limp \phi)$, and, more interestingly
		those in Table~\ref{table:ImportantMacroDefinitions}.\footnote{%
			The problem of defining interactive refutations was suggested to me by Rajeev Gor\'{e}.} 
	\begin{table}[t!]
	\centering
	\caption{Some macro-definable proof concepts}
	\smallskip
	$\boxed{\begin{array}{@{}r@{\ \ }c@{\ \ }l@{}}
		\refutes{M}{\phi}{a}{\community} &\defeq& 
			\begin{array}[t]{@{}l@{}}
				\proves{M}{\neg\phi}{a}{\community}\\
				\text{($M$ is a $\community\cup\set{a}$-reviewable \emph{refutation} of $\phi$ to $a$)}
			\end{array}\\
		\proofdiamond{M}{\phi}{a}{\community} &\defeq&
			\begin{array}[t]{@{}l@{}}
				\neg(\refutes{M}{\phi}{a}{\community})\\
				\text{($M$ is a $\community\cup\set{a}$-reviewable \emph{proof diamond} of $\phi$ to $a$)}
			\end{array}\\				
		\decides{M}{\phi}{a}{\community} &\defeq&
			\begin{array}[t]{@{}l@{}}
				(\proves{M}{\phi}{a}{\community})\lor(\refutes{M}{\phi}{a}{\community})\\
				\text{($M$ is a $\community\cup\set{a}$-reviewable \emph{decider} of $\phi$ to $a$, see \cite{KramerIMLA2013})}
			\end{array}\\
		\undecides{M}{\phi}{a}{\community} &\defeq&
			\begin{array}[t]{@{}l@{}}
				\neg(\decides{M}{\phi}{a}{\community})\\
				\text{($M$ is a $\community\cup\set{a}$-reviewable \emph{non-decider} of $\phi$ to $a$)}
			\end{array}
	\end{array}}$
	\label{table:ImportantMacroDefinitions}
	\end{table}
	Variations on our notions of interactive proof can also be macro-defined, \eg  
		with respect to \emph{reviewer communities} (by 
			conjunction with respect to their members and based on a policy of 
				either one [dis]proof for \emph{all} members 
				or one [dis]proof for \emph{each} member) and
		with respect to \emph{exclusive communities} (members only).

LiP is defined by means of the following axiom and deduction-rule schemas, where 
	grey-shading highlights special, mostly modal interest.
\begin{definition}[The axioms and deduction rules of LiP]\label{definition:AxiomsRules}
	Let
	\begin{itemize}
		\item $\Gamma_{0}$ designate an adequate set of axioms for classical propositional logic
		\item $\Gamma_{1} \defeq \Gamma_{0} \cup \{$
			\begin{itemize}
				\item $\knows{a}{a}$\quad
						(knowledge of one's own name string)
				\item $\knows{a}{M}\limp\knows{a}{\sign{M}{a}}$\quad
						(\emph{personal} signature \emph{synthesis})
				\item $\knows{a}{\sign{M}{b}}\limp\knows{a}{\pair{M}{b}}$\quad
						(\emph{universal} signature \emph{analysis})
				\item $(\knows{a}{M}\land\knows{a}{M'})\lequiv\knows{a}{\pair{M}{M'}}$\quad
						([un]pairing)
				\item $(\proves{M}{(\phi\limp\phi')}{a}{\community})\limp
							((\proves{M}{\phi}{a}{\community})\limp
								\proves{M}{\phi'}{a}{\community})$\quad
						(Kripke's law, K)
				\item \colorbox[gray]{0.75}{%
						$(\proves{M}{\phi}{a}{\community})\limp(\knows{a}{M}\limp\phi)$}\quad
						(epistemic truthfulness)
				\item $(\proves{M}{\phi}{a}{\community})\limp
							\neg(\proves{M}{\neg\phi}{a}{\community})$\quad
						(proof consistency)
				\item \colorbox[gray]{0.75}{%
						$(\proves{M}{\phi}{a}{\community})\limp
							\bigwedge_{b\in\community\cup\set{a}}
								\proves{\sign{M}{a}}{(\knows{a}{M}\land\proves{M}{\phi}{a}{\community})}{b}{\community\cup\set{a}}$}\quad
						(peer review)
				\item \colorbox[gray]{0.75}{%
						$(\proves{M}{\phi}{a}{\community\cup\community'})\limp
							\proves{M}{\phi}{a}{\community}$}\quad
						(group decomposition) \}
			\end{itemize}
			designate a set of \emph{axiom schemas.}
	\end{itemize}
	Then, 
		$\colorbox[gray]{0.75}{%
			$\LiP\defeq\Clo{}{}(\emptyset)$}\defeq
				\bigcup_{n\in\mathbb{N}}\Clo{}{n}(\emptyset)$, where for all $\Gamma\subseteq\pFormulas$:
		\begin{eqnarray*}
			\Clo{}{0}(\Gamma) &\defeq& \Gamma_{1}\cup\Gamma\\
			\Clo{}{n+1}(\Gamma) &\defeq& 
				\begin{array}[t]{@{}l@{}}
					\Clo{}{n}(\Gamma)\ \cup\\
					\setst{\phi'}{\set{\phi,\phi\limp\phi'}\subseteq\Clo{}{n}(\Gamma)}\cup
						\quad\text{(\emph{modus ponens,} MP)}\\
					\setst{\proves{M}{\phi}{a}{\community}}{\phi\in\Clo{}{n}(\Gamma)}\cup
						\quad\text{(necessitation, N)}\\
					\negthickspace\colorbox[gray]{0.75}{%
						$\setst{(\proves{M'}{\phi}{a}{\community})\limp\proves{M}{\phi}{a}{\community}}
								{(\knows{a}{M}\limp\knows{a}{M'})\in\Clo{}{n}(\Gamma)}$}\\
									\quad\text{(epistemic antitonicity)}.\label{page:EpistemicAntitonicity}
						\end{array}
				\end{eqnarray*}
		We call $\LiP$ the \emph{base theory,} and
		$\Clo{}{}(\Gamma)$ an \emph{LiP-theory} for any $\Gamma\subseteq\pFormulas$.
\end{definition}
\noindent
This article is about the base theory (the logic), as suggested by the article title.
Notice the 
	logical order of LiP, which is,
		due to propositions about (proofs of) propositions, \emph{higher-order propositional}.
Further, observe that 
	we assume
		the existence of a dependable mechanism for signing messages, which we 
		model with the above synthesis and analysis axioms. 
In \emph{trusted} multi-agent distributed systems, signatures are unforg\emph{ed,} and thus
	such a mechanism is trivially given by 
		the inclusion of the sender's name in the sent message, or by
		the sender's sensorial impression on the receiver when communication is immediate.
In \emph{dis}trusted multi-agent distributed systems (\eg the open Internet, with a communication medium), 
	a practically unforge\emph{able} signature mechanism can be implemented with  
		classical \emph{certificate-based} or, more directly, with 
		\emph{identity-based} public-key cryptography \cite{DigitalSignatures}.
We also assume the existence of a pairing mechanism modelling finite sets via individual knowledge.
Such a mechanism is required by the important application of
	communication (not only cryptographic) protocols \cite[Chapter~3]{SecurityEngineering}, in which
		concatenation of high-level data packets is associative, commutative, and idempotent for
			an individual knower (\cf Corollary~\ref{corollary:ProofEquality}).
What is more, our peer-review axiom happens to be 
	a possible formalisation of one key property of Miller's foundational proof certificates, 
		which is the one of them being communicable as well as acceptable within 
			a group of peers \cite{FoundationalProofCertificates} 
				(see our Section~\ref{section:EpistemicExplication}).
	
As examples of application-specific data $B$ we conceive of:
		\begin{itemize}
			\item \emph{atomic data} other than the present term constants such as 
				random numbers (used in cryptographic communication),
				quoted formulas $\quoted{\phi}$ (\eg the G\"odel-number of $\phi$ in some G\"odel-numbering scheme)\footnote{%
					Quotation is a form of type down-casting in the sense that 
						data viewed as 
							compound at a certain logical level (here, at the formula-language level) is viewed as
							atomic at a lower level (here, at the term-language level), and
					thus is a form of encoding 
						meta-data (here, statements about messages) in 
						object data (here, messages).}, and others;
			\item \emph{compound data} such as 
				\begin{itemize}
					\item hashed\footnote{Cryptographic hash functions are one-way functions with
											certain cryptographically interesting properties
												such as collision and preimage resistance.} data $\hash{M}$, 
						for $M\in\messages$ and 
						with axiom $\knows{a}{M}\limp\knows{a}{\hash{M}}$
					\item encrypted data $\encrypt{M}{M'}$, 
						for 
							plaintext data $M\in\messages$ and 
							data used as a symmetric encryption key $M'\in\messages$, and
						with axioms 
						\begin{itemize}
							\item $\knows{a}{\pair{M}{M'}}\limp\knows{a}{\encrypt{M}{M'}}$\quad(encryption)
							\item $\knows{a}{\pair{\encrypt{M}{M'}}{M'}}\limp\knows{a}{M}$\quad(decryption)
						\end{itemize}
					This is the so-called Dolev-Yao conception of cryptography \cite{Dolev-Yao},\label{page:DolevYao} which 
						we could easily cast as the following LiP-theory\footnote{\label{footnote:CryptographyConceptions}%
						The integration of other conceptions such as
							the classical information-theoretic \cite{ShannonCryptography} and 
							the modern complexity-theoretic \cite{FoundationsCryptographyOne,FoundationsCryptographyTwo}
								will be presented in future work.}  
							$$\boxed{\LiP_{\text{DY}}\defeq\Clo{}{}(\set{\begin{array}[t]{@{}l@{}}
									\knows{a}{M}\limp\knows{a}{\hash{M}},\\ 
									\knows{a}{\pair{M}{M'}}\limp\knows{a}{\encrypt{M}{M'}},\\ 
									\knows{a}{\pair{\encrypt{M}{M'}}{M'}}\limp\knows{a}{M}}).
								\end{array}}$$
				\end{itemize}
		\end{itemize}

Finally,
	we could close individual knowledge under 
		certain \emph{equational theories} by adding 
			a relational symbol `$=$' for equality together with its standard axioms plus 
			the specific axiom schema
			$$(\knows{a}{M}\land M=M')\limp\knows{a}{M'}\quad\text{(epistemic equational closure)}.$$ 
An example of such a notion of equality is the one of 
	the equational theory of interactive combinators, which can be defined by 
		the following three additional axiom schemas 
			(adapting from classical, unsubscripted combinators \cite{LambdaCalculusAndCombinators}):
			\begin{itemize}
				\item $\pair{\pair{\Kcomb{a}}{M}}{M'}=M$
				\item $\pair{\pair{\pair{\Scomb{a}}{M}}{M'}}{M''}=\pair{\pair{M}{M''}}{\pair{M'}{M''}}$
				\item $\neg(\Kcomb{a}=\Scomb{b})$
			\end{itemize}
However, the epistemic equational closure resulting from this theory would be too strong:
	closing individual knowledge under $\pair{\pair{\Kcomb{a}}{M}}{M'}=M$ would give agents $a$ 
		arbitrary term-guessing power (of arbitrary terms $M'$)!
On the other hand, closing individual knowledge selectively just under 
	$\pair{\pair{\pair{\Scomb{a}}{M}}{M'}}{M''}=\pair{\pair{M}{M''}}{\pair{M'}{M''}}$
is not needed, since it already is (\cf Corollary~\ref{corollary:ScombProp}).

Some other candidate axioms the reader might want to consider are: 
	$(\Kcomb{a}=\Kcomb{b})\limp a=b$, 
	$(\Scomb{a}=\Scomb{b})\limp a=b$, 
	$(\knows{a}{\Kcomb{b}}\lor\knows{a}{\Scomb{b}})\limp\knows{a}{b}$, and 
	$\knows{a}{b}\limp\knows{a}{\pair{\Kcomb{b}}{\Scomb{b}}}$.
If added to LiP, 
	the last axiom would imply  
		the individual knowledge of all non-interactive programs and thus such algorithms (choose $b$ to be $a$)!

Note that in the sequel, ``:iff'' abbreviates ``by definition, if and only if''.
Logicians may want to skip the following proposition.
\begin{proposition}[Hilbert-style proof system]\label{proposition:Hilbert}
	Let 
		\begin{eqnarray*}
			\Phi\LiPded\phi &\defiff& \text{if $\Phi\subseteq\LiP$ then $\phi\in\LiP$}\\ 
			\phi\LiPdedBis\phi' &\defiff& \text{$\set{\phi}\LiPded\phi'$ and $\set{\phi'}\LiPded\phi$}\\
			\LiPded\phi &\defiff& \emptyset\LiPded\phi.
		\end{eqnarray*}
		In other words, ${\LiPded}\subseteq\powerset{\pFormulas}\times\pFormulas$ is a \emph{system of closure conditions} in the sense of 
				\cite[Definition~3.7.4]{PracticalFoundationsOfMathematics}.
		For example:
			\begin{enumerate}
				\item for all axioms $\phi\in\Gamma_{1}$, $\LiPded\phi$
				\item for \emph{modus ponens}, $\set{\phi,\phi\limp\phi'}\LiPded\phi'$
				\item for necessitation, $\set{\phi}\LiPded\proves{M}{\phi}{a}{\community}$
				\item for epistemic antitonicity, 
					$\set{\knows{a}{M}\limp\knows{a}{M'}}\LiPded
						(\proves{M'}{\phi}{a}{\community})\limp\proves{M}{\phi}{a}{\community}$.
			\end{enumerate}
		(In the space-saving, horizontal Hilbert-notation ``\,$\Phi\LiPded\phi$'', 
			$\Phi$ is not a set of hypotheses but a set of premises, see for example   
				\emph{modus ponens,} necessitation, and epistemic antitonicity.\footnote{%
					So for example \emph{modus ponens} can be presented on one line and even in-line as 
						$\set{\phi,\phi\limp\phi'}\LiPded\phi'$ rather than on two display lines as
							$$\frac{\phi\quad\phi\limp\phi'}{\phi'}.$$})

	Then $\LiPded$ can be viewed as being defined by 
		a $\Clo{}{}$-induced Hilbert-style proof system.
	In fact 
			${\Clo{}{}}:\powerset{\pFormulas}\rightarrow\powerset{\pFormulas}$ is a \emph{standard consequence operator,} that is, 
				a \emph{substitution-invariant compact closure operator:}
					\begin{enumerate}
						\item $\Gamma\subseteq\Clo{}{}(\Gamma)$\quad(extensivity)
						\item if $\Gamma\subseteq\Gamma'$ then $\Clo{}{}(\Gamma)\subseteq\Clo{}{}(\Gamma')$\quad(monotonicity)
						\item $\Clo{}{}(\Clo{}{}(\Gamma))\subseteq\Clo{}{}(\Gamma)$\quad(idempotency)
						\item $\Clo{}{}(\Gamma)=\bigcup_{\Gamma'\in\powersetFinite{\Gamma}}\Clo{}{}(\Gamma')$\quad(compactness)
						\item $\sigma[\Clo{}{}(\Gamma)]\subseteq\Clo{}{}(\sigma[\Gamma])$\quad(substitution invariance), 
					\end{enumerate}
					where $\sigma$ designates an arbitrary propositional $\pFormulas$-substitution.
\end{proposition}
\begin{proof}
	That a Hilbert-style proof system can be viewed as induced by 
		 a compact closure operator is well-known (\eg see \cite{WhatIsALogicalSystem});
	that $\Clo{}{}$ is indeed such an operator can be verified by 
		inspection of the inductive definition of $\Clo{}{}$; and
	substitution invariance follows from our definitional use of axiom \emph{schemas}.\footnote{%
		Alternatively to axiom schemas,
		we could have used 
			axioms together with an
			additional substitution-rule set
				$\setst{\sigma[\phi]}{\phi\in\Clo{}{n}(\Gamma)}$
		in the definiens of $\Clo{}{n+1}(\Gamma)$.}
\end{proof}

We are now going to present some useful (further-used), deducible \emph{structural} laws of LiP, including
	the deducible non-structural rule of epistemic bitonicity, used in the deduction of some of them.
Here, 
	``structural'' means 
	``deduced exclusively from term axioms.''
The laws are enumerated in a (total) order that respects (but cannot reflect) their respective proof prerequisites.
\begin{theorem}[Some useful deducible structural laws]\label{theorem:SomeUsefulDeducibleStructuralLaws}\ 
	\begin{enumerate}
		\item $\LiPded\knows{a}{\pair{M}{M'}}\limp\knows{a}{M}$
				\quad(left projection, \colorbox[gray]{0.75}{1-way $\Kcomb{}$-combinator property})
		\item $\LiPded\knows{a}{\pair{M}{M'}}\limp\knows{a}{M'}$
				\quad(right projection)
		\item $\LiPded\knows{a}{\pair{M}{M}}\lequiv\knows{a}{M}$
				\quad(pairing idempotency)
		\item $\LiPded\knows{a}{\pair{M}{M'}}\lequiv\knows{a}{\pair{M'}{M}}$
				\quad(pairing commutativity)
		\item $\LiPded(\knows{a}{M}\limp\knows{a}{M'})\lequiv(\knows{a}{\pair{M}{M'}}\lequiv\knows{a}{M})$
				
				(neutral pair elements)
		\item $\LiPded\knows{a}{\pair{M}{a}}\lequiv\knows{a}{M}$
				\quad(self-neutral pair element)
		\item $\LiPded\knows{a}{\pair{M}{\pair{M'}{M''}}}\lequiv\knows{a}{\pair{\pair{M}{M'}}{M''}}$
				\quad(pairing associativity)
		\item $\set{\knows{a}{M}\lequiv\knows{a}{M'}}\LiPded
					(\proves{M}{\phi}{a}{\community})\lequiv\proves{M'}{\phi}{a}{\community}$
				\quad(epistemic bitonicity)
		\item $\LiPded(\proves{M}{\phi}{a}{\community})\limp\proves{\pair{M'}{M}}{\phi}{a}{\community}$
				\quad(proof extension, left)
		\item $\LiPded(\proves{M}{\phi}{a}{\community})\limp\proves{\pair{M}{M'}}{\phi}{a}{\community}$
				\quad(proof extension, right)
		\item $\LiPded((\proves{M}{\phi}{a}{\community})\lor\proves{M'}{\phi}{a}{\community})\limp\proves{\pair{M}{M'}}{\phi}{a}{\community}$
				\quad(proof extension)
		\item $\LiPded(\proves{\pair{M}{M}}{\phi}{a}{\community})\lequiv\proves{M}{\phi}{a}{\community}$
				\quad(proof idempotency)
		\item $\LiPded(\proves{\pair{M}{M'}}{\phi}{a}{\community})\lequiv\proves{\pair{M'}{M}}{\phi}{a}{\community}$
				\quad(proof commutativity)
	\item $\set{\knows{a}{M}\limp\knows{a}{M'}}\LiPded(\proves{\pair{M}{M'}}{\phi}{a}{\community})\lequiv\proves{M}{\phi}{a}{\community}$
	
				(neutral proof elements)
		\item $\LiPded(\proves{\pair{M}{a}}{\phi}{a}{\community})\lequiv\proves{M}{\phi}{a}{\community}$
				\quad(self-neutral proof element)
		\item $\LiPded(\proves{\pair{M}{\pair{M'}{M''}}}{\phi}{a}{\community})\lequiv\proves{\pair{\pair{M}{M'}}{M''}}{\phi}{a}{\community}$
				\quad(proof associativity)
		\item $\LiPded(\proves{\sign{M}{a}}{\phi}{a}{\community})\limp\proves{M}{\phi}{a}{\community}$
				\quad(self-signing elimination)
		\item $\LiPded((\proves{M}{\phi}{a}{\community})\lor
						\proves{b}{\phi}{a}{\community})\limp\proves{\sign{M}{b}}{\phi}{a}{\community}$
				\quad(signing introduction)
		\item $\LiPded(\proves{\sign{M}{a}}{\phi}{a}{\community})\lequiv\proves{M}{\phi}{a}{\community}$
				\quad(self-signing idempotency)
		\item When $\agents=\set{a}$ (singleton society) and $\messages\setminus\set{\Kcomb{a},\Scomb{a}}$:
			\begin{enumerate}
				\item $\LiPded\knows{a}{M}$\quad(total knowledge)
				\item $\LiPded\knows{a}{M}\lequiv\knows{a}{M'}$\quad(epistemic indifference)	
				\item $\LiPded(\proves{M}{\phi}{a}{\community})\lequiv\proves{M'}{\phi}{a}{\community}$\quad(proof indifference).	
			\end{enumerate}
	\end{enumerate}
\end{theorem}
\begin{proof} 
	See Appendix~\ref{appendix:StructuralProofs}.
\end{proof}
\noindent
Proof extension and idempotency jointly define \emph{proof redundancy}.
For a discussion of the hypothetical cases where 
	$\agents=\set{a}$ and $\messages\setminus\set{\Kcomb{a},\Scomb{a}}$ (thought experiment), 
		see Appendix~\ref{appendix:SingletonSociety}.
Next, the 1-way $\Kcomb{}$-combinator property and 
	the following simple corollary of Theorem~\ref{theorem:SomeUsefulDeducibleStructuralLaws} jointly establish the fact that
		our agents can be viewed as combinators in the sense of Combinatory Logic (CL) 
			viewed as a (non-equational) theory of term reduction \cite{LambdaCalculusAndCombinators}.
The implicational converse of the $\Kcomb{}$-combinator property, that is,  
	$\knows{a}{M}\limp\knows{a}{\pair{M}{M'}}$, must (and does) fail 
		(\ie $\not\LiPded\knows{a}{M}\limp\knows{a}{\pair{M}{M'}}$ by consistency), since otherwise 
			agents would have arbitrary term-guessing power (as already explained), which
				would invalidate our modelling of
					the non-trivial distribution of information among agents.
\begin{corollary}[$\Scomb{}$-combinator property]\label{corollary:ScombProp}\  
	\begin{enumerate}
		\item $\LiPded\knows{a}{\pair{\pair{M}{M'}}{M''}}\lequiv\knows{a}{\pair{\pair{M}{M''}}{\pair{M'}{M''}}}$
		\item $\LiPded(\proves{\pair{\pair{M}{M'}}{M''}}{\phi}{a}{\community})\lequiv
						\proves{\pair{\pair{M}{M''}}{\pair{M'}{M''}}}{\phi}{a}{\community}$
	\end{enumerate}
\end{corollary}
\begin{proof}
	(1) follows from the idempotency (copy $M''$), commutativity, and associativity of pairing; and
	(2) follows from (1) and epistemic bitonicity.
\end{proof}
Note that thanks to the modular set-up of LiP,
	epistemic antitonicity would equally easily yield the application-specific modal laws for:
		\begin{itemize}
			\item hashing: $(\proves{\hash{M}}{\phi}{a}{\community})\limp\proves{M}{\phi}{a}{\community}$
			\item encryption:
					$(\proves{\encrypt{M}{M'}}{\phi}{a}{\community})\limp\proves{\pair{M}{M'}}{\phi}{a}{\community}$
			\item decryption: 
					$(\proves{M}{\phi}{a}{\community})\limp\proves{\pair{\encrypt{M}{M'}}{M'}}{\phi}{a}{\community}$
		\end{itemize}

We are continuing to present also some useful (further-used), deducible \emph{logical} laws of LiP.
Here, 
	``logical'' means 
	``not structural'' in the previously defined sense.
Also these laws are enumerated in an order that respects their respective proof prerequisites.
Our resulting library of structural and logical laws is, analogously to a programming library, 
	very useful when it comes to actually proving 
		more applied theorems of---rather than meta-theorems about---our logic, and 
	gives technical insights into its inner workings (consequences of axiom choices).
\begin{theorem}[Some useful deducible logical laws]\label{theorem:SomeUsefulDeducibleLogicalLaws}\  
	\begin{enumerate}
		\item $\LiPded(\proves{M}{(\phi\limp\phi')}{a}{\community})\limp
				((\proves{M'}{\phi}{a}{\community})\limp\proves{\pair{M}{M'}}{\phi'}{a}{\community})$\\
				\quad(\colorbox[gray]{0.75}{generalised Kripke-law, GK})
		\item $\set{\phi\limp\phi'}\LiPded(\proves{M}{\phi}{a}{\community})\limp
				\proves{M}{\phi'}{a}{\community}$
				\quad(regularity, R)
		\item $\set{\phi\lequiv\phi'}\LiPded(\proves{M}{\phi}{a}{\community})\lequiv
				\proves{M}{\phi'}{a}{\community}$\quad(R \emph{bis})
		\item $\set{\knows{a}{M}\limp\knows{a}{M'},\phi\limp\phi'}\LiPded
				(\proves{M'}{\phi}{a}{\community})\limp\proves{M}{\phi'}{a}{\community}$\quad(epistemic R, ER)
		\item $\set{\knows{a}{M}\lequiv\knows{a}{M'},\phi\lequiv\phi'}\LiPded
				(\proves{M'}{\phi}{a}{\community})\lequiv\proves{M}{\phi'}{a}{\community}$\quad(ER \emph{bis})
		\item $\LiPded((\proves{M}{\phi}{a}{\community})\land\proves{M'}{\phi'}{a}{\community})\limp
						\proves{\pair{M}{M'}}{(\phi\land\phi')}{a}{\community}$
				\quad(proof conjunctions)
		\item $\LiPded((\proves{M}{\phi}{a}{\community})\land
						\proves{M}{\phi'}{a}{\community})\lequiv\proves{M}{(\phi\land\phi')}{a}{\community}$
				\quad(proof conjunctions \emph{bis})
		\item $\LiPded((\proves{M}{\phi}{a}{\community})\lor\proves{M'}{\phi'}{a}{\community})\limp
						\proves{\pair{M}{M'}}{(\phi\lor\phi')}{a}{\community}$
				\quad(proof disjunctions)
		\item $\LiPded((\proves{M}{\phi}{a}{\community})\lor\proves{M}{\phi'}{a}{\community})\limp
						\proves{M}{(\phi\lor\phi')}{a}{\community}$
				\quad(proof disjunctions \emph{bis})
		\item $\LiPded\proves{M}{\true}{a}{\community}$
				\quad(anything can prove tautological truth)
		\item $\LiPded(\proves{a}{\phi}{a}{\community})\limp\phi$
				\quad(self-truthfulness)
		\item $\phi\LiPdedBis\proves{a}{\phi}{a}{\community}$
				\quad(self-truthfulness \emph{bis})
		\item $\LiPded\neg(\proves{M}{\false}{a}{\community})$
				\quad(nothing can prove falsehood)	
		\item $\LiPded\bigwedge_{b\in\community\cup\set{a}}
				(\proves{\sign{M}{a}}{\knows{a}{M}}{b}{\community\cup\set{a}})$
					\quad(\colorbox[gray]{0.75}{authentic knowledge})
		\item $\LiPded\proves{M}{\knows{a}{M}}{a}{\emptyset}$
				\quad(self-knowledge)
		\item $\knows{a}{M}\limp\phi\LiPdedBis
							\bigwedge_{b\in\community\cup\set{a}}
								(\proves{\sign{M}{a}}{\phi}{b}{\community\cup\set{a}})$
				\quad(\colorbox[gray]{0.75}{authentic epistemic N, AEN})
		\item $\knows{a}{M}\limp\knows{a}{M'}\LiPdedBis
				\bigwedge_{b\in\community\cup\set{a}}(\proves{\sign{M}{a}}{\knows{a}{M'}}{b}{\community\cup\set{a}})$
				\quad(AEN \emph{bis}) 
		\item $\knows{a}{M}\limp\phi\LiPdedBis\proves{M}{\phi}{a}{\emptyset}$
				\quad(self-epistemic N, SEN)
		\item $\knows{a}{M}\limp\knows{a}{M'}\LiPdedBis\proves{M}{\knows{a}{M'}}{a}{\emptyset}$
				\quad(SEN \emph{bis})
		\item $\LiPded\knows{a}{\sign{M}{b}}\limp\knows{b}{M}$\quad(\colorbox[gray]{0.75}{message attribution})
		\item $\LiPded(\proves{M}{\phi}{a}{\community})\limp
				\bigwedge_{b\in\community\cup\set{a}}(\proves{\sign{M}{a}}{\phi}{b}{\community\cup\set{a}})$
				\quad(simple peer review)
		\item $\LiPded(\proves{M}{\phi}{a}{\community\cup\community'})\limp
				((\proves{M}{\phi}{a}{\community})\land\proves{M}{\phi}{a}{\community'})$
				\quad(group decomposition \emph{bis})
		\item $\LiPded(\proves{M}{\phi}{a}{\community\cup\set{a}})\lequiv(\proves{M}{\phi}{a}{\community})$
				\quad(self-neutral group element)
		\item $\LiPded\proves{M}{((\proves{M}{\phi}{a}{\community})\limp\phi)}{a}{\community}$
				\quad(self-proof of truthfulness)
		\item $\LiPded\proves{M}{(\neg(\proves{M}{\false}{a}{\community}))}{a}{\community}$
				\quad(self-proof of proof consistency)
		\item $\LiPded(\proves{M}{\phi}{a}{\community})\limp
				\proves{M}{(\bigwedge_{b\in\community\cup\set{a}}(
					\proves{\sign{M}{a}}{\phi}{b}{\community\cup\set{a}}))}{a}{\community}$\\
				\quad(simple peer review \emph{bis})
		\item $\LiPded(\proves{M}{(\proves{M}{\phi}{a}{\community})}{a}{\community})\lequiv
				\proves{M}{\phi}{a}{\community}$
				\quad(modal idempotency)
		\item When $\agents=\set{a}$ (singleton society) and $\messages\setminus\set{\Kcomb{a},\Scomb{a}}$,
				\begin{enumerate} 
					\item $\LiPded(\proves{M}{\phi}{a}{\community})\limp\phi$\quad(truthfulness)
					\item $\phi\LiPdedBis\proves{M}{\phi}{a}{\community}$
				\quad(truthfulness \emph{bis})
				\end{enumerate} 
\end{enumerate}
\end{theorem}
\begin{proof}
	See Appendix~\ref{appendix:LogicalProofs}.	
\end{proof}
\noindent
Kripke's laws (K and GK) and the law of modal idempotency are discussed in Section~\ref{section:LogicalLaws}.
The key to their validity is that LiP-agents 
	are resource-unbounded (though are unable to guess) and 
	act themselves as proof checkers (no need for LP's `$!$').
Notice that
	regularity and epistemic antitonicity resemble each other in that
		both laws relate an implicational premise with an implicational conclusion about proof modalities, but
			while regularity relates the modality operands monotonically,
				epistemic antitonicity relates the proof parameters antitonically.
Both laws are combined in the law of epistemic regularity.
The following remark flags an important observation.
\begin{remark}[BHK- or \emph{realisability} interpretation of LiP-terms]\label{remark:BHK}
	The laws of proof conjunctions, proof disjunctions, and GK (implication) mean that
		the proof terms of LiP satisfy the construction interpretation of 
			(propositional) Intuitionistic Logic (IL) given by Brower-Heyting-Kolmogorow 
				(the so-called \emph{BHK-} or \emph{realisability} interpretation of IL)   
					with the notable (but semantically inessential \cite[Page~552 and 553]{Fine:TruthMakerSemanticsIL}) 
						difference that unordered (\cf Corollary~\ref{corollary:ProofEquality}) pairing 
							suffices as constructor (no need for choice constructors for disjunction).
								(Also, recall that negation is definable in terms of implication and falsehood.)
\end{remark}
\noindent
Given the classicality of our interactive proof terms, 
	their constructiveness in the sense of BHK-realisability 
		is a remarkable feature of them.
Thus, BHK-realisability does not characterise IL, since 
	not only IL satisfies it.

\begin{fact}[Normality]\label{fact:Normality}
		LiP is a normal modal logic.
\end{fact}
\begin{proof}
	By
		Kripke's law,
		\emph{modus ponens,}
		necessitation, and
		substitution invariance (\cf Proposition~\ref{proposition:Hilbert}).
\end{proof}
\noindent
In contrast, LP is, technically speaking, not a normal modal logic \cite[Section~5]{ModalLogicInMathematics}.

The following theorem asserts an important correspondence between classical propositional logic (PL) and LiP.
\begin{theorem}[Internalisation/Externalisation Property]\label{theorem:EXIN}
	For all formulas $\varphi$ of the language of classical propositional logic PL 
		with atomic propositions $\knows{a}{M}$:
	\begin{enumerate}
		\item $\set{\knows{a}{M}}\vdash_{\mathrm{PL}}\varphi$ if and only if 
				$\LiPded\bigwedge_{b\in\community\cup\set{a}}(\proves{\sign{M}{a}}{\varphi}{b}{\community\cup\set{a}})$\,;
		\item $\set{\knows{a}{M}}\vdash_{\mathrm{PL}}\varphi$ if and only if 
				$\LiPded\proves{M}{\varphi}{a}{\emptyset}$.
	\end{enumerate}
\end{theorem}
\begin{proof}
	Let $\varphi$ (and $\varphi'$ and $\varphi''$) designate a formula of the language of PL.
	Then, 
		for (1), apply 
			the Deduction Theorem (DT) of PL 
				($\set{\varphi'}\vdash_{\mathrm{PL}}\varphi$ iff $\vdash_{\mathrm{PL}}\varphi'\limp\varphi$) and 
			authentic epistemic necessitation (AEN) of LiP,
				both in both ways, and
			the definitional fact that LiP contains PL 
				($\vdash_{\mathrm{PL}}\varphi''$ iff $\LiPded\varphi''$); and
		for (2), proceed as with (1) but  
			apply self-epistemic necessitation (SEN) instead of AEN.
\end{proof}
The following corollary may be viewed as a corroboration of Remark~\ref{remark:BHK}.
\begin{corollary}[Internalisation Property of IL into LiP]\label{corollary:ILintoLiP}
	For all $\varphi$ of the language of intuitionistic propositional logic IL
		with atomic propositions $\knows{a}{M}$:\footnote{%
			The languages of IL and PL are the same, but IL and PL (the logics) are of course not.}
	\begin{enumerate}
		\item if $\set{\knows{a}{M}}\vdash_{\mathrm{IL}}\varphi$ then 
					$\LiPded\bigwedge_{b\in\community\cup\set{a}}
						(\proves{\sign{M}{a}}{\varphi}{b}{\community})$\,;
		\item if $\set{\knows{a}{M}}\vdash_{\mathrm{IL}}\varphi$ then
					$\LiPded\proves{M}{\varphi}{a}{\emptyset}$.
	\end{enumerate}
\end{corollary}
\begin{proof}
	Proceed as for the proof of Theorem~\ref{theorem:EXIN}, but   
		apply AEN only from left to right, and 
		additionally apply the fact that $\mathrm{IL}\subsetneq\mathrm{PL}$.
\end{proof}
\noindent
This corollary can be partially (for the communication medium $\CM$ only) strengthened in the sense of 
	Theorem~\ref{theorem:EXIN}.2 in the (intuitionistic) Logic of intuitionistic interactive Proofs (LIiP)
		\cite{LIiP:ACMTOCL}, so that for all $\varphi$ of the language of IL with atomic propositions $\knows{\CM}{M}$,
			$$\text{$\set{\knows{\CM}{M}}\vdash_{\mathrm{IL}}\varphi$ if and only if 
				$\vdash_{\mathrm{LIiP}} M\pm_{\CM}\varphi$.}$$

\subsection{Semantically}\label{section:Semantically}
\subsubsection{Concretely}\label{section:Concretely}
We now present 
	the concretely constructed semantics as well as 
	the standard abstract semantic interface for LiP, and 
prove the axiomatic adequacy of the proof system with respect to this interface.
The core ingredient of the concrete semantics of LiP are so-called \emph{input histories,} which
	were introduced in \cite{KramerIMLA2013}. 
Input histories 
	are finite words of input events and 
	serve as concrete states $s\in\states$ in 
		the state space $\states$, on which 
			the concrete and abstract accessibility relation 
				${\pAccess{M}{a}{\community}}\subseteq\states\times\states$ and
				${\access{M}{a}{\community}}\subseteq\states\times\states$ for LiP is defined, respectively.
\begin{definition}[Semantic ingredients]\label{definition:SemanticIngredients}
For the knowledge-constructive model-theoretic study of LiP,
	let 
\begin{itemize}
	\item $\states\ni s\bnfeq\mathtt{0}\bnfor \gscc{a}{M}(s)$ 
		designate the concrete state space $\states$ of
			\emph{input histories} $s$ constrained so that 
				only $a$ can generate $a$'s signature (\textbf{\emph{signature unforgeability}}), that is, 
				for all $s\in\states$, $a,b\in\agents$, and $\sign{M}{a}\in\clo{b}{s}(\emptyset)$,
					there are $s',s''\in\states$ such that 
						$s=s'\star s''$ and 
						for all $c\in\agents$, if $\sign{M}{a}\in\clo{c}{s'}(\emptyset)$ then $c=a$, where 
			$\mathtt{0}$ designates the empty input history  
				(\ie a zero data point, \eg an initial state) and 
			$\gscc{a}{M}$ reads as ``agent $a$ receives message $M$''  
				(from some other, oracle agent or from herself), 
			$\star:(\states\times\states)\to\states$ monoidal concatenation on $\states$
				(with neutral element $\mathtt{0}$), and
			$\clo{}{}$ is defined below; 
	\item $\pi_{a}:\states\rightarrow\states$ designate (local) \emph{state projection on $a$'s view} such that 
			\begin{align*}
				\pi_{a}(\mathtt{0}) &\defeq \mathtt{0}\\
				\pi_{a}(\gscc{b}{M}(s)) &\defeq 
					\begin{cases}
						\gscc{b}{M}(\pi_{a}(s)) & \text{if $a\in\set{b,\CM}$, and}\\
						\pi_{a}(s) & \text{otherwise;}
					\end{cases}
			\end{align*}
			(The communication medium $\CM$ sees any agent's $b$ [including its own] input events, that is, 
				$\CM$ has a global view on the current global state $s$.)
	\item $\msgs{}:\states\rightarrow\powerset{\messages}$ designate \emph{raw-data extraction} such that 
			\begin{align*}
				\msgs{}(\mathtt{0}) &\defeq \emptyset\\
				\msgs{}(\gscc{a}{M}(s)) &\defeq \msgs{}(s)\cup\set{M}\,;
			\end{align*}
	\item $\msgs{a}\defeq\msgs{}\circ\pi_{a}$ designate (local) \emph{raw-data extraction by $a$}
			(the set $\msgs{a}(s)$ can be viewed as $a$'s \emph{data base} in $s$\label{page:RawData});
		\item $\clo{a}{s}:\powerset{\messages}\rightarrow\powerset{\messages}$ designate a \emph{data-mining operator} such that \label{page:DataMining}
			$\clo{a}{s}(\data)\defeq\clo{a}{}(\msgs{a}(s)\cup\data)\defeq\bigcup_{n\in\mathbb{N}}\clo{a}{n}(\msgs{a}(s)\cup\data)$, where for all $\data\subseteq\messages$:
				\begin{eqnarray*}
					\clo{a}{0}(\data) &\defeq& \set{a}\cup\data\\
					\clo{a}{n+1}(\data) &\defeq& 
						\begin{array}[t]{@{}l@{}}
							\clo{a}{n}(\data)\ \cup\\
							\setst{\pair{M}{M'}}{\set{M,M'}\subseteq\clo{a}{n}(\data)}\cup\quad\text{(pairing)}\\
							\setst{M, M'}{\pair{M}{M'}\in\clo{a}{n}(\data)}\cup\quad\text{(unpairing)}\\
							\setst{\sign{M}{a}}{M\in\clo{a}{n}(\data)}\cup\quad\text{(\emph{personal} signature \emph{synthesis})}\\
							\setst{\pair{M}{b}}{\sign{M}{b}\in\clo{a}{n}(\data)}\quad\text{(\emph{universal} signature \emph{analysis})}
						\end{array}
				\end{eqnarray*}
				($\clo{a}{s}(\emptyset)$ can be viewed as $a$'s \emph{individual-knowledge base} in $s$. 
					For application-specific terms such as encryption, 
					we would have to add here the closure conditions corresponding to their characteristic term axioms.)
	\item ${\preorder{a}}\subseteq\states\times\states$ designate a \emph{data preorder} on states such that
		for all $s,s'\in\states$,
			$s\preorder{a}s'$ :iff $\clo{a}{s}(\emptyset)\subseteq\clo{a}{s'}(\emptyset)$\,;
			
			(The reader is invited to consider the effects of encryption on closure here.)
	\item ${\preorder{\community}}\defeq(\bigcup_{a\in\community}{\preorder{a}})^{*}$, where 
			`$^{*}$' designates the Kleene (\ie the reflexive transitive) closure operation on binary relations;
	\item ${\indist{a}{}{}}\defeq{\preorder{a}\cap(\preorder{a})^{-1}}$ designate an equivalence relation of 
		\emph{state indistinguishability}, where
			 `$^{-1}$' designates the converse operation on binary relations; 
	\item ${\pAccess{M}{a}{\community}}\subseteq\states\times\states$ designate 
		our \emph{concretely constructed accessibility relation}---short, \emph{concrete accessibility}---for 
			the proof modality so that for all $s,s'\in\states$, 
			\label{page:ProofAccessibility}
			\begin{eqnarray}
				s\pAccess{M}{a}{\community}s' &\defiff&
				s'\in\hspace{-5ex}\bigcup_{\scriptsize
					\mbox{
					$\begin{array}{@{}c@{}}
						\text{$s\preorder{\community\cup\set{a}}\check{s}$ and}\\[0.5\jot] 
						M\in\clo{a}{\check{s}}(\emptyset) 
					\end{array}$}
					}\hspace{-5ex}[\check{s}]_{\indist{a}{}{}}\\\notag
				&\text{(iff}& 
					\text{there is $\check{s}\in\states$ \st 
							$s\preorder{\community\cup\set{a}}\check{s}$ and
							$M\in\clo{a}{\check{s}}(\emptyset)$ and
							$\indist{a}{\check{s}}{s'}$)}
			\end{eqnarray}
					
					(See Section~\ref{section:EpistemicExplication} for an extensive
							explication of this elementary construction.)
\end{itemize}
\end{definition}

\begin{proposition}[Data closure]
	$\clo{a}{}:\powerset{\messages}\rightarrow\powerset{\messages}$ is a \emph{compact closure operator:}
		\begin{enumerate}
			\item $\data\subseteq\clo{a}{}(\data)$\quad(extensivity)
			\item if $\data\subseteq\data'$ then $\clo{a}{}(\data)\subseteq\clo{a}{}(\data')$\quad(monotonicity)
			\item $\clo{a}{}(\clo{a}{}(\data))\subseteq\clo{a}{}(\data)$\quad(idempotency)
			\item $\clo{a}{}(\data)=\bigcup_{\data'\in\powersetFinite{\data}}\clo{a}{}(\data')$\quad(compactness)
		\end{enumerate}
\end{proposition}
\begin{proof}
	By inspection of the inductive definition of $\clo{a}{}$.
\end{proof}
\noindent
The operator $\clo{a}{}$ induces 
	a relation ${\derives{a}{}{}}\subseteq\powerset{\messages}\times\messages$ of \emph{data derivation} such that
			$$\text{$\derives{a}{M}{\data}$ :iff $M\in\clo{a}{}(\data)$.}$$
Hence, an agent $a$ can be viewed as a \emph{data miner} who 
			mines the data $\data$ by means of the \cite[\emph{association}]{DataMining} \emph{rules} for pairing and signing (and possibly other, application-specific constructors) that
				define the closure operator $\clo{a}{}$.

\begin{proposition}[Data derivation]\label{proposition:PropertiesOfDerivability}\  
	\begin{description}
		\item[Cut] 
					If $\derives{a}{M}{\data}$ and $\derives{a}{M'}{\set{M}}$ 
					then $\derives{a}{M'}{\data}$.
		\item[Compactness] 
					If $\derives{a}{M}{\data}$ 
					then there is a finite $\data'\subseteq\data$ such that 
					$\derives{a}{M}{\data'}$. 
		\item[Complexity] For all 
			finite $\data\subseteq\messages$,
			``\;$\derives{a}{M}{\data}$'' is decidable in deterministic polynomial time in 
				the size of $\data$ and $M$.
		\item[Connection to Scott information systems] Let for all $a\in\agents$, $s\in\states$, and $\data\subseteq\messages$,
				$$\mathcal{C}_{a}^{s}(\data) \defeq \setst{\data'\subseteq\data}{\clo{a}{s}(\data')=\data'}.$$
			Further, let 
				$$\mathit{Con}_{a}^{s} \defeq \bigcup_{\data\in\mathcal{C}_{a}^{s}(\messages)}\powersetFinite{\data}.$$
			Then, $$\langle\messages,\mathit{Con}_{a}^{s},\derives{a}{}{}\rangle$$ is
			a \emph{Scott information system,} 
				that is,  for all $M\in\messages$, $\data\in\mathit{Con}_{a}^{s}$, and $\data'\subseteq\messages$:
								\begin{enumerate}
									\item\label{corollary:closure:one:three:one} $\set{M}\in\mathit{Con}_{a}^{s}$
									\item\label{corollary:closure:one:three:two} if $M\in\data$ then $\derives{a}{M}{\data}$
									\item\label{corollary:closure:one:three:three} if $\data'\subseteq\data$ then $\data'\in\mathit{Con}_{a}^{s}$
									\item\label{corollary:closure:one:three:four} if $\derives{a}{M}{\data}$ then $\data\cup\set{M}\in\mathit{Con}_{a}^{s}$
									\item\label{corollary:closure:one:three:five} if $\data'\in\mathit{Con}_{a}^{s}$ and 
												$\derives{a}{\data'}{\data}$ 
												and $\derives{a}{M}{\data'}$
												then $\derives{a}{M}{\data}$,
												where $\derives{a}{\data'}{\data}$ :iff for all $M'\in\data'$, $\derives{a}{M'}{\data}$.
								\end{enumerate}
						(Message terms are \emph{information tokens} in the sense of Dana Scott 
							\cite[Chapter~9]{DaveyPriestley}.)
	\end{description}
\end{proposition}
\begin{proof}
	The cut and the compactness property follow by 
		inspection of the defining cases of $\clo{a}{}$.
	The complexity follows from the complexity of message derivation for 
		even more complex message languages 
			(\eg including encryption and other constructors \cite{TGJ} and \cite{BaskarJamSuresh2010}).
	Regarding the connection to Scott information systems:
		Property~1 follows from the fact that $\set{M}\in\powersetFinite{\messages}$ and $\messages\in\mathcal{C}_{a}^{s}(\messages)$,
		Property~2 from the definition of $\derives{a}{}{}$,
		Property~3 from the powerset construction,
		Property~4 from the definition of $\derives{a}{}{}$, and
		Property~5 jointly from 
				the finiteness of $\data'$ (which can be transformed into a message pair [of pairs]) and 
					the cut property of $\derives{a}{}{}$.
\end{proof}

\begin{proposition}[Concrete accessibility]\label{proposition:ConcreteAccessibility}\ 
\begin{enumerate}
	\item	If for all $b\in\community\cup\set{a}$, $s\pAccess{\sign{M}{a}}{b}{\community\cup\set{a}}s'$ 
	
			then $M\in\clo{a}{s'}(\emptyset)$\quad(signature property).
	\item If $M\in\clo{a}{s}(\emptyset)$
			then $s\pAccess{M}{a}{\community}s$\quad(conditional reflexivity).
	\item There is $s'\in\states$ such that 
			$s\pAccess{M}{a}{\community}s'$\quad(seriality).
	\item For all $b\in\community\cup\set{a}$, 
			$({\pAccess{\sign{M}{a}}{b}{\community\cup\set{a}}}\circ{\pAccess{M}{a}{\community}})\subseteq{\pAccess{M}{a}{\community}}$\quad(communal transitivity).
	\item If $\community\subseteq\community'$ 
			then ${\pAccess{M}{a}{\community}}\subseteq{\pAccess{M}{a}{\community'}}$\quad(communal monotonicity).
	\item If $M\leq_{a}M'$ then ${\pAccess{M}{a}{\community}}\subseteq{\pAccess{M'}{a}{\community}}$\quad(epistemic persistency).
\end{enumerate}
\end{proposition}
\begin{proof}
	For (1), suppose that for all $b\in\community\cup\set{a}$, 
		$s\pAccess{\sign{M}{a}}{b}{\community\cup\set{a}}s'$.
	Hence for all $b\in\community\cup\set{a}$,
		$\sign{M}{a}\in\clo{b}{s'}(\emptyset)$, by definition of $\pAccess{\sign{M}{a}}{b}{\community\cup\set{a}}$.
	Hence,
		$\sign{M}{a}\in\clo{a}{s'}(\emptyset)$ by signature unforgeability, and then
		$M\in\clo{a}{s'}(\emptyset)$ by signature analysis and unpairing.
	For (2), suppose that $\underline{M\in\clo{a}{s}(\emptyset)}$.
	Further, $\underline{s\preorder{\community\cup\set{a}}s}$ and $\underline{\indist{a}{s}{s}}$, by reflexivity.
	Hence $s\pAccess{M}{a}{\community}s$.
	For (3), let $s\in\states$.
	Then,
		$\underline{s\preorder{\community\cup\set{a}}\gscc{a}{M}(s)}$ and
		$\underline{M\in\clo{a}{\gscc{a}{M}(s)}(\emptyset)}$ and
		$\underline{\indist{a}{\gscc{a}{M}(s)}{\gscc{a}{M}(s)}}$.
	Hence $s\pAccess{M}{a}{\community}\gscc{a}{M}(s)$.
	Thus there is $s'\in\states$ such that $s\pAccess{M}{a}{\community}s'$.
	For (4), 
		let 
			$s,s',s''\in\states$ and 
			$b\in\community\cup\set{a}$ (thus $\community\cup\set{a}=\community\cup\set{a}\cup\set{b}$) and 
		suppose that 
			$s\pAccess{\sign{M}{a}}{b}{\community\cup\set{a}}s'$ and $s'\pAccess{M}{a}{\community}s''$.
	That is,
		there is $\check{s}\in\states$ such that
			$s\preorder{\community\cup\set{a}\cup\set{b}}\check{s}$ 
				(thus $s\preorder{\community\cup\set{a}}\check{s}$) and
			$\sign{M}{a}\in\clo{b}{\check{s}}(\emptyset)$ and
			$\indist{b}{\check{s}}{s'}$ 
				(thus $\check{s}\preorder{\community\cup\set{a}\cup\set{b}}s'$ and
				 thus $\check{s}\preorder{\community\cup\set{a}}s'$), and
		there is $\check{s}'\in\states$ such that
			$s'\preorder{\community\cup\set{a}}\check{s}'$ and
			$\underline{M\in\clo{a}{\check{s}'}(\emptyset)}$ and
			$\underline{\indist{a}{\check{s}'}{s''}}$.
	Hence	
		$s\preorder{\community\cup\set{a}}s'$ and then
		$\underline{s\preorder{\community\cup\set{a}}\check{s}'}$, both by transitivity, and
	thus $s\pAccess{M}{a}{\community}s''$.
	For (5), suppose that $\community\subseteq\community'$.
	Further, 
		let $s,s'\in\states$ and
		suppose that $s\pAccess{M}{a}{\community}s'$.
	That is,
		there is $\check{s}\in\states$ such that
			$s\preorder{\community\cup\set{a}}\check{s}$ and
			$\underline{M\in\clo{a}{\check{s}}(\emptyset)}$ and
			$\underline{\indist{a}{\check{s}}{s'}}$.
	Hence $\underline{s\preorder{\community'\cup\set{a}}\check{s}}$, and
	thus $s\pAccess{M}{a}{\community'}s'$.
	For (6), suppose that $M\leq_{a}M'$.
	Further, 
		let $s,s'\in\states$ and
		suppose that $s\pAccess{M}{a}{\community}s'$.
	That is,
		there is $\check{s}\in\states$ such that
			$\underline{s\preorder{\community\cup\set{a}}\check{s}}$ and
			$M\in\clo{a}{\check{s}}(\emptyset)$ and
			$\underline{\indist{a}{\check{s}}{s'}}$.
	Hence $\underline{M'\in\clo{a}{\check{s}}(\emptyset)}$ by the first hypothesis 
		(expanding the definition of ${\leq_{a}}\subseteq\messages\times\messages$).
	Thus $s\pAccess{M'}{a}{\community}s'$.
\end{proof}
\noindent
Note that we could have called the property of epistemic persistency also 
	epistemic \emph{mono}tonicity, but have not done so, 
		so as not to cause possible consternation with the reader about 
			the opposite naming of the corresponding axiom of epistemic \emph{anti}tonicity.
Axioms of modal operators and the corresponding properties of their respective accessibility relation have 
	the flipping tendency to have flipped senses.
(So is also the case for communal monotonicity.)

\subsubsection{Abstractly}
\begin{definition}[Kripke-model]\label{definition:KripkeModel}
We define the \emph{satisfaction relation} $\models$ for 
		LiP such that:
	\begin{eqnarray*}
		(\aModalFrame, \mathcal{V}), s\models P &\text{:iff}& s\in\mathcal{V}(P)\\[\jot]
		(\aModalFrame, \mathcal{V}), s\models\neg\phi &\text{:iff}& \text{not $(\aModalFrame, \mathcal{V}), s\models\phi$}\\[\jot]
		(\aModalFrame, \mathcal{V}), s\models\phi\land\phi' &\text{:iff}& \text{$(\aModalFrame, \mathcal{V}), s\models\phi$ and $(\aModalFrame, \mathcal{V}), s\models\phi'$}\\[\jot]
		(\aModalFrame, \mathcal{V}), s\models\proves{M}{\phi}{a}{\community} &\text{:iff}& 
			\begin{array}[t]{@{}l@{}}
				\text{for all $s'\in\states$, }
			 	\text{if $s\access{M}{a}{\community}s'$ then $(\aModalFrame, \mathcal{V}), s'\models\phi$\,,}
			\end{array}
	\end{eqnarray*}
where 
	\begin{itemize}
		\item $\mathcal{V}:\mathcal{P}\rightarrow\powerset{\states}$ designates a usual \emph{valuation function,} yet
			partially predefined such that for all $a\in\agents$ and $M\in\messages$,
				\begin{eqnarray*}
					\mathcal{V}(\knows{a}{M}) &\defeq& 
						\setst{s\in\states}{M\in\clo{a}{s}(\emptyset)}
				\end{eqnarray*}
				(If agents are Turing-machines 
					then $a$ knowing $M$ can be understood as $a$ being able to parse $M$ on its tape.)
		\item $\aModalFrame\defeq(\states,\set{\access{M}{a}{\community}}_{M\in\messages,a\in\agents,\community\subseteq\agents})$
			designates a (modal) \emph{frame} for LiP with 
				an \emph{abstract accessibility relation}---short, \emph{abstract accessibility}---${\access{M}{a}{\community}}\subseteq\states\times\states$ for 
			the proof modality such that \label{page:AbstractProofAccessibility} 
				\begin{itemize}
					\item	if for all $b\in\community\cup\set{a}$, $s\access{\sign{M}{a}}{b}{\community\cup\set{a}}s'$ 
						then $M\in\clo{a}{s'}(\emptyset)$
					\item if $M\in\clo{a}{s}(\emptyset)$
						then $s\access{M}{a}{\community}s$
					\item there is $s'\in\states$ such that 
						$s\access{M}{a}{\community}s'$
					\item for all $b\in\community\cup\set{a}$, 
						$({\access{\sign{M}{a}}{b}{\community\cup\set{a}}}\circ{\access{M}{a}{\community}})\subseteq{\access{M}{a}{\community}}$
					\item if $\community\subseteq\community'$ 
						then ${\access{M}{a}{\community}}\subseteq{\access{M}{a}{\community'}}$
					\item if $M\leq_{a}M'$ then ${\access{M}{a}{\community}}\subseteq{\access{M'}{a}{\community}}$
				\end{itemize}
	\item $(\aModalFrame,\mathcal{V})$ designates a (modal) \emph{model} for LiP.
	\end{itemize}
\end{definition}
\noindent
Looking back, 
	we recognise that Proposition~\ref{proposition:ConcreteAccessibility} actually establishes the important fact that
		our concrete accessibility $\pAccess{M}{a}{\community}$ in 
		Definition~\ref{definition:SemanticIngredients} realises 
			all the properties stipulated by  
		our abstract accessibility $\access{M}{a}{\community}$ in Definition~\ref{definition:KripkeModel};
		we say that 
		$$\text{$\pAccess{M}{a}{\community}$ \emph{exemplifies} (or \emph{realises}) $\access{M}{a}{\community}$\,.}$$
Further, observe that 
	LiP has a Herbrand-style semantics, that is, 
			logical constants (agent names) and 
			functional symbols (pairing, signing) are self-inter\-preted rather than 
		interpreted in terms of (other, semantic) constants and functions.
This simplifying design choice spares our framework from 
	the additional complexity that would arise from term-variable assignments \cite{FOModalLogic}, which in turn
		keeps our models propositionally modal.
Our choice is admissible because our individuals (messages) are finite.
(Infinitely long ``messages'' are non-messages; they can never be completely received, \eg
	transmitting irrational numbers as such is impossible.)

\begin{definition}[Truth \& Validity \cite{ModalLogicSemanticPerspective}]\label{definition:TruthValidity}\  
	\begin{itemize}
	\item The formula $\phi\in\pFormulas$ is \emph{true} (or \emph{satisfied}) 
		in the model $(\aModalFrame,\mathcal{V})$ at the state $s\in\states$ 
			:iff $(\aModalFrame,\mathcal{V}), s\models\phi$\,.
	\item The formula $\phi$ is \emph{satisfiable} in the model $(\aModalFrame,\mathcal{V})$ 
			:iff there is $s\in\states$ such that 
				$(\aModalFrame,\mathcal{V}), s\models\phi$\,.
	\item The formula $\phi$ is \emph{globally true} (or \emph{globally satisfied}) 
		in the model $(\aModalFrame,\mathcal{V})$, 
			written $(\aModalFrame,\mathcal{V})\models\phi$, :iff 
				for all $s\in\states$, $(\aModalFrame,\mathcal{V}),s\models\phi$\,.
	\item The formula $\phi$ is \emph{satisfiable}  
			:iff there is a model $(\aModalFrame,\mathcal{V})$ and a state $s\in\states$ such that 
				$(\aModalFrame,\mathcal{V}),s\models\phi$\,.	
	\item The formula $\phi$ is (\emph{universally true} or) \emph{valid}, written $\models\phi$, :iff 
			for all models $(\aModalFrame,\mathcal{V})$, $(\aModalFrame,\mathcal{V})\models\phi$\,.
	\end{itemize}
\end{definition}
\noindent
So we can paraphrase the law of epistemic antitonicity in Definition~\ref{definition:AxiomsRules} as:  
``Whatever a universally poorer message $M'$ can prove to $a$, 
any universally richer message $M$ can also prove to $a$, 
and this in all social contexts $\community\cup\set{a}$.''

\begin{proposition}[Admissibility of specific axioms and rules]\label{proposition:AxiomAndRuleAdmissibility}\    
	\begin{enumerate}
		\item $\models\knows{a}{a}$
		\item $\models\knows{a}{M}\limp\knows{a}{\sign{M}{a}}$
		\item $\models\knows{a}{\sign{M}{b}}\limp\knows{a}{\pair{M}{b}}$
		\item $\models(\knows{a}{M}\land\knows{a}{M'})\lequiv\knows{a}{\pair{M}{M'}}$
		\item $\models(\proves{M}{\phi}{a}{\community})\limp(\knows{a}{M}\limp\phi)$
		\item $\models(\proves{M}{\phi}{a}{\community})\limp
							\neg(\proves{M}{\neg\phi}{a}{\community})$
		\item $\models(\proves{M}{\phi}{a}{\community})\limp\bigwedge_{b\in\community\cup\set{a}}\proves{\sign{M}{a}}{(\knows{a}{M}\land\proves{M}{\phi}{a}{\community})}{b}{\community\cup\set{a}}$
		\item $\models(\proves{M}{\phi}{a}{\community\cup\community'})\limp\proves{M}{\phi}{a}{\community}$
		\item If $\models\knows{a}{M}\limp\knows{a}{M'}$ then 
					$\models(\proves{M'}{\phi}{a}{\community})\limp\proves{M}{\phi}{a}{\community}$\,.
	\end{enumerate}
\end{proposition}
\begin{proof}
	1--4 are immediate, and 
	5--9 follow directly from 
		the conditional reflexivity,
		the seriality,
		the signature property and communal transitivity,
		the communal monotonicity, and
		the epistemic persistency of $\access{M}{a}{\community}$, respectively.
\end{proof}

\begin{definition}[Semantic consequence and equivalence]\ 
	\begin{itemize}
		\item The formula $\phi'\in\pFormulas$ is \emph{a semantic consequence of} $\phi\in\pFormulas$, 
		written $\phi\Limp\phi'$, :iff 
			for all models $(\aModalFrame,\mathcal{V})$ and states $s\in\states$, 
				if $(\aModalFrame,\mathcal{V}),s\models\phi$ then $(\aModalFrame,\mathcal{V}),s\models\phi'$.
	\item $\phi'\in\pFormulas$ is \emph{semantically equivalent to} $\phi\in\pFormulas$, written $\phi\Lequiv\phi'$, :iff 
		$\phi\Limp\phi'$ and $\phi'\Limp\phi$.
	\end{itemize}
\end{definition}

\begin{fact}\label{fact:MaterialConditionalVSSemanticConsequence}%[Material conditional \& Semantic consequence]
	$\text{$\models\phi\limp\phi'$ if and only if $\phi\Limp\phi'$}$
\end{fact}
\begin{proof}
	By expansion of definitions.
\end{proof}

\subsection{Epistemic explication}\label{section:EpistemicExplication}
As announced,
	our interactive proofs have an \emph{epistemic explication} in terms of the epistemic impact that 
			they effectuate with their intended interpreting agents 
				(\ie the knowledge of their proof goals).
To see this, 
	consider that the elementary definition of proof accessibility on Page~\pageref{page:ProofAccessibility} can 
		be transformed by applying elementary-logical rules so that\label{page:ProofAccessibilityExpansion}
		$$\boxed{\begin{array}{@{}l@{}}
			(\aModalFrame, \mathcal{V}), s\models\proves{M}{\phi}{a}{\community}\quad\text{if and only if}\\ 
			\qquad\text{for all $\check{s}\in\states$,
					if $s\preorder{\community\cup\set{a}}\check{s}$ then\quad
						(data $\check{s}$ and peer $\community\cup\set{a}$ persistent)}\\
			\qquad\quad\text{$(\aModalFrame, \mathcal{V}), \check{s}\models\knows{a}{\hspace{-6ex}\underbrace{M}_{\text{sufficient evidence}}}\hspace{-5.5ex}\limp\K{a}(\hspace{-10ex}\underbrace{\phi}_{\hspace{9ex}\text{induced knowledge}}\hspace{-10ex})$\quad(epistemic impact),}
		\end{array}}$$
	with the standard epistemic modality $\K{a}$ being defined as
		 $$\begin{array}{@{}l@{}}
		 	(\aModalFrame, \mathcal{V}), \check{s}\models\K{a}(\phi)\quad\text{:iff}
		 	\quad\text{for all $s'\in\states$, if $\indist{a}{\check{s}}{s'}$ then 
				$(\aModalFrame, \mathcal{V}), s'\models\phi$.}
			\end{array}$$
As required,
	$\K{a}$---being defined by means of an equivalence relation---is S5, that is,  
		S4 plus the property $\models\neg\K{a}(\phi)\limp\K{a}(\neg\K{a}(\phi))$ of
			negative introspection \cite{Epistemic_Logic,MultiAgents}.
Hence, spelled out, 
	the epistemic explication is: 
	\begin{quote}
		\fbox{\parbox{0.825\textwidth}{%
		A proof effectuates
			a persistent epistemic impact in its intended community of peer reviewers that
					consists in the induction of the (propositional) knowledge of the proof goal 
						by means of the (individual) knowledge of the proof with the interpreting reviewer.
		}}
	\end{quote}

Observe that
	our notion of \emph{knowledge induction} (impact effectuation) is an instance of a \emph{parameterised persistent implication}, which:
		\begin{enumerate}
			\item is compatible with C.I.~Lewis relevant implication (\aka \emph{strict implication}), which
				does not stipulate any constraint on the accessibility relation of the implication (here $\preorder{\community\cup\set{a}}$)
			\item is \emph{intuitionistic implication} in Kripke's interpretation
					when the preorder $\preorder{\community\cup\set{a}}$ happens to be partial, \eg
						when $\community\cup\set{a}=\set{\CM}$ (total knowledge).
		\end{enumerate}
D.~Lewis relevant implication \label{page:DLewis} however
		(and \emph{a fortiori} Stalnaker's) is insufficient for capturing the induction.
Recall that  
	a statement $\phi$ implies $\phi'$ in a state $s$, by definition of D.~Lewis,
		if and only if $\phi\limp\phi'$ is true at all states closest to $s$ (here with respect to $\preorder{\community\cup\set{a}})$.
(Stalnaker required that there be a \emph{single} closest state.)
Order-theoretically,
	``closest to $s$ with respect to $\preorder{\community\cup\set{a}}$'' means 
	``that are atomic (\ie if minored then only by bottom) in the up-set 
		${\uparrow_{\preorder{\community\cup\set{a}}}}(s)\defeq\setst{s'\in\states}{s\preorder{\community\cup\set{a}}s'}$ of $s$ with respect to $\preorder{\community\cup\set{a}}$''.
Yet we do need to stipulate 
	truth at all states \emph{close} to $s$ (\ie \emph{all} states in ${\uparrow_{\preorder{\community\cup\set{a}}}}(s)$), 
	not just truth at all states close\emph{st} (\ie all \emph{atomic} states).
Otherwise persistency, which is
	essential to obtaining intuitionistic logic,  
may fail (\cf \cite[Section~2]{ModalLogicInformation} and \cite{InformationIL}).

Still,
	we believe that
		D.~Lewis relevant implication could be suitable for defining induction of 
			\emph{belief} (to be enshrined in a \emph{Logic of Evidence}) and even 
			\emph{false belief} (to be enshrined in a \emph{Logic of Deception}).\label{page:EvidenceDeception}
For belief,
	it does not make sense to insist on (peer) persistency, except perhaps for \emph{religious belief} 
		(among sectarian peers), and
	so quantifying over all closest states could be preferable over
	quantifying over all close states.
This is to be explored in future work.

We close this section with the statement of five epistemic interaction laws.
The first law---to be used as a lemma for the second---describes a reflexive interaction between
	individual and propositional knowledge in the following sense.
\begin{proposition}[Self-knowledge]\label{proposition:SelfKnowledge}
	$$\models\K{a}(\knows{a}{M})\lequiv\knows{a}{M}$$
\end{proposition}
\begin{proof}
	The $\limp$-direction follows from the reflexivity of $\indist{a}{}{}$, and
	the $\leftarrow$-direction from the definition of $\indist{a}{}{}$ as state indistinguishability with respect to individual knowledge.
\end{proof}

The second law describes an important interaction between
	individual and propositional knowledge by means of their respective languages $\messages$ and $\pFormulas$.
For the sake of stating the law succinctly, we recall the following standard definition.
\begin{definition}[Language equivalence]
	Let $L\subseteq\pFormulas$ designate a sublanguage of $\pFormulas$. 
	Then two pointed models $(\aModalFrame,\mathcal{V}),s$ and $(\aModalFrame,\mathcal{V}),s'$ are \emph{$L$-equivalent,} 
		written $\indist{L}{(\aModalFrame,\mathcal{V}),s}{(\aModalFrame,\mathcal{V}),s'}$, :iff 
			for all $\phi\in L$, $(\aModalFrame,\mathcal{V}),s\models\phi$ iff 
				$(\aModalFrame,\mathcal{V}),s'\models\phi$.
			(The relation $\indist{\pFormulas}{}{}$ is called \emph{elementary equivalence}.) 
\end{definition}
\noindent
The law says that 
	state indistinguishability with respect to individual knowledge equals 
	state indistinguishability with respect to propositional knowledge.
\begin{proposition}[Indistinguishability]\label{proposition:Indistinguishability}
	Let $a\in\agents$ and
		\begin{eqnarray*}
			\mathit{Re} &\defeq& \setst{\knows{a}{M}}{M\in\messages}\\
			\mathit{Dicto} &\defeq& \setst{\K{a}(\phi)}{\phi\in\pFormulas}.
		\end{eqnarray*}
	Then,  
		\begin{eqnarray*}
			{\indist{\mathit{Re}}{}{}} &=& {\indist{\mathit{Dicto}}{}{}}\,.
		\end{eqnarray*}
\end{proposition}
\begin{proof}
	The $\subseteq$-direction follows from 
		the definition of $\indist{a}{}{}$ as state indistinguishability with respect to individual knowledge, and
		the transitivity of $\indist{a}{}{}$; and 
	the $\supseteq$-direction from 
		the fact that for all $M\in\messages$, $(\knows{a}{M})\in\pFormulas$ and
		Proposition~\ref{proposition:SelfKnowledge}.
\end{proof}

The third law---to be used as a lemma for the fourth---describes 
	an important interaction between individual and propositional knowledge by means of message signing.
		(It is an epistemic expansion of Theorem~\ref{theorem:SomeUsefulDeducibleLogicalLaws}.20.)
The law also gives an example of interpreted communication:  
	how to induce propositional knowledge with 
		a certain piece of individual knowledge (\ie a signed message).
\begin{proposition}[The purpose of signing]\label{proposition:ThePurposeOfSigning} 
	$$\models\knows{a}{\sign{M}{b}}\limp\K{a}(\knows{b}{\sign{M}{b}})$$
\end{proposition}
\begin{proof}
	Let 
		$(\aModalFrame,\mathcal{V})$ designate an arbitrary LiP-model, and
	let $s\in\states$, $a,b\in\agents$, and $M\in\messages$.
	Further, 
		suppose that $(\aModalFrame,\mathcal{V}),s\models\knows{a}{\sign{M}{b}}$ and
		let $s'\in\states$ such that $\indist{a}{s'}{s}$.
	Hence,
		$(\aModalFrame,\mathcal{V}),s'\models\knows{a}{\sign{M}{b}}$ by definition of $\indist{a}{}{}$ as  
			state indistinguishability with respect to individual knowledge, and thus
		$(\aModalFrame,\mathcal{V}),s'\models\knows{b}{\sign{M}{b}}$ due to 
			the unforgeability of signatures
				(only $b$ can generate $\sign{M}{b}$).
\end{proof}
The fourth law describes an important interaction between knowledge and interactive proofs, again by means of message signing.
The law also gives an explication of the epistemic impact of \emph{signed} interactive proofs.
\begin{theorem}[Proofs of Knowledge]\label{theorem:ProofsOfKnowledge}
	Signed interactive proofs are peer-re\-viewable \emph{proofs of knowledge}\footnote{This terminology is inspired by \cite[Page~262]{FoundationsCryptographyOne}, where 
	such proofs are defined as ``[\ldots\negthinspace] proofs in which the prover [here $a$] asserts 
		``knowledge'' of some object [\ldots\negthinspace] and 
			not merely its existence [\ldots\negthinspace]'' by means of 
				probabilistic polynomial-time interactive Turing machines.}
		in the following formal sense:
		$$\models(\proves{M}{\phi}{a}{\community})\limp
			\bigwedge_{b\in\community\cup\set{a}}\proves{\sign{M}{a}}{(\underbrace{\knows{a}{M}\land\K{a}(\phi)}_{\mathrm{induced\ knowledge}})}{b}{\community\cup\set{a}}.$$
\end{theorem}
\begin{proof}
	We first prove the stronger fact that
		$$\models(\proves{M}{\phi}{a}{\community})\limp
			\bigwedge_{b\in\community\cup\set{a}}\proves{\sign{M}{a}}{(\knows{a}{M}\land\proves{M}{\phi}{a}{\community})}{b}{\community\cup\set{a}}.$$
	Let 
		$(\aModalFrame,\mathcal{V})$ designate an arbitrary LiP-model, and
	let $s\in\states$, $a\in\agents$, $\community\subseteq\agents$, $b\in\community\cup\set{a}$, and $M\in\messages$.
	Further, 
		suppose that $(\aModalFrame,\mathcal{V}),s\models\proves{M}{\phi}{a}{\community}$, 
		let $\check{s}\in\states$ such that $s\preorder{\community\cup\set{a}\cup\set{b}}\check{s}$, and
		suppose that $(\aModalFrame,\mathcal{V}),\check{s}\models\knows{b}{\sign{M}{a}}$.
	Hence, 
		$(\aModalFrame,\mathcal{V}),\check{s}\models\K{b}(\knows{a}{\sign{M}{a}})$ by 
				Proposition~\ref{proposition:ThePurposeOfSigning}, and thus
		$(\aModalFrame,\mathcal{V}),\check{s}\models\K{b}(\knows{a}{M})$ by 
			\emph{modus ponens} of  
				$\models\K{b}(\knows{a}{\sign{M}{a}}\limp\knows{a}{M})$ (epistemic necessitation of signature analysis) and
				$\models\K{b}(\knows{a}{\sign{M}{a}}\limp\knows{a}{M})\limp
					(\K{b}(\knows{a}{\sign{M}{a}})\limp\K{b}(\knows{a}{M}))$ (Kripke's law).
	Now,  
		let $\tilde{s}\in\states$ such that $\indist{b}{\check{s}}{\tilde{s}}$.
	Thus,
		$\check{s}\preorder{\community\cup\set{a}\cup\set{b}}\tilde{s}$, thus
		$s\preorder{\community\cup\set{a}\cup\set{b}}\tilde{s}$ by transitivity, and thus
		$s\preorder{\community\cup\set{a}}\tilde{s}$ by the hypothesis that $b\in\community\cup\set{a}$.
	Hence, 
		$(\aModalFrame,\mathcal{V}),\tilde{s}\models\proves{M}{\phi}{a}{\community}$ by 
			peer persistency, 
		$(\aModalFrame,\mathcal{V}),\check{s}\models\K{b}(\proves{M}{\phi}{a}{\community})$ by 
			discharge of the last hypothesis, and thus
		$(\aModalFrame,\mathcal{V}),\check{s}\models\K{b}(\knows{a}{M}\land\proves{M}{\phi}{a}{\community})$.
		
	Our theorem now follows from 
		a stronger version of epistemic truthfulness, that is, 
			$\models(\proves{M}{\phi}{a}{\community})\limp(\knows{a}{M}\limp\K{a}(\phi))$, which 
			in turn follows from the expansion of the truth condition of $\proves{M}{\phi}{a}{\community}$.
\end{proof}

The fifth law describes an important interaction between 
	\emph{common knowledge} \cite{Epistemic_Logic,MultiAgents} and 
	purported interactive proofs, namely
		their falsifiability in a \emph{communal} sense of \emph{Popper's critical rationalism}.
More precisely,
	we refer to Popper's dictum that 
		a hypothesis (here, that a purported interactive proof is indeed a proof) should be falsifiable in the sense that 			\textbf{\emph{if the hypothesis is false} then its falsehood should be cognisable} (here, commonly knowable).
In the present paper,
	we restrict the relation between 
		Popper's \emph{\oe uvre} and 
		our work 
	to this succinct dictum.
Recall from \cite{Epistemic_Logic,MultiAgents} that
	common knowledge among a community $\community$ can be captured with a modality $\CK{\community}$ defined as \label{page:CommonKnowledge}
		 $$\begin{array}{@{}l@{}}
		 	(\aModalFrame, \mathcal{V}), s\models\CK{\community}(\phi)\quad\text{:iff}
		 	\quad\text{for all $s'\in\states$, if $\indist{\community}{s}{s'}$ then 
				$(\aModalFrame, \mathcal{V}), s'\models\phi$,}
			\end{array}$$
where ${\indist{\community}{}{}}\defeq(\bigcup_{a\in\community}\indist{a}{}{})^{*}$.
The intuition is that
a statement $\phi$ is common knowledge in a community $\community$ of agents when:  
	all agents know that $\phi$ is true (call this new statement $\phi'$), 
	all agents know that $\phi'$ is true (call this new statement $\phi''$), 
	all agents know that $\phi''$ is true (call this new statement $\phi'''$), etc.
Note that depending on the properties of the employed communication lines, 
common knowledge may have to be pre-established off those lines along other lines \cite{CommonKnowledge},
which is also why there is no built-in common-knowledge operator in LiP.

\begin{theorem}[Falsifiability of interactive ``proofs'']\label{theorem:Falsifiability}
Interactive ``proofs'' are falsifiable in a communal sense of Popper's, that is,  
	if a datum $M\in\messages$ is not a $\community\cup\set{a}$-reviewable proof of a statement $\phi\in\pFormulas$
	then this fact is communally cognisable as such by $\community\cup\set{a}$ in terms of 
		the common knowledge among $\community\cup\set{a}$ of that fact. Formally,
				$$\models(\neg\thinspace\proves{M}{\phi}{a}{\community})\limp
					\CK{\community\cup\set{a}}(\neg\thinspace\proves{M}{\phi}{a}{\community}).$$
\end{theorem}
\begin{proof}
	Let 
		$(\aModalFrame,\mathcal{V})$ designate an arbitrary LiP-model, and
	let $s\in\states$, $a\in\agents$, $\community\subseteq\agents$, and $M\in\messages$.
	Further, 
		suppose that $(\aModalFrame,\mathcal{V}),s\models\neg\thinspace\proves{M}{\phi}{a}{\community}$,
		let $s'\in\states$ such that $\indist{\community\cup\set{a}}{s}{s'}$ (thus $s'\preorder{\community\cup\set{a}}s$), and
		suppose by contradiction that 
			$(\aModalFrame,\mathcal{V}),s'\models\proves{M}{\phi}{a}{\community}$.
	Hence $(\aModalFrame,\mathcal{V}),s\models\proves{M}{\phi}{a}{\community}$ by peer persistency---contradiction!
\end{proof}
\noindent
Note also the following simpler fact, which
	asserts that what is commonly accepted as proof constitutes common knowledge.
\begin{fact}[Common proof knowledge]\label{fact:CommonProofKnowledge}
	$$\models(\proves{M}{\phi}{a}{\community})\limp
		\CK{\community\cup\set{a}}(\proves{M}{\phi}{a}{\community})$$
\end{fact}
\noindent
This however does \emph{not} mean that $M$ is known by everybody in $\community\cup\set{a}$!

\subsection{Oracle-computational explication}\label{section:OracleComputationalExplication}
As announced,
	our interactive proofs have also an \emph{oracle-computational explication} in terms of 
		a computation oracle that acts as a hypothetical provider 
			and thus as an imaginary epistemic source of our interactive proofs.		
To see this, 
	consider that the elementary definition of 
		proof accessibility in Definition~\ref{definition:SemanticIngredients} can be \emph{redefined} (for the time being) such that
			for all $s,s'\in\states$, \label{page:ProofAccessibilityBis}
			\begin{eqnarray} 
				s\pAccess{M}{a}{\community}s' &\defiff&
					s'\in\hspace{-5.5ex}
					\bigcup_{\scriptsize
					\mbox{
					$\begin{array}{@{}c@{}}
						\text{$s<_{\community\cup\set{a}}^{M}\tilde{s}$ and
						}\\[0.5\jot] 
						M\in\clo{a}{\tilde{s}}(\emptyset) 
					\end{array}$}
					}\hspace{-5ex}[\tilde{s}]_{\indist{a}{}{}}\\\notag	
				\text{for all $M\in\messages$ and $\community\subseteq\agents$, $<_{\community}^{M}$} &\defeq& (\bigcup_{a\in\community}{<_{a}^{M}})^{++}\\\notag 
				s<_{a}^{M}s' &\defiff& \clo{a}{s}(\set{M})=\clo{a}{s'}(\emptyset),
			\end{eqnarray}
where `$^{++}$' designates the closure operation of so-called \emph{generalised transitivity}
				in the sense that 
					 ${<_{\community}^{M}}\circ{<_{\community}^{M'}}\subseteq{<_{\community}^{\pair{M}{M'}}}$.
Note that when $s<_{a}^{M}s'$ for some states $s,s'\in\states$, 
	agent $a$ can conceive of $s'$ as 
		$s$ yet \emph{minimally enriched} with the
			information token $M$, for which 
				$a$ could imagine invoking an \emph{oracle agent}.
In other words, 
	if $a$ knew $M$ (\eg if $a$ \emph{received} $M$ from the oracle) 
	then $a$ could not distinguish $s$ from $s'$ in the sense of $\indist{a}{}{}$.
This hypothetical knowledge was
	called \emph{adductive knowledge} in \cite{skramerIMLA2008}---from now on 
		also \emph{oracle knowledge}---and 
	implemented with a concrete message reception \emph{event} for $a$ that
		carries the information of $M$ in $s'$.
Now, similarly to Page~\pageref{page:ProofAccessibilityExpansion}, 
	our above-redefined proof-accessibility relation can be 
		transformed and then
		used for redefining (again, for the time being) the proof modality as follows:
		$$\boxed{\begin{array}{@{}l@{}}
			(\aModalFrame, \mathcal{V}), s\models\proves{M}{\phi}{a}{\community}\quad\defiff\\ 
			\qquad\text{for all $s'\in\states$,
					if $s<_{\community\cup\set{a}}^{M}s'$ then\quad(peer $\community\cup\set{a}$ persistent)}\\
			\qquad\quad\text{$(\aModalFrame, \mathcal{V}), s'\models\knows{a}{M}\limp\K{a}(\phi)$\quad(epistemic impact).}
		\end{array}}$$
The new notion of proof resulting from Accessibility~Relation~2 
	is obviously weaker than 
	our original notion resulting from Accessibility~Relation~1 on Page~\pageref{page:ProofAccessibility}, in the sense that
		the epistemic impact of Notion~1 is \emph{data persistent}, \eg 
				is the case even when more messages than just the proof are learnt, whereas 
			the one of Notion~2 is not necessarily so, that is,  
				is the case possibly only at the instant of learning the proof.
(Still, both notions induce knowledge and not only belief!)
Therefore,  
	we call 
		interactive proofs in the sense of Notion~1 \emph{persistent} or \emph{extant} and
		those in the sense of Notion~2 \emph{instant} interactive proofs.
For multi-agent distributed systems, instant interactive proofs are interesting, \eg 
	for \emph{accountability} (\cf \cite{MMFAAMAS} and \cite{MMMSRTETEVVVS}, both based on \cite{skramerIMLA2008}).
In accountable multi-agent distributed systems, 
	an agent may prove her correct past behaviour in the present state to some judge, \eg 
		with a signed logfile \cite{MMFAAMAS}, but
			may well then cease behaving correctly in the future.
Hence her correctness proof is instant but may well not be persistent.
The epistemic explication for Notion~2 is, spelled out: 
	\begin{quote}
		\fbox{\parbox{0.89\textwidth}{%
		An \emph{instant} proof effectuates
			an \emph{instant} epistemic impact in its intended community of peer reviewers that
					consists in the induction of the (propositional) knowledge of the proof goal 
						by means of the (individual) knowledge of the proof with the interpreting reviewer \cite{KramerICLA2013}.
		}}
	\end{quote}
Observe that
	our notion of knowledge induction (impact effectuation) for instant interactive proofs is 
		a parameterised instant implication, which
			\emph{is} compatible with D.~Lewis relevant implication	
				(\cf our corresponding discussion on Page~\pageref{page:DLewis}).
That is, 
	$\knows{a}{M}\limp\K{a}(\phi)$ is true at all states $s'$ closest to $s$ with respect to $\preorder{\community\cup\set{a}}$, that is, 
	for which $s<_{\community\cup\set{a}}^{M}s'$.
The token $M$ represents the minimal difference.
Of course, $a$ may in fact know $M$ in $s$; so the conditional is \emph{not necessarily counter-}factual.

Our above definitions can be related to our original ones as follows.
\begin{proposition}\label{proposition:MessageDifference} 
For all $s,s'\in\states$:
\begin{enumerate}
	\item $s\preorder{a}s'$ if and only if there is $M\in\messages$ such that $s<_{a}^{M}s'$ 
	\item $s\preorder{\community}s'$ if and only if
			there is $M\in\messages$ such that $s<_{\community}^{M}s'$
\end{enumerate}
\end{proposition}
\begin{proof}
	We prove the if-direction of 
		(1)---the only-if-direction being obvious, and 
		(2) obviously following from (1).
	Let $s,s'\in\states$ and
	suppose that $s\preorder{a}s'$.
	Hence there is a finite $\data\subseteq\messages$ such that 
			$\clo{a}{s}(\data)=\clo{a}{s'}(\emptyset)$, because
				$\msgs{a}(s)$ and $\msgs{a}(s')$ are finite (\cf Page~\pageref{page:RawData}).
	Hence there is $M\in\messages$ such that
		$\clo{a}{s}(\set{M})=\clo{a}{s}(\data)$.
	(For example, 
		choose $M=(M_{1},\ldots,M_{n})$ where $\data=\{M_{1},\ldots,M_{n}\}$.)
	Thus, 
		$\clo{a}{s}(\set{M})=\clo{a}{s'}(\emptyset)$ by transitivity, and
		$s<_{a}^{M}s'$ by definition.
\end{proof}	
Hence, 
	Notion~1 can be recovered from Notion~2 by \emph{redefining}
			the proof accessibility on Page~\pageref{page:ProofAccessibility} such that
				for all $s,s'\in\states$, 
					\begin{eqnarray} 
						s\pAccess{M}{a}{\community}s' &\defiff& 
							s'\in\hspace{-7.5ex}\bigcup_{\scriptsize
								\mbox{
									$\begin{array}{@{}c@{}}
										s\mathrel{(\bigcup_{M'\in\messages}<_{\community\cup\set{a}}^{M'})}\check{s}\\[0.5\jot] 
									\text{and $M\in\clo{a}{\check{s}}(\emptyset)$} 
								\end{array}$}
							}\hspace{-7.5ex}[\check{s}]_{\indist{a}{}{}},
					\end{eqnarray}
					and thus Notion~3 and Notion~1 are equivalent.

\begin{proposition}\label{proposition:Instancy}
When the proof modality is interpreted with Notion~2, 
	$$\models\proves{a}{\phi}{a}{\emptyset}\lequiv\K{a}(\phi).$$
\end{proposition}
\begin{proof}
		By the fact that ${<_{\set{a}}^{a}}={\indist{a}{}{}}$.
\end{proof}

We leave the further study of instant interactive proofs for future work.

\subsection{More results}\label{section:ImportantProperties}
\begin{theorem}[Adequacy]\label{theorem:Adequacy} 
	$\LiPded$ is \emph{adequate} for $\models$, that is:
	\begin{enumerate}
		\item if $\LiPded\phi$ then $\models\phi$\quad(axiomatic soundness)
		\item if $\models\phi$ then $\LiPded\phi$\quad(semantic completeness).
	\end{enumerate}
\end{theorem}
\begin{proof}
	Soundness follows from
		the admissibility of axioms and rules (\cf Proposition~\ref{proposition:AxiomAndRuleAdmissibility}), and
	completeness by means of the classical construction of canonical models,
		using Lindenbaum's construction of maximally consistent sets (\cf Appendix~\ref{appendix:Proofs}).
\end{proof}
We leave the study of \emph{strong adequacy} \cite[Section~3]{ModalProofTheory} for future work.

\begin{corollary}[Consistency]\label{corollary:Consistency}\ 
	\begin{enumerate}
		\item If $\LiPded\phi$ then $\not\LiPded\neg\phi$.
		\item $\not\LiPded\false$
	\end{enumerate}
\end{corollary}
\begin{proof}
	As usual: suppose that $\LiPded\phi$.
		Hence $\models\phi$ by semantic completeness.
		Hence $\not\models\neg\phi$ by the definition of $\models$\,.
		Hence $\not\LiPded\neg\phi$ by contraposition of axiomatic soundness; and 
	(2) follows jointly from 
			the instance of (1) where $\phi\defeq\true$, 
			the axiom $\LiPded\knows{a}{a}$, and  
			the macro-definitions of 
				$\true$ as $\knows{a}{a}$ and
				$\false$ as $\neg\true$.
\end{proof}

\begin{corollary}[Stateful proof equality]\label{corollary:ProofEquality}
Let $(\aModalFrame, \mathcal{V})$ designate an arbitrary LiP-model, and
let $s\in\states$, $M\in\messages$, $a\in\agents$, $\community\subseteq\agents$, and $\phi\in\pFormulas$.
Further let:
\begin{eqnarray*}
	\denotation{M}{s}{a}{\community} &\defeq& \setst{\phi}{(\aModalFrame, \mathcal{V}),s\models\proves{M}{\phi}{a}{\community}}\quad\text{(\textbf{message meaning})}\\
	\provesEq{s}{a}{\community} &\defeq& \setst{(M,M')\in\messages\times\messages}{\denotation{M}{s}{a}{\community}=\denotation{M'}{s}{a}{\community}}\\
	0 &\defeq& \equivclass{a}{\provesEq{s}{a}{\community}}\\
	\equivclass{M}{\provesEq{s}{a}{\community}} + \equivclass{M'}{\provesEq{s}{a}{\community}} &\defeq& 
		\equivclass{\pair{M}{M'}}{\provesEq{s}{a}{\community}}.
\end{eqnarray*}

Then, $$\langle\messages/_{\provesEq{s}{a}{\community}},0,+\rangle$$ is an \emph{idempotent commutative monoid,} that is, 
	for all $\boldsymbol{M},\boldsymbol{M}',\boldsymbol{M}''\in\messages/_{\provesEq{s}{a}{\community}}:$
	\begin{enumerate}
		\item $\boldsymbol{M}+(\boldsymbol{M}'+\boldsymbol{M}'')=(\boldsymbol{M}+\boldsymbol{M}')+\boldsymbol{M}''$
				\quad(associativity)
		\item $\boldsymbol{M}+\boldsymbol{M}'=\boldsymbol{M}'+\boldsymbol{M}$\quad(commutativity)
		\item $\boldsymbol{M}+\boldsymbol{M}=\boldsymbol{M}$\quad(idempotency)
		\item $\boldsymbol{M}+0=\boldsymbol{M}$\quad(neutral element).
	\end{enumerate}
\end{corollary}
\begin{proof}
	By the soundness of 
		proof associativity, commutativity, and idempotency, and 
		the law of a self-neutral proof element, respectively 
			(\cf Theorem~\ref{theorem:SomeUsefulDeducibleStructuralLaws}).
\end{proof}

\begin{theorem}[Finite-model property]\label{theorem:FMP}
	For any LiP-model $\mathfrak{M}$,
	if $\mathfrak{M}, s\models\phi$ 
	then there is a finite LiP-model $\mathfrak{M}_{\mathrm{fin}}$ such that 
		$\mathfrak{M}_{\mathrm{fin}}, s\models\phi$.
\end{theorem}
\begin{proof}
	By the fact that 
		the \emph{minimal filtration} \cite{ModalModelTheory}
			$$\mathfrak{M}_{\mathrm{flt}}^{\mathrm{min},\Gamma}\defeq
				(\states/_{\sim_{\Gamma}},
					\set{(\access{M}{a}{\community})^{\mathrm{min},\Gamma}}_{M\in\messages,a\in\agents,\community\subseteq\agents},\mathcal{V}_{\Gamma})$$ of 
			any LiP-model $\mathfrak{M}\defeq(\states,\set{\access{M}{a}{\community}}_{M\in\messages,a\in\agents,\community\subseteq\agents},\mathcal{V})$ 
			through a finite $\Gamma\subseteq\pFormulas$ is a finite LiP-model such that 
				for all $\gamma\in\Gamma$, 
					$\mathfrak{M},s\models\gamma$ if and only if 
					$\mathfrak{M}_{\mathrm{flt}}^{\mathrm{min},\Gamma},[s]_{\sim_{\Gamma}}\models\gamma$.
			Following \cite{ModalModelTheory} for our setting, 
				we define 
				\begin{eqnarray*}
					{\sim_{\Gamma}}&\defeq&
						\setst{(s,s')\in\states\times\states}{\text{for all $\gamma\in\Gamma$, 
							$\mathfrak{M},s\models\gamma$ iff $\mathfrak{M},s'\models\gamma$}}\\ 
					{(\access{M}{a}{\community})^{\mathrm{min},\Gamma}} &\defeq&
						\setst{([s]_{\sim_{\Gamma}},[s']_{\sim_{\Gamma}})}{(s,s')\in{\access{M}{a}{\community}}}\\
					\mathcal{V}_{\Gamma}(P)&\defeq&
							\setst{[s]_{\sim_{\Gamma}}}{s\in\mathcal{V}(P)}\,.
				\end{eqnarray*}
			We further fix 
					$M\in\clo{a}{[s]_{\sim_{\Gamma}}}(\emptyset)$ as  
						$[s]_{\sim_{\Gamma}}\in\mathcal{V}_{\Gamma}(\knows{a}{M})$
				and choose $\Gamma$ to be the (finite) sub-formula closure of $\phi$.
			Hence, we are left to prove that 
				$\mathfrak{M}_{\mathrm{flt}}^{\mathrm{min},\Gamma}$ is indeed an LiP-model, which
				means that we are left to prove that 
						$(\access{M}{a}{\community})^{\mathrm{min},\Gamma}$ has 
					all the properties of ${\access{M}{a}{\community}}$.
			Simply apply definitions back and forth.
\end{proof}

\begin{corollary}[Algorithmic decidability]\label{corollary:Decidability}
	LiP is algorithmically decidable.
\end{corollary}
\begin{proof}
	In order to algorithmically decide whether or not $\phi\in\LiP$ 
	(that is, $\LiPded\phi$) for some $\phi\in\pFormulas$ (and the current choice of $\messages$), 
		axiomatic adequacy allows us to check whether or not $\neg\phi$ is locally satisfiable (That is,
			whether or not $\mathfrak{M},s\models\neg\phi$ for 
				some LiP-model $\mathfrak{M}$ and state $s$.
					Also, $M\in\clo{a}{s}(\emptyset)$ on the currently chosen  
						message language $\messages$ is obviously decidable;
						for other, more complex message languages including cryptographic messages, 
							see for example \cite{TGJ} and \cite{BaskarJamSuresh2010}).
	But then, the finite-model property of LiP allows us 
		to enumerate all finite LiP-models $\mathfrak{M}_{\mathrm{fin}}$ up to a size of at most 2 to the power 
			of the size $n$ of the sub-formula closure of $\neg\phi$ and
		to check whether or not $\mathfrak{M}_{\mathrm{fin}},s\models\neg\phi$.
	(First, there are at most $2^n$ equivalence classes for $n$ formulas.
	 Second, 
		checking classical negation within a finite model is also a finite task.)
\end{proof}
Note that 
	the algorithmic complexity of LiP will depend on 
		the specific choice of $\messages$ and the correspondingly chosen term axioms.

\section{Interactive programs from interactive proofs}\label{section:TiCLfromLiP}
In this section,
	we present the remaining interactive formalisms mentioned in Figure~\ref{figure:Methodology}, namely: 
		(1) an interactive multi-agent S4-modal logic (iS4),
		(2) interactive Intuitionistic Logic (iIL), and  
		(3) (typed) interactive Combinatory Logic (TiCL) providing \emph{(typed) interactive programs as 
				agent-specific combinators.}

\subsection{Interactive S4 (iS4)}\label{section:iS4}
An interactive multi-agent classical normal S4-modal logic (iS4) can be obtained as 
	a fragment of LiP with epistemically guarded quantifiers over term variables (egFOLiP) by
		interpreting \emph{the iS4-necessity modality as 
			the existence of the individual knowledge of an interactive proof in the sense of LiP.}

We define iS4 and its necessity modality by the following seven laws
			\begin{itemize}
				\item for all axioms $\phi$ of classical propositional logic,
						$\vdash_{\mathrm{iS4}}\phi$
				\item $\vdash_{\mathrm{iS4}}\Box_{a}^{\community}(\phi\limp\phi')\limp
						(\Box_{a}^{\community}(\phi)\limp\Box_{a}^{\community}(\phi'))$
						\quad(K)
				\item $\vdash_{\mathrm{iS4}}\Box_{a}^{\community}(\phi)\limp
							\Box_{a}^{\community}(\Box_{a}^{\community}(\phi))$
						\quad(4)
				\item $\vdash_{\mathrm{iS4}}\Box_{a}^{\community}(\phi)\limp\phi$
						\quad(T)
				\item $\vdash_{\mathrm{iS4}}\Box_{a}^{\community\cup\community'}(\phi)\limp
						\Box_{a}^{\community}(\phi)$
						\quad(group decomposition)
				\item $\set{\phi}\vdash_{\mathrm{iS4}}\Box_{a}^{\community}(\phi)$
						\quad(N)
				\item $\set{\phi,\phi\limp\phi'}\vdash_{\mathrm{iS4}}\phi$
						\quad(\emph{modus ponens})
			\end{itemize}
and its egFOLiP-interpretation map as
		\begin{eqnarray*}
			\Box_{a}^{\community}(\phi) &\mapsto& \exists m(\knows{a}{m}\land\proves{m}{\phi}{a}{\community})
				\qquad\text{(epistemic provability)}
		\end{eqnarray*}
(S4- and iS4-syntax are identical modulo iS4-modality parameters $a$ and $\community$.)
The corresponding iS4-possibility modality can then be macro-defined as usual and interpreted in egFOLiP accordingly:
		\begin{eqnarray*}
			\Diamond_{a}^{\community}(\phi) &
				\defeq& \neg\Box_{a}^{\community}(\neg\phi)\\
				&\mapsto& \neg\exists m(\knows{a}{m}\land\proves{m}{\neg\phi}{a}{\community})\\
				&\lequiv& \forall m(\knows{a}{m}\limp\neg(\proves{m}{\neg\phi}{a}{\community}))\\
				&\defeq& \forall m(\knows{a}{m}\limp\proofdiamond{M}{\phi}{a}{\community})
		\end{eqnarray*}
The following proposition asserts that 
	our egFOLiP-interpretation of our $\Box_{a}^{\community}$-modality is adequate in
		the sense of satisfying all the required modal laws.
\begin{proposition}
	The following iS4-interpretations are valid egFOLiP-laws:
	\begin{enumerate}
		\item $\exists m(\knows{a}{m}\land\proves{m}{(\phi\limp\phi')}{a}{\community})\limp
					(\exists m(\knows{a}{m}\land\proves{m}{\phi}{a}{\community})\limp
						\exists m(\knows{a}{m}\land\proves{m}{\phi'}{a}{\community}))$
		\item $\exists m(\knows{a}{m}\land\proves{m}{\phi}{a}{\community})\limp
				\exists m(\knows{a}{m}\land\proves{m}{(\exists m(\knows{a}{m}\land\proves{m}{\phi}{a}{\community}))}{a}{\community})$
		\item $\exists m(\knows{a}{m}\land\proves{m}{\phi}{a}{\community})\limp\phi$
		\item $\exists m(\knows{a}{m}\land\proves{m}{\phi}{a}{\community\cup\community'})\limp
				\exists m(\knows{a}{m}\land\proves{m}{\phi}{a}{\community})$
		\item if $\phi$ is valid in egFOLiP 
				then $\exists m(\knows{a}{m}\land\proves{m}{\phi}{a}{\community})$ is valid in egFOLiP
	\end{enumerate}
\end{proposition}
\begin{proof}
	We reason rigorously but only semi-formally, as 
		we have not explicitly stipulated all axioms for egFOLiP, by
			appealing only to obviously valid principles.
	These principles are the LiP-laws as well as usual FOL-laws.
	No special first-order modal laws or models \cite{FOModalLogic} are required for the proof.
	
	For (1), suppose that 
		$\exists m(\knows{a}{m}\land\proves{m}{(\phi\limp\phi')}{a}{\community})$ is locally true,
			that is, at an arbitrary location.
	Thus there is $M\in\messages$ such that
		$\knows{a}{M}\land\proves{M}{(\phi\limp\phi')}{a}{\community}$ is true at the same location.
	Further suppose that 
		$\exists m(\knows{a}{m}\land\proves{m}{\phi}{a}{\community})$ is true there too.
	Thus there is $M'\in\messages$ such that
		$\knows{a}{M'}\land\proves{M'}{\phi}{a}{\community}$ is true there.
	Hence,
		$\knows{a}{\pair{M}{M'}}\land\proves{\pair{M}{M'}}{\phi'}{a}{\community}$ is true there, by
			message pairing and the generalised Kripke-law, respectively.
	In conclusion,
		$\exists m(\knows{a}{m}\land\proves{m}{\phi'}{a}{\community})$ is locally true.
	
	For (2), suppose that 
		$\exists m(\knows{a}{m}\land\proves{m}{\phi}{a}{\community})$ is locally true.
	Thus there is $M\in\messages$ such that 
		$\knows{a}{M}\land\proves{M}{\phi}{a}{\community}$ is true at the same location.
	Hence,
		$\knows{a}{\sign{M}{a}}\land\proves{\sign{M}{a}}{(\knows{a}{M}\land\proves{M}{\phi}{a}{\community})}{a}{\community}$ is true there too,
			by signature synthesis as well as authentic knowledge and peer review, respectively.
	Thus there is $M'\in\messages$ such that
		$\knows{a}{M'}\land\proves{M'}{(\knows{a}{M}\land\proves{M}{\phi}{a}{\community})}{a}{\community}$ is true there.
	In conclusion,
		$\exists m(\knows{a}{m}\land\proves{m}{(\exists m(\knows{a}{m}\land\proves{m}{\phi}{a}{\community}))}{a}{\community})$ is locally true.
	
	For (3), suppose that 
		$\exists m(\knows{a}{m}\land\proves{m}{\phi}{a}{\community})$ is locally true.
	Thus there is $M\in\messages$ such that 
		$\knows{a}{M}\land\proves{M}{\phi}{a}{\community}$ is true at the same location.
	Hence and in conclusion,
		$\phi$ is true there too, by epistemic truthfulness.
	
	For (4), suppose that 
		$\exists m(\knows{a}{m}\land\proves{m}{\phi}{a}{\community\cup\community'})$ is locally true.
	Thus there is $M\in\messages$ such that
		$\knows{a}{M}\land\proves{M}{\phi}{a}{\community\cup\community'}$ is true at the same location.
	Hence,
		$\knows{a}{M}\land\proves{M}{\phi}{a}{\community}$ is true there too, by
			group decomposition.
	In conclusion,	
		$\exists m(\knows{a}{m}\land\proves{m}{\phi}{a}{\community})$ is locally true.
	
	For (5), suppose that 
		$\phi$ is valid in egFOLiP.
	Hence,
		$\proves{a}{\phi}{a}{\community}$ is valid in egFOLiP, by
			self-truthfulness \emph{bis}.
	On the other hand,
		$\knows{a}{a}$ (knowledge of one's own name string) is valid in egFOLiP too.
	Hence $\knows{a}{a}\land\proves{a}{\phi}{a}{\community}$ is valid in egFOLiP.
	In conclusion,
		$\exists m(\knows{a}{m}\land\proves{m}{\phi}{a}{\community})$ is valid in egFOLiP.
\end{proof}
\noindent
Note that the stronger principle 
	\begin{center}
		``\,for all $b\in\community\cup\set{a}$,
			$\vdash_{\mathrm{iS4}}\Box_{a}^{\community}(\phi)\limp
				\Box_{b}^{\community\cup\set{a}}(\Box_{a}^{\community}(\phi))$\,''
	\end{center}
cannot be valid due to 
	$a$ knowing $M$ not necessarily implying that $b$ know $\sign{M}{a}$ (at the same location),
		at least not without $a$ communicating $\sign{M}{a}$ to $b$ (and thus possibly inducing a successor and thus different location),
			as the reader can verify herself by
				trying (but necessarily failing) to prove this stronger principle.
Further note that a more general interactive provability modality interpreted as 
	\begin{eqnarray*}
		\Box_{\pair{a}{b}}^{\community}(\phi) &\mapsto& \exists m(\knows{a}{m}\land\proves{m}{\phi}{b}{\community})
	\end{eqnarray*}
cannot satisfy the T-law, thus is not S4, and thus does not embed (i)IL.

In summary,
	the multi-agency and thus interactivity of iS4 cannot be strong, which 
		is a strong argument for stronger, more explicit modalities than those of iS4 
			in the sense of modalities with an added proof parameter, like those of LiP.
Even less strong in this sense must be interactive Intuitionistic Logic (iIL).

\subsection{Interactive Intuitionistic Logic (iIL)}\label{section:iIL}
As demonstrated in \cite[Remark~2.8]{LIiP:ACMTOCL},
	\emph{the interactive intuitionistic truths are 
		the truths of the communication medium,}
			the adversarial agent \emph{par excellence}, and 
				only of that.
No other agent can access the same interactive truths.
This is so, because
	only the communication medium has a view of 
		the communication network, thus also of itself, that 
			is sufficiently global to 
				induce the necessary and sufficient partial order 
					(rather than only an insufficient preorder) for 
						a Kripke-semantics of 
							interactive Intuitionistic Logic (iIL),
								the non-modal intuitionistic fragment of 
									the intuitionistic modal Logic of intuitionistic interactive Proofs (LIiP) 
										\cite[Remark~2.12]{LIiP:ACMTOCL}.
Recall that Intuitionistic Logic (IL) can only be defined by means of 
	a Kripke-semantics that does require a partial order.

In other words,
	iIL is isomorphic to IL.
\begin{fact}\label{fact:iILisoIL}
		$\mathrm{iIL}\cong\mathrm{IL}$
\end{fact}
\noindent
Thus all connectives of IL can be viewed as pertaining to the communication medium,
	the adversary, and only to that.

Hence,
	iIL embeds into iS4 like 
		IL does into S4 \emph{\`{a} la} G\"odel-McKinsey-Tarski, by
			prefixing intuitionistic subformulas with 
				an S4-necessity modality like
					$\Box_{a}^{\community}$\,.
In the other direction, 
	iS4 embeds into iIL like S4 does into IL \emph{\`{a} la} \cite{S4mapstoIL}.

\subsection{Typed interactive Combinatory Logic (TiCL)}\label{section:TiCL}
In order to define typed interactive Combinatory Logic (TiCL),
	we first need to define untyped interactive Combinatory Logic (iCL), just as with  
		typed and untyped non-interactive Combinatory Logic (TCL and CL, respectively) 
			\cite{LambdaCalculusAndCombinators}.
\begin{definition}[Interactive Combinatory Logic (iCL)]
Let $\mathcal{X}$ designate a countably infinite set of (term) variables $x$.
Then, 
	$$\mathcal{T}\ni T\bnfeq x\bnfor\Kcomb{a}\bnfor\Scomb{a}\bnfor\pair{T}{T}$$
shall designate the language of \emph{(pure) interactive combinatory terms} (iCL-terms) $T$ with 
	basic combinators $\Kcomb{a}$ and $\Scomb{a}$ for all $a\in\agents$, and 
	${\reduc{a}}\subseteq\mathcal{T}\times\mathcal{T}$ such that
		\begin{itemize}
			\item $\pair{\pair{\Kcomb{a}}{T}}{T'}\reduc{a}T$
			\item $\pair{\pair{\pair{\Scomb{a}}{T}}{T'}}{T''}\reduc{a}\pair{\pair{T}{T''}}{\pair{T'}{T''}}$
			\item if $T\reduc{a}T'$ then $\pair{T}{T''}\reduc{a}\pair{T'}{T''}$ and $\pair{T''}{T}\reduc{a}\pair{T''}{T'}$
		\end{itemize}
shall designate \emph{(local) reduction at agent $a$} on iCL-terms.

Further let 
	$\alpha\in\agents^{*}\cup\agents^{\omega}$ designate 
		a finite or infinite word over $\agents$ with 
	$\epsilon\in\agents^{*}$ the empty word.
Then, 
	${\reduc{\alpha}}\subseteq\mathcal{T}\times\mathcal{T}$ such that 
		\begin{eqnarray*}
			{\reduc{\epsilon}}&\defeq&\emptyset\\
			{\reduc{a\cdot\alpha'}}&\defeq&{\reduc{a}}\circ{\reduc{\alpha'}}
		\end{eqnarray*}
shall designate \emph{reduction (tout court)} on iCL-terms.

A (finite or an infinite) reduction $T\reduc{\alpha}T'$ such that 
	$\alpha\in\community^{*}\cup\community^{\omega}$ for some $\emptyset\neq\community\subseteq\agents$ shall be called 
	a \emph{local,} \emph{single-agent,} or \emph{non-interactive computation} 
		when $|\community|=1$, and 
	a \emph{global} or 
		\emph{multi-agent}---and only possibly but not necessarily an interactive (\cf Definition~\ref{definition:interactiveLambda})---\emph{computation} 
			when $|\community|>1$.
		
\emph{iCL-combinators} are defined to be variable-free iCL-terms.
\end{definition}
\noindent
The reason for our introduction of explicit parentheses (pairing) in iCL-syntax is 
	the resulting possibility of interactive combinators to be messages (\emph{mobile code}).
Further, like for CL, 
	\emph{applied variants of iCL} can be obtained by 
		including other terms, such as agent names and others, into the language of pure iCL.

In fact, each agent has her own local copy of CL bearing her name.
More precisely, the agent with name $a$ \emph{is,} or \emph{is identified with,} 
	the relation $\reduc{a}$ and thus indirectly a Turing-machine labelled $a$.
Hence each such copy, and thus iCL as a whole (containing all those CL-copies), 
	is Turing-powerful, as CL is.
\begin{fact}\label{fact:iCLTuringStrength}
	iCL is at least Turing-powerful.
\end{fact}

Further like CL, iCL has the important Church-Rosser Property.
\begin{theorem}[Church-Rosser Property of iCL]\label{theorem:ChurchRosser}
	Reduction on iCL-terms is confluent (at most finitely diverging):
		for all 
			$T,T_{1},T_{2}\in\mathcal{T}$ and 
			$\alpha_{1},\alpha_{2}\in\agents^{*}$, 
	\begin{multline*}
		\text{if $T\reduc{\alpha_{1}}T_{1}$ and $T\reduc{\alpha_{2}}T_{2}$ then}\\ 
		\text{there are $T'\in\mathcal{T}$ and $\alpha_{1}',\alpha_{2}'\in\agents^{*}$ such that}\\
		\text{$T_{1}\reduc{\alpha_{1}'}T'$ and 
				$T_{2}\reduc{\alpha_{2}'}T'$.}
	\end{multline*}
\end{theorem}
\begin{proof}
	See Appendix~\ref{theorem:ChurchRosser}.
\end{proof}

Based on iCL, an \emph{interactive (agent-centric) lambda-operator} can be defined.
\begin{definition}[The interactive lambda-operator]\label{definition:interactiveLambda}
	Let $a\in\agents$, $x\in\mathcal{X}$, and $T\in\mathcal{T}$.
	Then in analogy with \cite{LambdaCalculusAndCombinators}, 
		 we define $\lambda_{a}x.T$, 
			the \emph{interactive lambda-operator in $x$ for $a$ and with scope $T$}, inductively on 
				the structure of iCL-terms such that:
		$$\begin{array}{l@{}l@{\ \ }l@{\ \ }l}
			\lambda_{a}x.&T &\defeq& \pair{\Kcomb{a}}{T}\qquad\text{if $x$ does not occur in $T$}\\[\jot]
			\lambda_{a}x.&x &\defeq& \pair{\pair{\Scomb{a}}{\Kcomb{a}}}{\Kcomb{a}}\qquad\text{($a$'s identity combinator)}\\[2\jot]
			\lambda_{a}x.&\pair{T}{x} &\defeq& T\qquad\text{if $x$ does not occur in $T$}\\[\jot]
			\lambda_{a}x.&\pair{T}{T'} &\defeq& \pair{\pair{\Scomb{a}}{\lambda_{a}x.T}}{\lambda_{a}x.T'}\qquad\text{if $x$ occurs in $\pair{T}{T'}$}
		\end{array}$$
	For $a\neq b$,
		a macro $\lambda_{a}x.T$ with 
			another macro $\lambda_{b}y.T'$ occurring in $T$ and $x$ occurring in $T'$ 
				shall be called a \emph{communication channel from agent $a$ to agent $b$} and
		its reduction an \emph{interactive} and thus a \emph{global and multi-agent computation}.
\end{definition}
\noindent
(Interactivity implies multi-agency, but not necessarily vice versa, because 
	multiple agents can compute locally for themselves, without communicating with each other.)
For example,
	the macro $\lambda_{a}x.\lambda_{b}y.\pair{y}{x}$ for the interactive program of 
		$a$ being able to send a message (through her port) $x$ to $b$ and thus 
		$b$ being able to receive $x$ through her port $y$, 
		expanding to $\lambda_{a}x.\pair{\pair{\Scomb{b}}{\lambda_{b}y.y}}{\lambda_{b}y.x}$), in turn
			expanding to $\lambda_{a}x.\pair{\pair{\Scomb{b}}{\pair{\pair{\Scomb{b}}{\Kcomb{b}}}{\Kcomb{b}}}}{\pair{\Kcomb{b}}{x}}$), \etc, 
				is such a channel.
From this, it should be obvious that 
	our interactive lambda-operator is useful syntactic sugar for 
		the agent- and channel-oriented programming of communication protocols.
Of course, we could extend our definition of reduction to iCL-terms with lambda-operators, and
	define each communication reduction to generate a concrete input-history event in the sense of
		Definition~\ref{definition:SemanticIngredients}.

Note that 
	due to its Turing-powerful interactivity,
		iCL can be appreciated as a 
			bound-variable free, grammatically much simpler alternative to 
				the popular Pi-calculus \cite{PiCalculus} and its derivatives (including its applied variants).

We postpone the further independent study of iCL to future work and
	continue to present the connection of iCL to iIL and LiP via TiCL.

\begin{definition}[Typed iCL (TiCL)]\label{definition:TiCL}
	Let $\mathcal{P}$ designate a set of type variables 
		(for example the set of propositional variables in Definition~\ref{definition:LiPLanguage}).
	Then,
		$P\in\mathcal{P}$ designates a type and
		if $\varphi$ and $\varphi'$ designate types 
		then so does $\varphi\limp\varphi'$, and
	\begin{itemize}
		\item $\Gamma\cup\set{x:\varphi}\vdash_{\mathrm{TiCL}}x:\varphi$
		\item $\Gamma\vdash_{\mathrm{TiCL}}\mathtt{K}_{a}:(\varphi\limp(\varphi'\limp\varphi))$
				\quad($a\in\agents$)
		\item $\Gamma\vdash_{\mathrm{TiCL}}\mathtt{S}_{a}:((\varphi\limp(\varphi'\limp\varphi''))\limp((\varphi\limp\varphi')\limp(\varphi\limp\varphi'')))$
				\quad($a\in\agents$)
		\item if $\Gamma\vdash_{\mathrm{TiCL}}T:(\varphi\limp\varphi')$ and
				$\Gamma\vdash_{\mathrm{TiCL}}T':\varphi$ then
				$\Gamma\vdash_{\mathrm{TiCL}}\pair{T}{T'}:\varphi'$
	\end{itemize}
	shall designate Typed iCL (TiCL).
\end{definition}
\noindent
TiCL is defined 
	in Curry's style of typing, where
		type formulas label proof terms, which  
	is more general than Church's style, where 
		type formulas are parts of the proof terms themselves.
For more details on these two typing styles, see \cite{LambdaCalculusAndCombinators}.

\begin{theorem}[Termination and Type Invariance of TiCL]\label{theorem:TerminationSubjectReduction}
	In TiCL, 
		all computations local or global are finite (terminating), and 
		preserve their typing (types are computational invariants, that is, 
			they are invariant under reduction):
		\begin{multline*}
			\text{for all $T,T'\in\mathcal{T}$,}\\ 
			\text{if 
						$T\reduc{\alpha}T'$ and 
						$\Gamma\vdash_{\mathrm{TiCL}}T:\varphi$
				  then $\alpha\in\agents^{*}$ and $\Gamma\vdash_{\mathrm{TiCL}}T':\varphi$.}
		\end{multline*}
\end{theorem}
\begin{proof}
	See Appendix~\ref{theorem:TerminationSubjectReduction}.
\end{proof}
\noindent
As usual, 
	the \emph{termination} of computations (executing programs) corresponds to 
		the \emph{totality} of the corresponding computed functions.
That types are computational invariants is natural, because 
	(1) types are intuitionistic formulas, which
			are (forward) invariant under computations (see for example \cite{LIiP:ACMTOCL}), and 
	(2) types and interactive programs are interconnected via 
			the following Curry-Howard isomorphism
				(formulation adapted from \cite{LecturesOnTheCurryHowardIsomorphism}).
\begin{proposition}[Interactive Curry-Howard Isomorphism]\label{proposition:iCH}
	Let $\mathcal{N}\subsetneq\mathbb{N}$.
	Then, 
	\begin{enumerate}
		\item if $\set{x_{i}:\varphi_{i}}_{i\in\mathcal{N}}\vdash_{\mathrm{TiCL}}T:\varphi$
				then $\set{\varphi_{i}}_{i\in\mathcal{N}}\vdash_{\mathrm{iIL}}\varphi$\,, and
		\item if $\set{\varphi_{i}}_{i\in\mathcal{N}}\vdash_{\mathrm{iIL}}\varphi$
				then there is $T\in\mathcal{T}$ such that 
					$\set{x_{i}:\varphi_{i}}_{i\in\mathcal{N}}\vdash_{\mathrm{TiCL}}T:\varphi$\,.
	\end{enumerate}
\end{proposition}
\begin{proof}
	Immediate:
		for (1) by induction on the derivation of 
					$\set{x_{i}:\varphi_{i}}_{i\in\mathcal{N}}\vdash_{\mathrm{TiCL}}T:\varphi$, and 
		for (2) by induction on the derivation of 
					$\set{\varphi_{i}}_{i\in\mathcal{N}}\vdash_{\mathrm{iIL}}\varphi$.
	Just inspect the (corresponding) axioms and definitions of (i)IL and TiCL, respectively.
\end{proof}

The following proposition asserts that 
	the modal logic LiP can be viewed as 
		including the type system that is categorical TiCL, where 
			categorical\footnote{For more details on categoricity, see for example \cite{TheFunctionalInterpretationOfLogicalDeduction}.} TiCL is TiCL without variables.
Of course, adding variables to LiP would be trivial, and
	so LiP with variables could be viewed as including (full) TiCL!
\begin{proposition}[LiP as including categorical TiCL] For all $a,b\in\agents:$ 
	\begin{enumerate}
		\item $\LiPded\proves{\mathtt{K}_{b}}{(\phi\limp(\phi'\limp\phi))}{a}{\community}$
		\item $\LiPded\proves{\mathtt{S}_{b}}{((\phi\limp(\phi'\limp\phi''))\limp((\phi\limp\phi')\limp(\phi\limp\phi'')))}{a}{\community}$
		\item $\set{\proves{M}{(\phi\limp\phi')}{a}{\community},\proves{M'}{\phi}{a}{\community}}\LiPded
				\proves{\pair{M}{M'}}{\phi'}{a}{\community}$
	\end{enumerate}
\end{proposition}
\begin{proof}
	(1) and (2) hold by necessitation (N), and
	(3) by the generalised Kripke-law (GK).
\end{proof}
\noindent
Note however that the converse view ``categorical TiCL as including LiP'' is incorrect.
Categorical TiCL can only be viewed as a proper fragment of LiP, since
	type systems can only talk about atomic type statements.
		(Type systems have no logical connectives that could form compound type statements.)
So, modal logics of combinators like LiP are more general than simple-type systems.

\begin{theorem}[Fundamental Isomorphism of Interactive Computation]\label{theorem:ITIC}
	$$\mathrm{TiCL}\cong\mathrm{TCL}$$
\end{theorem}
\begin{proof}
	By 
		the isomorphisms 
			$\mathrm{TiCL}\cong\mathrm{iIL}$ (Proposition~\ref{proposition:iCH}),
			$\mathrm{iIL}\cong\mathrm{IL}$ (Fact~\ref{fact:iILisoIL}), and
			$\mathrm{IL}\cong\mathrm{TCL}$ (Curry-Howard), and 
		the transitivity property of isomorphy.
\end{proof}
\noindent
The simple typing of iCL can of course be extended to more complex and thus powerful typings,
	like the simple typing of CL has been.
Thus the question of whether or not interactive computation is strictly more powerful (recall Fact~\ref{fact:iCLTuringStrength}) than 
	non-interactive computation \cite{InteractiveComputation} may well be settled in stages of typings capturing 
		levels of gradually increasing computational strength.
\begin{corollary}[Equipotency of Simply-Typed Interactivity and Non-Interactivity]\label{corollary:Equipotency}
	Interactive and non-interactive computation
		are equipotent at the level of simple types (as defined by TiCL and TCL, respectively), and 
			capture the so-called extended polynomial functions over the natural numbers 
				\cite{LambdaCalculusAndCombinators,DefinableFunctionsSimplyTypedlambdaCalculus}.
\end{corollary}
\begin{proof}
	By Theorem~\ref{theorem:ITIC} and \cite{DefinableFunctionsSimplyTypedlambdaCalculus}, respectively.
\end{proof}
\noindent
This result as well as the strategy of settling computability and complexity questions by type stages could have 
	a profound impact on communication complexity and distributed (and thus concurrent and parallel) computation research.
Ultimately, they enable a comparative approach to the Church-Turing Thesis.

\section{Related work}\label{section:RelationToLP}
In this section,
	we relate 
		our Logic of interactive Proofs (LiP) 
			to Art\"{e}mov's Logic of Proofs (LP) \cite{LP} and 
			to a generalised variant thereof, namely
				his Symmetric Logic of Proofs (SLP) \cite{SymmetricLP}.
We also relate LiP to two extensions of LP with multi-agent character, namely 
	Yavorskaya's LP$^{2}$ \cite{YavorskayaSidon} and Renne's UL \cite{Renne2}.
The general aim of this section is to give a detailed description of 
	crucial design decisions for interactive and non-interactive systems on
		the example of related works.
Essentially, we argue that,  
	first, LP and LiP can be related but 
		have typically different (not always) but complementary scopes, namely 
			non-interactive computation and universal truths, and 
			interactive computation and local truths, respectively; and, 
	second, LiP improves LP-like systems with respect to interactivity.
That LP and LiP can indeed be related is evidenced   
		to some extent by Theorem~\ref{theorem:EmbeddingLPintoLiP} and  
		proved by the example following it, which 
			happens to be formalisable in both LP and LiP.
That we discuss multi-agent extensions of LP is justified by
	the fact that 
		LP$^{2}$ and UL are intended to be interactive but
			inherit the lack of message-passing interactivity from LP.
As a matter of fact, the example with \emph{signing} is formalisable only in LiP.

\subsection{Concepts}\label{section:Concepts}
In (S)LP, 
	$p{:}F$ stands for an atomic concept.
Whereas in LiP, 
	$\proves{M}{\phi}{a}{\community}$ stands for a compound concept 
		analysable into epistemic constituents (\cf Section~\ref{section:EpistemicExplication}), 
			\emph{nota bene} thanks to a constructive semantics defined
				in terms of the proof terms themselves (\cf Page~\pageref{page:ProofAccessibility}).
In that,
	our construction is reminiscent of the canonical-model construction, which
		like ours is a constructive semantics defined in terms of syntax, but 
		unlike ours not in terms of terms but in terms of formulas (\cf Appendix~\ref{appendix:Proofs}).

\subsubsection{Interactivity}\label{section:Interactivity} 
(S)LP proofs are non-interactive, whereas 
	LiP proofs are interactive (knowledge-inducing).
(S)LP proofs are non-interactive also due to (S)LP's reflection axiom, which
	stipulates that provability imply truth\footnote{(S)LP (and LiP) has a semantics, 
		so we may use the word `truth' here.}.
However, 
	in a truly interactive setting, 
		(S)LP's reflection axiom is unsound.
By a truly interactive setting, 
	we mean a multi-agent distributed system 
		where not all proofs are known by all agents, that is,  
			a setting with a non-trivial distribution of information 
				in the sense of Scott (\cf Proposition~\ref{proposition:PropertiesOfDerivability}), 
					in which $\not\models\knows{a}{M}$.
In other words, 
	in truly interactive settings,
		agents are not omniscient with respect to messages. Otherwise, why communicate?
As proof, consider the following, self-referential counter-example: 
	$\models\proves{M}{(\knows{a}{M})}{a}{\emptyset}$ (self-knowledge) but 
		$\not\models\proves{M}{(\knows{a}{M})}{a}{\emptyset}\limp\knows{a}{M}$.
In truly interactive settings, 
	there being a proof does not imply knowledge of that proof.
When 
	an agent $a$ does not know the proof and 
	the agent cannot generate the proof \emph{ex nihilo} herself by guessing it,
		only \emph{communication} from a peer, who thus acts as an oracle, can entail
			the knowledge of the proof with $a$.
In sum, 
	\textbf{\emph{provability and truth are necessarily concomitant in the non-interactive setting, whereas 
	in interactive settings they are not necessarily so.}}

\subsubsection{Proof terms}\label{section:ProofTerms}
(S)LP needs three proof-term constructors, namely sum, application, and proof checker.
Whereas LiP only needs two, namely pairing and signing, but 
	as opposed to LP can even handle interaction (with signing).
Incidentally, 
	G\"odel conjectured that two proof-term constructors were sufficient for proofs \cite{ArtemovBSL}.
In LiP,
	pairing plays a pair of roles, namely the two roles played by sum and application in LP, and
	thanks to Fact~\ref{fact:CommonProofKnowledge} the agents themselves within their own communities 
		may---not a term constructor like `!' in (S)LP must---play the proof-checker role!
In sum,
	first, LiP-agents play a pair of roles, namely the two roles of proof as well as signature checker, and, 
	second, signatures can be conceived as proof-checker-apposed, communally verifiable seals of check.

\subsubsection{Formulas}
(S)LP's proof modality `${:}$' has no parameters, whereas LiP's `$\proves{}{}{a}{\community}$' has two.
The advantage of LiP's parametric modality is agent-centricity and thus greater generality.
As a nice side effect, 
	LiP's proof terms have neutral elements.

\subsection{Laws}
\subsubsection{Structural laws (\cf Theorem~\ref{theorem:SomeUsefulDeducibleStructuralLaws})}\label{section:StructuralLaws}
In LP,
	the proof-sum operation `$+$' is 
		neither commutative 
		nor idempotent, but in SLP, it is both, like `$\pair{\cdot}{\cdot}$' in LiP.
In (S)LP,
	`$+$' has no neutral element, whereas
in LiP the corresponding `$\pair{\cdot}{\cdot}$' has.
As said previously,
	LiP's `$\pair{\cdot}{\cdot}$' can simulate not only LP's proof sum but also (S)LP's proof application.
However,
	LiP's `$\pair{\cdot}{\cdot}$' \emph{cannot} simulate \emph{S}LP's sum.
To see why,
	consider that 
		if (S)LP were defined analogously to LiP by means of a separate term theory using 
			atomic propositions `$\knows{\negthinspace}{p}$' (for ``$p$ is known'') and 
			an analog of epistemic antitonicity 
		then the structural modal laws of (S)LP could be (partially) generated from the structural term laws,  
			analogously to LiP.
			
\paragraph{LP} From the term axiom schema 
$$\knows{\negthinspace}{p{+}q}\limp(\knows{\negthinspace}{p}\land\knows{\negthinspace}{q})$$
	generate the corresponding characteristic law $$((p{:}F)\lor q{:}F)\limp (p{+}q){:}F.$$
							
\paragraph{SLP}
		\begin{enumerate}
			\item From the term axiom schema
					$$\knows{\negthinspace}{p{+}q}\limp(\knows{\negthinspace}{p}\land\knows{\negthinspace}{q})$$
				generate the corresponding characteristic law $$((p{:}F)\lor q{:}F)\limp (p{+}q){:}F.$$
			\item Add the axiom schema $$((p{+}q){:}F)\limp((p{:}F)\lor q{:}F),$$ and
					(disregarding SLP's proof application)  
					obtain the characteristic law $$((p{:}F)\lor q{:}F)\lequiv (p{+}q){:}F$$ of SLP's sum,
							which subsumes LP's sum law.
		\end{enumerate}
				However in the case of LiP, 
						$$\not\models(\proves{\pair{M}{M'}}{\phi}{a}{\community})\limp
								((\proves{M}{\phi}{a}{\community})\lor\proves{M'}{\phi}{a}{\community}),$$ due to
									the obvious counter-example 
										(recall that $\models\proves{\pair{M}{M'}}{\knows{a}{\pair{M}{M'}}}{a}{\emptyset}$)
									$$\not\models(\proves{\pair{M}{M'}}{\knows{a}{\pair{M}{M'}}}{a}{\emptyset})\limp
								((\proves{M}{\knows{a}{\pair{M}{M'}}}{a}{\emptyset})\lor\proves{M'}{\knows{a}{\pair{M}{M'}}}{a}{\emptyset}).$$				
				That is, it is not generally true that 
					single projections prove pair knowledge.

\subsubsection{Logical laws (\cf Theorem~\ref{theorem:SomeUsefulDeducibleLogicalLaws})}\label{section:LogicalLaws}
(S)LP does not obey 
	Kripke's law K, 
	the law of necessitation, nor 
	a law of modal idempotency.
Whereas LiP does obey 
	K as well as the generalised Kripke-law GK,
	necessitation, and 
	the law of modal idempotency.

Note that for resource-bounded agents, 
	restricting the (resource-unbounded) pairing axiom would be desirable in order 
		to prevent the (resource-unbounded) K from being deducible in LiP.
Incidentally, (S)LP can be understood as being reconstructed 
	only from the (resource-bounded) \emph{un}pairing axiom and
	not from the (resource-unbounded) pairing axiom (\cf Section~\ref{section:StructuralLaws}).

The justification for choosing (plain) necessitation instead of LP's constant specification for LiP is that
	in the interactive setting,
		validities, and thus \emph{a fortiori} tautologies 
			(in the strict sense of validities of the propositional fragment), 
				are in some sense trivialities.
To see why,
	recall from Definition~\ref{definition:TruthValidity} that
		validities are true in \emph{all} pointed models, and thus
			not worth being communicated from one point to another in a given model, \eg 
				by means of specific interactive proofs.
(Nothing is logically more embarrassing than  
	talking in tautologies.)
Therefore, 
	validities deserve \emph{arbitrary} messages as proof.
What is worth being communicated are
	truths weaker than validities, namely 
		local truths in the sense of Definition~\ref{definition:TruthValidity}, which
			do not hold universally (\cf Table~\ref{table:InterestingTruths}).
\begin{table}[t]\centering
	\caption{Interesting truths}
	\smallskip
	\begin{tabular}{@{}|c|c|@{}}
		\hline
		\textbf{Computation} & \textbf{Truth}\\
		\hline
		\hline
		interactive & local\\
		\hline
		non-interactive & universal\\
		\hline
	\end{tabular}
	\label{table:InterestingTruths}
\end{table}
Note that
	our choice is not forced but free: 	
		we could have chosen constant specification for LiP too 
			(\eg ``$\LiPded\proves{a}{\phi}{a}{\community}$, for $\phi\in\Gamma_{1}$'') and thus
			kept a closer relationship between LP and LiP, but that would have, first,  
				put unnecessarily strong proof obligations on validities
						as far as interactivity is concerned, as explained; and,  
				second, unfaithfully modelled resource-unbounded interacting agents, which 
						already know all universal truths or validities, 
							though of course not all local truths, which 
								is the whole point of interacting with each other!
In sum, while LP weakens necessitation,
	LiP weakens truthfulness (without risking falsehood) by conditioning it on proof knowledge 
		(\cf Section~\ref{section:Interactivity}).

(S)LP does not obey the law of modal idempotency, because
		it does not have agents that could act as proof checkers and 
		thus needs a term constructor for proof-checking.
Whereas LiP does obey modal idempotency, because
	LiP does have agents that can act as proof checkers (\cf Section~\ref{section:ProofTerms}) and
		thus does not need a term constructor for proof-checking.
Observe that 
	modal idempotency is deducible in LiP due to 
		the law of self-signing elimination, which in turn 
			is deducible in LiP due to 
		the axiom of personal signature synthesis
				(\cf Section~\ref{appendix:LogicalProofs}).
Note that for resource-bounded agents, 
	restricting (resource-unbounded) personal signature synthesis could be desirable in order   
		to prevent (resource-unbounded) modal idempotency from being deducible in LiP.
Incidentally, (S)LP can be understood as being reconstructed 
	from no term axioms involving the proof checker `$!$' (\cf Section~\ref{section:StructuralLaws}).

\subsubsection{Meta-logical properties}
(S)LP is not a normal modal logic, 
	because (S)LP does not obey Kripke's law. 
Whereas LiP is a normal logic (\cf Fact~\ref{fact:Normality}).
LP is in $\Sigma_{2}^{p}$
	\cite{Kuznets}, but the decidability and thus complexity of SLP is unknown \cite{SymmetricLP}.
A lower complexity bound for LiP is EXPTIME, which follows 
		from the complexity of the logic of common knowledge, which is EXPTIMEcomplete 
			\cite{ComplexityKnowledgeBelief}, and
		from the fact that 
			the concrete accessibility relation $\pAccess{M}{a}{\community}$ for LiP requires 
				$\preorder{\community\cup\set{a}}$, which contains
				the one for common knowledge $\indist{\community\cup\set{a}}{}{}$ (\cf Page~\pageref{page:CommonKnowledge}).
As mentioned at the end of Section~\ref{section:ImportantProperties},
	complexity and decidability depend on term axioms.

\subsection{Formal relation}\label{section:LPmapstoLiP}
In order to establish a formal relation between LP and LiP,
	we consider LiP 
		over a singleton society and
		over the term forms suggested on Page~\pageref{page:TermForms}.
So without loss of generality let $\agents=\set{a}$ and $\messages\setminus\set{\Kcomb{a},\Scomb{a}}$.
Further, 
	fix LP's set of specification constants to consist of $\set{a}$, and
	consider the mapping $h$ over LP-formulas that maps LP's 
	\begin{itemize}
		\item proof-sum `$+$' and proof-application `$\cdot$' to LiP's proof-pair constructor `$\pair{\cdot}{\cdot}$'
		\item proof checker `$!$' to LiP's proof-signature constructor `$\sign{\cdot}{a}$'
		\item proof modality `${:}$' to LiP's proof modality `$\proves{}{}{a}{\emptyset}$'.
	\end{itemize}

\begin{lemma}[Admissibility of LP-laws for LiP]\label{lemma:LPLawsAdmissibility}
When $\agents=\set{a}:$
	\begin{enumerate}\setcounter{enumi}{-1}
		\item $\LiPded\varphi$, for any axiom $\varphi$ of classical propositional logic
		\item $\LiPded((\proves{M}{\phi}{a}{\emptyset})\lor\proves{M'}{\phi}{a}{\emptyset})\limp\proves{\pair{M}{M'}}{\phi}{a}{\emptyset}$
		\item $\LiPded(\proves{M}{(\phi\limp\phi')}{a}{\emptyset})\limp
						((\proves{M'}{\phi}{a}{\emptyset})\limp\proves{\pair{M}{M'}}{\phi'}{a}{\emptyset})$
		\item $\LiPded (\proves{M}{\phi}{a}{\emptyset})\limp\phi$
		\item $\LiPded (\proves{M}{\phi}{a}{\emptyset})\limp
						\proves{\sign{M}{a}}{(\proves{M}{\phi}{a}{\emptyset})}{a}{\emptyset}$
		\item $\set{\phi\limp\phi',\phi'}\LiPded\phi'$
		\item $\LiPded\proves{a}{\phi}{a}{\emptyset}$, for any formula $\phi$ 
					for which $\LiPded\phi$ in Item~0--4.
\end{enumerate}
\end{lemma}
\begin{proof}
	(0) holds by definition of LiP.
	For the rest, 
		set $\community=\emptyset$.
	Then (1) is LiP's law of proof extension (\cf Theorem~\ref{theorem:SomeUsefulDeducibleStructuralLaws}.11);
	(2) is LiP's generalised Kripke-law;
	(3) is, given that $\agents=\set{a}$ and $\messages\setminus\set{\Kcomb{a},\Scomb{a}}$, 
		LiP's law of truthfulness 
			(\cf Theorem~\ref{theorem:SomeUsefulDeducibleLogicalLaws}.28.a);
	(4) is LiP's laws of peer review; 
	(5) holds by definition of LiP; and 
	(6) follows by particularising LiP-necessitation.
\end{proof}
	
\begin{theorem}[Homomorphism from LP into LiP]\label{theorem:EmbeddingLPintoLiP}
	For $\agents$ a singleton, 
		for all LP-formulas $F$,
			$$\text{if $\LPded F$ then $\LiPded h(F)$.}$$
\end{theorem}
\begin{proof}
	By the admissibility of LP-axioms and -rules for LiP (\cf Lemma~\ref{lemma:LPLawsAdmissibility}).
\end{proof}
\noindent
However the converse is not true, and thus $h$ is only a homomorphism and not an embedding.
As counter example consider Kripke's law, which holds in LiP, but
		does not hold in LP (\cf Section~\ref{section:LogicalLaws}).
In sum,
	while plain propositional logic can be viewed as a modal logic interpreted over a singleton universe,  
	LP can be viewed only to a limited extent as LiP over a singleton society.
The extent is limited because LiP does not mathematically contain LP, as   
	LP does not embed (injectively homomorph) but 
		only non-injectively homomorph into LiP, which
			we believe reflects the essential difference between 
				their scopes.
We stress that 
	LP and LiP have typically different, complementary scopes, namely 
		non-interactive computation and universal truths, and
		interactive computation and local truths, respectively.
Nevertheless:
	\begin{enumerate}
		\item LP and LiP have a non-empty intersection, as 
				the following example proves, which 
					happens to be formalisable in both LP \cite{JustificationLogic} and LiP, which
						is also why we have chosen it (comparative explanatory power).
		\item LiP is richer than S4, since LiP 
				generalises S4 with agent centricity and
					refines S4 with explicit, transmittable proofs.
	\end{enumerate}
The example involves two elementary formal proofs, which for clarity 
	we present in the style of Frederic Fitch, 
		justified by the following definition and facts.
So its purpose is not mathematical (to reason about) but elementary (to use) logic.
The example makes all things used as explicit and thus as formal as possible.
So the mathematical (meta-)logician may want to simply skip it.
\begin{definition}[Local hypotheses]\label{definition:LocalHypotheses}
Let 
	$\Lambda\subsetneq\pFormulas$ such that $\Lambda$ is finite, and 
\begin{eqnarray*}
	\Gamma;\Lambda\LiPded\phi 
		&\text{:iff}& \Gamma\LiPded(\bigwedge\Lambda)\limp\phi
\end{eqnarray*}
	(\cf Proposition~\ref{proposition:Hilbert}), where
		$\Lambda$ is understood as a finite set of \emph{local} hypotheses.
\end{definition}
\begin{fact}[\emph{A fortiori} true, persistently provable and known as true] 
	$$\Gamma,\phi;\Lambda\LiPded\phi\land(\proves{M}{\phi}{a}{\community})\land\proves{a}{\phi}{a}{\emptyset},$$
	where $\Gamma,\phi$ means $\Gamma\cup\set{\phi}$.
\end{fact}
\begin{proof} 
	From Proposition~\ref{proposition:Hilbert} 
	by
	 	the fact that $\LiPded\phi$ implies $\LiPded\phi$,
		necessitation, and
		self-truthfulness \emph{bis} for the above case 
			$\phi$, 
			$\proves{M}{\phi}{a}{\community}$, and 
			$\proves{a}{\phi}{a}{\emptyset}$,
				respectively.
\end{proof}
\noindent
Recall from Section~\ref{section:EpistemicExplication}, 
	that $\LiPded\proves{a}{\phi}{a}{\emptyset}$ can be read
		as ``$a$ persistently knows that $\phi$ is true''
			(unless interpreted defeasibly, \cf 
				Proposition~\ref{proposition:Instancy}).
\begin{fact}[Fitting-style deduction ``theorems'' \cite{ModalProofTheory}]  
	$$\mathrm{LDT}\ \begin{array}{@{}c@{}}
		\Gamma;\Lambda,\phi\LiPded\phi'\\
		\hline
		\hline
		\Gamma;\Lambda\LiPded\phi\limp\phi'
	\end{array}\qquad
	\mathrm{MP}\ \begin{array}{@{}c@{}}
		\Gamma;\Lambda,\phi\LiPded\phi'\\
		\hline
		\Gamma,\phi;\Lambda\LiPded\phi'
	\end{array}$$
	Here,
		``LDT'' abbreviates ``Local Deduction Theorem'',
		``MP'' abbreviates ``\emph{modus ponens}'',
		$\Lambda,\phi$ means $\Lambda\cup\set{\phi}$,  
		the double horizontal bar means ``if and only if'', and 
		the simple horizontal bar reads ``if \ldots then \ldots'' from top to bottom.
\end{fact}
\begin{proof}
	The validity of 
		the LDT rule schema is warranted by Definition~\ref{definition:LocalHypotheses}, and
		the one of the MP rule schema by the \emph{modus ponens} rule schema of LiP.
\end{proof}
\newcommand{\Smith}{\mathtt{Smith}}
\newcommand{\Jones}{\mathtt{Jones}}
\newcommand{\job}{\mathsf{job}}
\newcommand{\ten}{\mathsf{10}}
\newcommand{\HR}{\mathtt{HR}}
\noindent
Following \cite{JustificationLogic},
	we now present the more difficult Case I of Gettier's Case I and II,
		which according to \cite{JustificationLogic} 
			``were supposed to provide examples of justified true beliefs 
				which should not be considered knowledge.''
\begin{example}[Gettier, from \cite{JustificationLogic}]
	Suppose that Smith and Jones have applied for a certain job.
	And suppose that Smith has strong evidence for the following conjunctive proposition:
	(d) Jones is the man who will get the job, and
		Jones has ten coins in his pocket.
	Proposition (d) entails:
	(e) The man who will get the job has ten coins in his pocket.
	Let us suppose that Smith sees the entailment from (d) to (e), and
	accepts (e) on the grounds of (d), for which he has strong evidence.
	In this case, Smith is clearly justified in believing that 
	(e) is true.
	But imagine, further, that unknown to Smith, he himself, not Jones,
	will get the job. And also, unknown to Smith, he himself has ten coins in his pocket.
	Then, all of the following are true:
	1) (e) is true, 
	2) Smith believes that (e) is true.
		But it is equally clear that Smith does not know that (e) is true.
\end{example}
\noindent
Interpreting ``strong evidence'' in Gettier's example as ``proof'' in our sense,
Gettier's Case I can be formalised in LiP as follows.
Let:
\begin{itemize}
	\item $a\in\agents\defeq\set{\Smith,\Jones};$
	\item for all $a\in\agents$, $\job(a),\ten(a)\in\mathcal{P}.$
\end{itemize}
Then Gettier's assumptions stated in his example are contradictory,
	as asserted by Proposition~\ref{proposition:GettierExample} and proved by
		jointly Lemma~\ref{lemma:Gettier} and the proof in Table~\ref{table:GettierExample}.
Lemma~\ref{lemma:Gettier} corresponds to the assertion that 
	(d) entails (e) in Gettier's example.
\begin{lemma}[Gettier example]\label{lemma:Gettier}
	$$\begin{array}{@{}l@{}}
		\set{(\job(\Smith)\land\job(\Jones))\limp\false};\emptyset\LiPded\\
		(\job(\Jones)\land\ten(\Jones))\limp\bigwedge_{a\in\agents}(\job(a)\limp\ten(a))
	\end{array}$$
\end{lemma}
\begin{proof}\ 
	\begin{enumerate}
		\item\quad$\LiPded(\job(\Smith)\land\job(\Jones))\limp\false$\hfill global hypothesis
		\item\begin{enumerate}
				\item$\job(\Jones)\land\ten(\Jones)$\hfill local hypothesis 
				\item$\job(\Jones)\limp\ten(\Jones)$\hfill 2.a, PL
				\item$(\job(\Smith)\land\job(\Jones))\limp\false$\hfill 1, \emph{a fortiori}
				\item$\neg\job(\Smith)$\hfill 2.a, 2.c, PL
				\item$\job(\Smith)\limp\ten(\Smith)$\hfill 2.d, PL
				\item$\underbrace{(\job(\Jones)\limp\ten(\Jones))\land(\job(\Smith)\limp\ten(\Smith))}_{%
						\text{the man who will get the job has 10 coins in his pocket}}$\hfill 2.b, 2.e, PL
			\end{enumerate}
		\item\quad$\LiPded(\job(\Jones)\land\ten(\Jones))\limp
					\bigwedge_{a\in\agents}(\job(a)\limp\ten(a))$\hfill 2.a--2.f, LDT
		\item if $\LiPded(\job(\Smith)\land\job(\Jones))\limp\false$\hfill 1--3, PL\\  
				then $\LiPded(\job(\Jones)\land\ten(\Jones))\limp
					\bigwedge_{a\in\agents}(\job(a)\limp\ten(a))$
		\item $\set{(\job(\Smith)\land\job(\Jones))\limp\false};\emptyset\LiPded\\(\job(\Jones)\land\ten(\Jones))\limp
					\bigwedge_{a\in\agents}(\job(a)\limp\ten(a))$\hfill 4, definition.
	\end{enumerate}
\end{proof}

\begin{proposition}[Gettier example]\label{proposition:GettierExample}
	$$\begin{array}{@{}l@{}}
		\set{(\job(\Smith)\land\job(\Jones))\limp\false};\emptyset\LiPded\\
			(\knows{\Smith}{M}\land\proves{M}{(\job(\Jones)\land\ten(\Jones))}{\Smith}{\emptyset})\limp\\
				((\job(\Smith)\land\ten(\Smith))\limp\false)
		\end{array}$$
\end{proposition}
\begin{proof}
	See Table~\ref{table:GettierExample}.
\end{proof}

\begin{sidewaystable}
%\vspace{35\baselineskip}
\caption{Gettier example (proof of Proposition~\ref{proposition:GettierExample})}
	\begin{enumerate}
		\item\quad$\LiPded(\job(\Smith)\land\job(\Jones))\limp\false$\hfill global hypothesis
		\item\quad$\LiPded(\job(\Jones)\land\ten(\Jones))\limp\bigwedge_{a\in\agents}(\job(a)\limp\ten(a))$\hfill 1, Lemma~\ref{lemma:Gettier}, PL
		\item\quad$\LiPded\proves{\Smith}{((\job(\Jones)\land\ten(\Jones))\limp\bigwedge_{a\in\agents}(\job(a)\limp\ten(a)))}{\Smith}{\emptyset}$\hfill 2, N
		\item\quad$\LiPded\knows{\Smith}{\Smith}$\hfill knowledge of own's own name string
		\item\begin{enumerate}
				\item$\knows{\Smith}{M}\land\proves{M}{(\job(\Jones)\land\ten(\Jones))}{\Smith}{\emptyset}$\hfill local hypothesis
				\item$\proves{M}{(\job(\Jones)\land\ten(\Jones))}{\Smith}{\emptyset}$\hfill 5.a, PL
		\item$\proves{\Smith}{((\job(\Jones)\land\ten(\Jones))\limp\bigwedge_{a\in\agents}(\job(a)\limp\ten(a)))}{\Smith}{\emptyset}$\hfill 3, \emph{a fortiori}
				\item$\proves{\pair{\Smith}{M}}{(\bigwedge_{a\in\agents}(\job(a)\limp\ten(a)))}{\Smith}{\emptyset}$\hfill 5.b, 5.c, GK
				\item$\knows{\Smith}{\Smith}$\hfill 4, \emph{a fortiori}
				\item$\knows{\Smith}{M}$\hfill 5.a, PL
				\item$\knows{\Smith}{\Smith}\land\knows{\Smith}{M}$\hfill 5.e, 5.f, PL
				\item$\knows{\Smith}{\pair{\Smith}{M}}$\hfill 5.g, pairing
				\item$\bigwedge_{a\in\agents}(\job(a)\limp\ten(a))$\hfill 5.d, 5.h, epistemic truthfulness
				\item \begin{enumerate}
						\item$\job(\Smith)\land\ten(\Smith)$\hfill local hypothesis
						\item$\job(\Smith)$\hfill 5.j.i, PL
						\item$\knows{\Smith}{M}\land\proves{M}{(\job(\Jones)\land\ten(\Jones))}{\Smith}{\emptyset}$\hfill 5.a, \emph{a fortiori}
						\item$\job(\Jones)\land\ten(\Jones)$\hfill 5.j.iii, epistemic truthfulness
						\item$\job(\Jones)$\hfill 5.j.iv, PL
						\item$\job(\Smith)\land\job(\Jones)$\hfill 5.j.ii, 5.j.v, PL
						\item$(\job(\Smith)\land\job(\Jones))\limp\false$\hfill 1, \emph{a fortiori}
						\item$\false$\hfill 5.j.vi, 5.j.vii, PL
					\end{enumerate}
				\item$(\job(\Smith)\land\ten(\Smith))\limp\false$\hfill 5.j.i--5.j.viii, LDT
			\end{enumerate}
		\item\quad$\LiPded(\knows{\Smith}{M}\land\proves{M}{(\job(\Jones)\land\ten(\Jones))}{\Smith}{\emptyset})\limp
				((\job(\Smith)\land\ten(\Smith))\limp\false)$\hfill 5.a--5.k, LDT
		\item if $\LiPded(\job(\Smith)\land\job(\Jones))\limp\false$\\  
				then $\LiPded(\knows{\Smith}{M}\land\proves{M}{(\job(\Jones)\land\ten(\Jones))}{\Smith}{\emptyset})\limp
				((\job(\Smith)\land\ten(\Smith))\limp\false)$\hfill 1--6, PL
		\item $\set{(\job(\Smith)\land\job(\Jones))\limp\false};\emptyset\LiPded$\\
					$(\knows{\Smith}{M}\land\proves{M}{(\job(\Jones)\land\ten(\Jones))}{\Smith}{\emptyset})\limp
				((\job(\Smith)\land\ten(\Smith))\limp\false)$\hfill 7, definition.
	\end{enumerate}\label{table:GettierExample}
\end{sidewaystable}

In order to illustrate the
	working of signatures and the 
	application of the other logical laws of LiP,
		we now refine Gettier's example with signing.
That is,
	we identify 
		the proof $M$ in Proposition~\ref{proposition:GettierExample} with
		a term pair $\pair{C}{\sign{W}{\HR}}$ consisting of, 
			first, a proof $C$ for the fact $\ten(\Jones)$ and, 
			second, a work contract $\sign{W}{\HR}$ for $\Jones$ signed by the HR department dealing with
				the job application.
\begin{lemma}[Gettier-example with signing]\label{lemma:GettierSigning}
Given $\agents\defeq\set{\Smith,\Jones,\HR}$,
	$$\LiPded\begin{array}{@{}l@{}}
				\left(\begin{array}{@{}l@{}}
				\knows{\Smith}{\sign{W}{\HR}}\land
				\proves{W}{\job(\Jones)}{\HR}{\agents}\\
				\land\;\knows{\Smith}{C}\land
				\proves{C}{\ten(\Jones)}{\Smith}{\emptyset}
				\end{array}\right)\limp\\
							\left(\begin{array}{@{}l@{}}
								\knows{\Smith}{\pair{C}{\sign{W}{\HR}}}\;\land\\
								\proves{\pair{C}{\sign{W}{\HR}}}{(\job(\Jones)\land\ten(\Jones))}{\Smith}{\emptyset}
							\end{array}\right)
							\end{array}$$
\end{lemma}
\begin{proof}\ 
	\begin{enumerate}
		\item\quad$\knows{\Smith}{\sign{W}{\HR}}$\hfill local hypothesis
		\item\qquad$\proves{W}{\job(\Jones)}{\HR}{\agents}$\hfill local hypothesis
		\item\qquad$(\proves{W}{\job(\Jones)}{\HR}{\agents})\limp
					\proves{\sign{W}{\HR}}{\job(\Jones)}{\Smith}{\agents}$\hfill simple peer review
		\item\qquad$\proves{\sign{W}{\HR}}{\job(\Jones)}{\Smith}{\agents}$\hfill 2, 3, PL
		\item\qquad$(\proves{\sign{W}{\HR}}{\job(\Jones)}{\Smith}{\agents})\limp
						\proves{\sign{W}{\HR}}{\job(\Jones)}{\Smith}{\emptyset}$\hfill group decomp.
		\item\qquad$\proves{\sign{W}{\HR}}{\job(\Jones)}{\Smith}{\emptyset}$\hfill 4, 5, PL
		\item\qquad\quad$\knows{\Smith}{C}$\hfill local hypothesis
		\item\qquad\qquad$\proves{C}{\ten(\Jones)}{\Smith}{\emptyset}$\hfill local hypothesis
		\item\qquad\qquad$(\proves{C}{\ten(\Jones)}{\Smith}{\emptyset})\limp
							\proves{\pair{C}{\sign{W}{\HR}}}{\ten(\Jones)}{\Smith}{\emptyset}$\hfill proof ext.
		\item\qquad\qquad$\proves{\pair{C}{\sign{W}{\HR}}}{\ten(\Jones)}{\Smith}{\emptyset}$\hfill 8, 9, PL
		\item\qquad\qquad$\knows{\Smith}{\sign{W}{\HR}}$\hfill 1, \emph{a fortiori}
		\item\qquad\qquad$\knows{\Smith}{C}$\hfill 7, \emph{a fortiori}
		\item\qquad\qquad$\knows{\Smith}{C}\land
					\knows{\Smith}{\sign{W}{\HR}}$\hfill 11, 12, PL
		\item\qquad\qquad$\knows{\Smith}{\pair{C}{\sign{W}{\HR}}}$\hfill 13, pairing
		\item\qquad\qquad$\proves{\sign{W}{\HR}}{\job(\Jones)}{\Smith}{\emptyset}$\hfill 6, \emph{a fortiori}
		\item\qquad\qquad$(\proves{\sign{W}{\HR}}{\job(\Jones)}{\Smith}{\emptyset})\limp
							\proves{\pair{C}{\sign{W}{\HR}}}{\job(\Jones)}{\Smith}{\emptyset}$\hfill p.\ ext.
		\item\qquad\qquad$\proves{\pair{C}{\sign{W}{\HR}}}{\job(\Jones)}{\Smith}{\emptyset}$\hfill 15, 16, PL
		\item\qquad\qquad$\proves{\pair{C}{\sign{W}{\HR}}}{(\job(\Jones)\land\ten(\Jones))}{\Smith}{\emptyset}$\hfill 10, 17, proof conj.\
		\item\qquad\qquad$\begin{array}{@{}l@{}}
							\knows{\Smith}{\pair{C}{\sign{W}{\HR}}}\;\land\\
							\proves{\pair{C}{\sign{W}{\HR}}}{(\job(\Jones)\land\ten(\Jones))}{\Smith}{\emptyset}
							\end{array}$\hfill 14, 18, PL
		\item\qquad\quad$\begin{array}{@{}l@{}}
							\proves{C}{\ten(\Jones)}{\Smith}{\emptyset}\limp\\
							\left(\begin{array}{@{}l@{}}
								\knows{\Smith}{\pair{C}{\sign{W}{\HR}}}\;\land\\
								\proves{\pair{C}{\sign{W}{\HR}}}{(\job(\Jones)\land\ten(\Jones))}{\Smith}{\emptyset}
							\end{array}\right)
							\end{array}$\hfill 8--19, LDT
		\item\qquad$\begin{array}{@{}l@{}}
					\knows{\Smith}{C}\limp\\
						\left(\begin{array}{@{}l@{}}		
							\proves{C}{\ten(\Jones)}{\Smith}{\emptyset}\limp\\
							\left(\begin{array}{@{}l@{}}
								\knows{\Smith}{\pair{C}{\sign{W}{\HR}}}\;\land\\
								\proves{\pair{C}{\sign{W}{\HR}}}{(\job(\Jones)\land\ten(\Jones))}{\Smith}{\emptyset}
							\end{array}\right)
							\end{array}\right)
							\end{array}$\hfill 7--20, LDT
		\item\quad$\begin{array}{@{}l@{}}
					\proves{W}{\job(\Jones)}{\HR}{\agents}\limp\\
					\left(\begin{array}{@{}l@{}}
					\knows{\Smith}{C}\limp\\
						\left(\begin{array}{@{}l@{}}		
							\proves{C}{\ten(\Jones)}{\Smith}{\emptyset}\limp\\
							\left(\begin{array}{@{}l@{}}
								\knows{\Smith}{\pair{C}{\sign{W}{\HR}}}\;\land\\
								\proves{\pair{C}{\sign{W}{\HR}}}{(\job(\Jones)\land\ten(\Jones))}{\Smith}{\emptyset}
							\end{array}\right)
							\end{array}\right)
							\end{array}\right)
							\end{array}$\hfill 2--21, LDT
		\item$\LiPded\begin{array}{@{}l@{}}
				\knows{\Smith}{\sign{W}{\HR}}\limp\\
				\left(\begin{array}{@{}l@{}}
					\proves{W}{\job(\Jones)}{\HR}{\agents}\limp\\
					\left(\begin{array}{@{}l@{}}
					\knows{\Smith}{C}\limp\\
						\left(\begin{array}{@{}l@{}}		
							\proves{C}{\ten(\Jones)}{\Smith}{\emptyset}\limp\\
							\left(\begin{array}{@{}l@{}}
								\knows{\Smith}{\pair{C}{\sign{W}{\HR}}}\;\land\\
								\proves{\pair{C}{\sign{W}{\HR}}}{(\job(\Jones)\land\ten(\Jones))}{\Smith}{\emptyset}
							\end{array}\right)
							\end{array}\right)
							\end{array}\right)
							\end{array}\right)
							\end{array}$\hfill 1--22, LDT
		\item$\LiPded\begin{array}{@{}l@{}}
				\left(\begin{array}{@{}l@{}}
				\knows{\Smith}{\sign{W}{\HR}}\land
				\proves{W}{\job(\Jones)}{\HR}{\agents}\\
				\land\;\knows{\Smith}{C}\land
				\proves{C}{\ten(\Jones)}{\Smith}{\emptyset}
				\end{array}\right)\limp\\
							\left(\begin{array}{@{}l@{}}
								\knows{\Smith}{\pair{C}{\sign{W}{\HR}}}\;\land\\
								\proves{\pair{C}{\sign{W}{\HR}}}{(\job(\Jones)\land\ten(\Jones))}{\Smith}{\emptyset}
							\end{array}\right)
							\end{array}$\hfill 23, PL.
	\end{enumerate}
\end{proof}
\noindent
In the preceding proof,
	observe the use of the law of proof extension, 
		deducible by means of epistemic antitonicity, and
		expressing the monotonicity of LiP-proofs.
Like Art\"{e}mov, 
	who interprets Lehrer and Paxson's indefeasibility condition for justified true belief 
	as possibly corresponding to LP's sum-axiom (\cf \cite{JustificationLogic}), 
we could thus interpret this condition 
	as corresponding to LiP's proof extension.

\begin{corollary}[Gettier-example with signing]\label{corollary:GettierSigning}
Given $\agents\defeq\set{\Smith,\Jones,\HR}$,
	$$\begin{array}{@{}l@{}}
		\set{(\job(\Smith)\land\job(\Jones))\limp\false};\emptyset\LiPded\\
			\left(\begin{array}{@{}l@{}}
			\knows{\Smith}{\sign{W}{\HR}}\land
			(\proves{W}{\job(\Jones)}{\HR}{\agents})\\
			\land\;\knows{\Smith}{C}\land
			\proves{C}{\ten(\Jones)}{\Smith}{\emptyset}
			\end{array}\right)\limp\\
				((\job(\Smith)\land\ten(\Smith))\limp\false)
		\end{array}$$
\end{corollary}
\begin{proof}
	From 
		Proposition~\ref{proposition:GettierExample} and 
		Lemma~\ref{lemma:GettierSigning}.
\end{proof}

\subsection{Multi-agent LP-like systems}
By their quality of being conservative extensions of non-interactive LP-like systems, 
	the following logical systems with multi-agent character inherit the lack of 
		message-passing interactivity of LP in 
			the following senses: namely the lack of 
			(1) a sound truth axiom for message passing 
				(\cf Section~\ref{section:Interactivity}), 
			(2) the transferability of local truths by means of messages 
				(\cf Section~\ref{section:LogicalLaws}), and 
			(3) signature checking that could act as proof checking of claimed local truths 
				(\cf Section~\ref{section:ProofTerms}).
In our understanding,
	these lacks of LP-like systems \emph{without} message passing are reflected by  
		the fact that LP can only homomorph but not embed into interactive-proof systems 
			\emph{with} message passing like LiP.

\subsubsection{LP$^{2}$}
Yavorskaya's LP$^{2}$ \cite{YavorskayaSidon} is an extension of LP with multi-agent character in the sense that
	LP$^{2}$ extends LP with a \emph{2-agent} view such that each one of the two agents
		\begin{enumerate}
			\item has her own proof-sum, proof-application, and proof-checker constructor
			\item may have a constructor for 
				\begin{enumerate}
					\item checking the other agent's proofs, that is,  peer proofs
					\item converting peer proofs into proofs of her own.
				\end{enumerate}
		\end{enumerate}
LP$^{2}$ being an extension of LP,
	our criticism of LP also applies to LP$^{2}$.
Also,
	LiP can manage an $n$-agent view for arbitrary $n\in\mathbb{N}$ with
		only $n+1$ (transmittable) proof-term constructors ($n$ signature constructors plus 1 pair constructor).
This feature is the fruit of our design decision to equip LiP with 
	proof-term signature constructors and
	an agent-parametric proof modality, which
		allows the association as proof of 
			arbitrary data to 
				arbitrary verifying agents within 
					arbitrary peer communities.
Whereas an extension of LP$^{2}$ to LP$^{n}$ for a fixed $n\in\mathbb{N}$ would require
	$3n+2n(n-1)=2n^{2}+n$ constructors 
		($n$ proof-sum plus $n$ proof-application plus $n$ proof-checker plus 
			$n(n-1)$ peer-proof-checker plus
			$n(n-1)$ peer-proof-conversion constructors), and 
	still not allow the free association of proofs to agents.
In sum, 
	LiP seems more appropriate for interactivity and is even simpler than would be LP$^{n}$.
However,
	it could be interesting to parametrise 
		Yavorskaya's agent-centric proof converters 
			with agent \emph{communities} so that
				two communities that do not share their respective common knowledge of 
					what should constitute a proof could 
						communicate with each other thanks 
							to such communal proof converters.

\subsubsection{UL}
Renne's UL \cite{Renne2} is an extension of Art\"{e}mov's Justification Logic, JL \cite{JustificationLogic} 
	(including Art\"{e}mov's LP) with multi-agent character in the sense that
		UL combines JL with (multi-agent) Dynamic Epistemic Logic \cite{DynamicEpistemicLogic}.
Of course, dynamic extensions of static logics are interesting.
The sophisticated language of UL 
	is defined by staged mutual recursion on the structure of terms and formulas, and
	has a semantic interface in the style of LP but only over finite Kripke-models.
The mutual recursion arises in the application term constructor of UL, which 
	has a formula parameter meant to indicate the relevance of 
		the second constructor argument to the constructor parameter in UL's application axiom.
Given that sum and application can be subsumed by pairing in LiP (\cf Section~\ref{section:ProofTerms}),
	it would be interesting to experiment with a formula-parametrised pair constructor in UL 
		intended to subsume sum and formula-parametrised application.
The justification terms in UL do not provide evidence for knowledge but only for belief, which 
	is expressed with a K4-modality. 
(Usually, belief is expressed with a KD45-modality \cite{MultiAgents}.)

\section{Conclusion}

\subsection{Assessment}
We have proposed a logic of interactive proofs with 
	as main contributions those described in Section~\ref{section:Contribution}.
Our resulting notion of proofs has the advantage of being
	not only operational thanks to a proof-theoretic definition 
	but also \emph{declarative} thanks to a complementary model-theoretic definition that 
		gives a \emph{constructive epistemic semantics} to proofs in the sense of
			explicating \emph{what}---\emph{knowledge,} as well as \emph{skill} by means of mobile code---proofs effect in agents, complementing thereby 
				the (operational) axiomatics, which explicates \emph{how} proofs do so. 
In particular, 
	first, \emph{interactive} computation is \emph{semantic} computation: 
		we not only compute result values (syntax), but 
			(knowledge) equivalence classes of them (semantics); and,  
	second, our definition of interactive proofs reflects the impact of mathematical proofs in a social sense (\cf Section~\ref{section:Solution}): 
		if my peer knew my proof for her of a given statement then she would know that the statement is true. 
(Notice the different kinds of knowledge and the conditional mode!) 
In contrast, the traditional definition of (mathematical) proofs is only operational in the sense that 
	proofs are defined purely in terms of the deductive operations that are used to construct them.
Their \emph{pragmatics,} that is,  their (epistemic) impact in proof-checking agents, was left unformalised, and
	their operational definition risks restricting their generality.
However now thanks to our formalisation, 
	we as a community have the formal common knowledge that 
	\begin{itemize}
		\item agents in distributed systems are at the same time 
			computation oracles,
			data miners,
			knowledge processors, 
			meaning interpreters, 
			message-passing communicators, 
			interactive provers, and
			logical combinators;
		\item \textbf{a proof is \emph{that which 
				if known to one of our peer members 
				would induce the knowledge of its proof goal with that member.}}
	\end{itemize}

\subsection{Future work}
Our future lines of research for LiP are the following: 
\begin{enumerate}
	\item develop the proof theory of LiP (alternative calculi, proof complexity);
	\item extend LiP with guarded quantifiers (gFOLiP), dynamic modalities, and fixpoint operators 
			(Hennessy-Milner correspondence, characteristic formulas);
	\item extend LiP with the classical and the modern conception of cryptography mentioned in 
			Footnote~\ref{footnote:CryptographyConceptions} (requiring resource-bounded agents);
	\item apply LiP and its variants to the analysis and synthesis of communication protocols 
		(proof-carrying code correct by construction via 
			program extraction from constructive proofs of correctness, 
				on-line interactive algorithms);
	\item create the Logic of Evidence and the Logic of Deception suggested on Page~\pageref{page:EvidenceDeception}.
\end{enumerate}
Applying LiP means fixing four things if need be, namely, at the level of
	\begin{enumerate}
		\item \emph{terms:} 
			\begin{enumerate}
				\item the choice of term axioms,
				\item the application-specific base data $B$,
				\item the implementation of signing, \eg
						in terms of public-key cryptography;
			\end{enumerate}
		\item \emph{formulas:} 
				the set $\mathcal{P}$ of atomic propositions (those besides $\knows{a}{M}$) together with
				the axioms governing their intended meaning. 
	\end{enumerate}
This will instantiate LiP as a theory of the specific subject matter of the application, 
	such as, for example, Dolev-Yao cryptography (\cf Page~\pageref{page:DolevYao}).

\begin{acknowledgements}
I thank 
	Johan van Benthem and Larry Moss for their early encouragement and 
	Sergei Art\"{e}mov for his sympathetic consideration of my work.
I also thank  
	Eiji Okamoto for being my generous host during my post-doctoral fellowship from 
		the Japan Society for the Promotion of Science,  
	Jean-Luc Beuchat 
		for suggesting to apply for such a fellowship and 
		for helping me with and sharing with me the fascinating life of a \emph{gaijin,} and 
	Shihoko Sekiya for being our perfect secretary and daily office sunshine.
Last but not least,
	I thank 
		Denis Saveliev for giving me his favourable opinion on this paper, and
		Olga Grinchtein for spotting and informing me about a few typos.
\end{acknowledgements}

\bibliographystyle{alpha}
%\bibliography{/Users/simonkramer/Documents/Sources/TeX/bibliography}

\appendix
\section{Completeness proof}\label{appendix:Proofs}
\paragraph{Completeness}
	For all $\phi\in\pFormulas$, 
		if $\models\phi$ then $\LiPded\phi$.
\begin{proof}
	Let
	\newcommand{\canrel}[3]{\mathrel{_{#1}\negthinspace\mathrm{C}_{#2}^{#3}}}
	\newcommand{\canVal}{\mathcal{V}_{\mathsf{C}}}
		\begin{itemize}
			\item $\mathcal{W}$ designate the set of all maximally LiP-consistent sets\footnote{*
				A set $W$ of LiP-formulas is maximally LiP-consistent :iff 
					$W$ is LiP-consistent and 
					$W$ has no proper superset that is LiP-consistent.
				A set $W$ of LiP-formulas is LiP-consistent :iff 
					$W$ is not LiP-inconsistent.
				A set $W$ of LiP-formulas is LiP-inconsistent :iff 
					there is a finite $W'\subseteq W$ such that $((\bigwedge W')\limp\false)\in\text{LiP}$.
				Any LiP-consistent set can be extended to a maximally LiP-consistent set by means of  
					the Lindenbaum Construction \cite[Page~90]{ModalProofTheory}.
				A set is maximally LiP-consistent if and only if 
				the set of logical-equivalence classes of the set is an ultrafilter of
				the Lindenbaum-Tarski algebra of LiP \cite[Page~351]{AlgebrasAndCoalgebras}.
				The canonical frame is isomorphic to the ultrafilter frame of that Lindenbaum-Tarski algebra 
					\cite[Page~352]{AlgebrasAndCoalgebras}.}
			\item for all $w,w'\in\mathcal{W}$,
				$w\canrel{M}{a}{\community}w'$ :iff $\setst{\phi\in\pFormulas}{\proves{M}{\phi}{a}{\community}\in w}\subseteq w'$
			\item for all $w\in\mathcal{W}$, $w\in\canVal(P)$ :iff $P\in w$.
		\end{itemize}
	Then \newcommand{\canModel}{\mathfrak{M}_{\mathsf{C}}}
		$\canModel\defeq
			(\mathcal{W},\set{\canrel{M}{a}{\community}}_{M\in\messages,a\in\agents,\community\subseteq\agents},\canVal)$
		designates the \emph{canonical model} for LiP.
	Following Fitting \cite[Section~2.2]{ModalProofTheory}, 
	the following useful property of $\canModel$,  
		$$\boxed{$\text{for all $\phi\in\pFormulas$ and $w\in\mathcal{W}$,
			$\phi\in w$ if and only if $\canModel,w\models\phi$,}$}$$
	the so-called \emph{Truth Lemma}, can be proved by induction on the structure of $\phi$:
	\begin{enumerate}
		\item Base case ($\phi\defeq P$ for $P\in\mathcal{P}$). 
			For all $w\in\mathcal{W}$,
				$P\in w$ if and only if $\canModel,w\models P$, 
					by definition of $\canVal$.
		\item Inductive step ($\phi\defeq \neg\phi'$ for $\phi'\in\pFormulas$).
			Suppose that
				for all $w\in\mathcal{W}$,
					$\phi'\in w$ if and only if $\canModel,w\models\phi'$.
			Further let
				$w\in\mathcal{W}$.
			Then, 
				$\neg\phi'\in w$ if and only if $\phi'\not\in w$ --- $w$ is consistent ---
				if and only if $\canModel,w\not\models\phi'$ --- by the induction hypothesis ---
				if and only if $\canModel,w\models\neg\phi'$.
		\item Inductive step ($\phi\defeq \phi'\land\phi''$ for $\phi',\phi''\in\pFormulas$).
			Suppose that 
				for all $w\in\mathcal{W}$,
					$\phi'\in w$ if and only if $\canModel,w\models\phi'$, and that
				for all $w\in\mathcal{W}$,
					$\phi''\in w$ if and only if $\canModel,w\models\phi''$.
			Further let
				$w\in\mathcal{W}$.
			Then, 
				$\phi'\land\phi''\in w$ if and only if 
					($\phi'\in w$ and $\phi''\in w$), because $w$ is maximal.
			Now suppose that
				$\phi'\in w$ and $\phi''\in w$.
			Hence, 
				$\canModel,w\models\phi'$ and 
				$\canModel,w\models\phi''$, by the induction hypotheses, and 
			thus $\canModel,w\models\phi'\land\phi''$.
			Conversely, suppose that
				$\canModel,w\models\phi'\land\phi''$.
			Then,
				$\canModel,w\models\phi'$ and 
				$\canModel,w\models\phi''$.
			Hence,
				$\phi'\in w$ and $\phi''\in w$, by the induction hypotheses.
			Thus,
				($\phi'\in w$ and $\phi''\in w$) if and only if
				($\canModel,w\models\phi'$ and 
				$\canModel,w\models\phi''$).
			Whence
				$\phi'\land\phi''\in w$ if and only if
				($\canModel,w\models\phi'$ and 
				$\canModel,w\models\phi''$), by transitivity.
		\item Inductive step ($\phi\defeq\proves{M}{\phi'}{a}{\community}$ for 
				$M\in\messages$, $a\in\agents$, $\community\subseteq\agents$, and $\phi'\in\pFormulas$).
			\begin{flushleft}
			\nn{4.1} for all $w\in\mathcal{W}$,
						$\phi'\in w$ if and only if $\canModel,w\models\phi'$\hfill ind.\ hyp.\\[\jot]
			\nn{4.2}\quad $w\in\mathcal{W}$\hfill hyp.\\[2\jot]
			\nn{4.3}\qquad $\proves{M}{\phi'}{a}{\community}\in w$\hfill hyp.\\[\jot]
			\nn{4.4}\qquad\quad	$w'\in\mathcal{W}$\hfill hyp.\\[\jot]
			\nn{4.5}\qquad\qquad $w\canrel{M}{a}{\community}w'$\hfill hyp.\\[\jot]
			\nn{4.6}\qquad\qquad $\setst{\phi''\in\pFormulas}{\proves{M}{\phi''}{a}{\community}\in w}\subseteq w'$\hfill 4.5\\[\jot]
			\nn{4.7}\qquad\qquad $\phi'\in\setst{\phi''\in\pFormulas}{\proves{M}{\phi''}{a}{\community}\in w}$\hfill 4.3, 4.6\\[\jot]
			\nn{4.8}\qquad\qquad $\phi'\in w'$\hfill 4.6, 4.7\\[\jot]
			\nn{4.9}\qquad\qquad $\canModel,w'\models\phi'$\hfill 4.1, 4.4, 4.8\\[\jot]
			\nn{4.10}\qquad\quad if $w\canrel{M}{a}{\community}w'$ then $\canModel,w'\models\phi'$\hfill 4.5--4.9\\[\jot]
			\nn{4.11}\qquad for all $w'\in\mathcal{W}$, 
					if $w\canrel{M}{a}{\community}w'$ 
					then $\canModel,w'\models\phi'$\hfill 4.4--4.10\\[\jot]
			\nn{4.12}\qquad $\canModel,w\models\proves{M}{\phi'}{a}{\community}$\hfill 4.11\\[2\jot]
			\nn{4.13}\qquad $\proves{M}{\phi'}{a}{\community}\not\in w$\hfill hyp.\\[\jot]
			\nn{4.14}\qquad\quad $\mathcal{F}=\setst{\phi''\in\pFormulas}{\proves{M}{\phi''}{a}{\community}\in w}\cup\set{\neg\phi'}$\hfill hyp.\\[\jot]
			\nn{4.15}\qquad\qquad $\mathcal{F}$ is LiP-inconsistent\hfill hyp.\\[\jot]
			\nn{4.16}\qquad\qquad there is $\set{\proves{M}{\phi_{1}}{a}{\community},\ldots,\proves{M}{\phi_{n}}{a}{\community}}\subseteq w$ such that\\
			\nn{}\qquad\qquad $\LiPded(\phi_{1}\land\ldots\land\phi_{n}\land\neg\phi')\limp\false$\hfill 4.14, 4.15\\[\jot]
			\nn{4.17}\qquad\qquad\quad $\set{\proves{M}{\phi_{1}}{a}{\community},\ldots,\proves{M}{\phi_{n}}{a}{\community}}\subseteq w$ and\\
			\nn{}\qquad\qquad\quad $\LiPded(\phi_{1}\land\ldots\land\phi_{n}\land\neg\phi')\limp\false$\hfill hyp.\\[\jot]
			\nn{4.18}\qquad\qquad\quad $\LiPded(\phi_{1}\land\ldots\land\phi_{n})\limp\phi'$\hfill 4.17\\[\jot]
			\nn{4.19}\qquad\qquad\quad $\LiPded(\proves{M}{(\phi_{1}\land\ldots\land\phi_{n})}{a}{\community})\limp\proves{M}{\phi'}{a}{\community}$\hfill 4.18, R\\[\jot]
			\nn{4.20}\qquad\qquad\quad $\LiPded((\proves{M}{\phi_{1}}{a}{\community})\land\ldots\land(\proves{M}{\phi_{n}}{a}{\community}))\limp\proves{M}{\phi'}{a}{\community}$\hfill 4.19\\[\jot]
			\nn{4.21}\qquad\qquad\quad $\proves{M}{\phi'}{a}{\community}\in w$\hfill 4.17, 4.20, $w$ is maximal\\[\jot]
			\nn{4.22}\qquad\qquad\quad false\hfill 4.13, 4.21\\[\jot]
			\nn{4.23}\qquad\qquad false\hfill 4.16, 4.17--4.22\\[\jot]
			\nn{4.24}\qquad\quad $\mathcal{F}$ is LiP-consistent\hfill 4.15--4.23\\[\jot]
			\nn{4.25}\qquad\quad there is $w'\supseteq\mathcal{F}$ \st $w'$ is maximally LiP-consistent\hfill 4.24\\[\jot]
			\nn{4.26}\qquad\qquad $\mathcal{F}\subseteq w'$ and $w'$ is maximally LiP-consistent\hfill hyp.\\[\jot]
			\nn{4.27}\qquad\qquad $\setst{\phi''\in\pFormulas}{\proves{M}{\phi''}{a}{\community}\in w}\subseteq\mathcal{F}$\hfill 4.14\\[\jot]
			\nn{4.28}\qquad\qquad $\setst{\phi''\in\pFormulas}{\proves{M}{\phi''}{a}{\community}\in w}\subseteq w'$\hfill 4.26, 4.27\\[\jot]
			\nn{4.29}\qquad\qquad $w\canrel{M}{a}{\community}w'$\hfill 4.28\\[\jot]
			\nn{4.30}\qquad\qquad $w'\in\mathcal{W}$\hfill 4.26\\[\jot]
			\nn{4.31}\qquad\qquad $\neg\phi'\in\mathcal{F}$\hfill 4.14\\[\jot]
			\nn{4.32}\qquad\qquad $\neg\phi'\in w'$\hfill 4.26, 4.31\\[\jot]
			\nn{4.33}\qquad\qquad $\phi'\not\in w'$\hfill 4.26 ($w'$ is LiP-consistent), 4.32\\[\jot]
			\nn{4.34}\qquad\qquad $\canModel,w'\not\models\phi'$\hfill 4.1, 4.33\\[\jot]
			\nn{4.35}\qquad\qquad there is $w'\in\mathcal{W}$ \st 
				$w\canrel{M}{a}{\community}w'$ and $\canModel,w'\not\models\phi'$\hfill 4.29, 4.34\\[\jot]
			\nn{4.36}\qquad\qquad $\canModel,w\not\models\proves{M}{\phi'}{a}{\community}$\hfill 4.35\\[\jot]
			\nn{4.37}\qquad\quad $\canModel,w\not\models\proves{M}{\phi'}{a}{\community}$\hfill 4.25, 4.26--4.36\\[\jot]
			\nn{4.38}\qquad $\canModel,w\not\models\proves{M}{\phi'}{a}{\community}$\hfill 4.14--4.37\\[2\jot]
			\nn{4.39}\quad $\proves{M}{\phi'}{a}{\community}\in w$ if and only if $\canModel,w\models\proves{M}{\phi'}{a}{\community}$\hfill 4.3--4.12, 4.13--4.38\\[\jot]
			\nn{4.40} for all $w\in\mathcal{W}$,
				$\proves{M}{\phi'}{a}{\community}\in w$ if and only if $\canModel,w\models\proves{M}{\phi'}{a}{\community}$\hfill 4.2--4.39\\[\jot]
		\end{flushleft}
	\end{enumerate}

	With the Truth Lemma we can now prove that 
			for all $\phi\in\pFormulas$,
				if $\not\LiPded\phi$ then $\not\models\phi$.
	Let 
		$\phi\in\pFormulas$, and 
	suppose that 
		$\not\LiPded\phi$.
	Thus, 
		$\set{\neg\phi}$ 
			is LiP-consistent, and 
			can be extended to a maximally LiP-consistent set $w$, that is, 
				$\neg\phi\in w\in\mathcal{W}$.
	Hence 
		$\canModel,w\models\neg\phi$, by the Truth Lemma.
	Thus: 
		$\canModel,w\not\models\phi$,
		$\canModel\not\models\phi$, and
		$\not\models\phi$.
	That is,
			$\canModel$ is a 
				\emph{universal} (for \emph{all} $\phi\in\pFormulas$) 
				\emph{counter-model} (if $\phi$ is a non-theorem then $\canModel$ falsifies $\phi$).
	
	We are left to prove that 
		$\canModel$ is also an \emph{LiP-model}. 
	So 
		let us instantiate 
			our data mining operator $\clo{a}{}$ (\cf Page~\pageref{page:DataMining}) on $\mathcal{W}$ by 
				letting for all $w\in\mathcal{W}$
					$$\msgs{a}(w)\defeq\setst{M}{\knows{a}{M}\in w}.$$
	Then, let us prove that:
			\begin{enumerate}
				\item	if for all $b\in\community\cup\set{a}$, 
							$w\canrel{\sign{M}{a}}{b}{\community\cup\set{a}}w'$ 
						then $M\in\clo{a}{w'}(\emptyset)$
				\item if $M\in\clo{a}{w}(\emptyset)$
						then $w\canrel{M}{a}{\community}w$
				\item there is $w'\in\mathcal{W}$ such that 
						$w\canrel{M}{a}{\community}w'$
				\item for all $b\in\community\cup\set{a}$, 
						$({\canrel{\sign{M}{a}}{b}{\community\cup\set{a}}}\circ{\canrel{M}{a}{\community}})\subseteq{\canrel{M}{a}{\community}}$
				\item if $\community\subseteq\community'$ 
						then ${\canrel{M}{a}{\community}}\subseteq{\canrel{M}{a}{\community'}}$
				\item if $M\leq_{a}M'$ then ${\canrel{M}{a}{\community}}\subseteq{\canrel{M'}{a}{\community}}$
			\end{enumerate}
	For (1),
		let $w,w'\in\mathcal{W}$ and 
		suppose that for all $b\in\community\cup\set{a}$, $w\canrel{\sign{M}{a}}{b}{\community\cup\set{a}}w'$.
	That is,	
		for all $b\in\community\cup\set{a}$ and $\phi\in\pFormulas$, 
			if $\proves{\sign{M}{a}}{\phi}{b}{\community\cup\set{a}}\in w$
			then $\phi\in w'$.
	Since $w$ is maximal,
		$$\text{$(\bigwedge_{b\in\community\cup\set{a}}\proves{\sign{M}{a}}{\knows{a}{M}}{b}{\community\cup\set{a}})\in w$\quad(authentic knowledge).}$$
	Hence 
		$\knows{a}{M}\in w'$, by \emph{modus ponens,} and
	thus $M\in\clo{a}{w'}(\emptyset)$ by definition.
	
	For (2), 
		let $w\in\mathcal{W}$ and
		suppose that 
			$M\in\clo{a}{w}(\emptyset)$.
	Hence $\knows{a}{M}\in w$ due to the maximality of $w$, which 
		contains all the term axioms corresponding to the defining clauses of $\clo{a}{w}$.
	Further suppose that $\proves{M}{\phi}{a}{\community}\in w$.
	Since $w$ is maximal,
		$$\text{$(\proves{M}{\phi}{a}{\community})\limp(\knows{a}{M}\limp\phi)\in w$\quad(epistemic truthfulness).}$$
	Hence, 
		$\knows{a}{M}\limp\phi\in w$, and
		$\phi\in w$, by consecutive  \emph{modus ponens.}
	
	For (3),
		let $w\in\mathcal{W}$ and $\phi\in\pFormulas$, and
		suppose that $\proves{M}{\phi}{a}{\community}\in w$.
	For the sake of deriving the contrary, 
		further suppose that $\phi\not\in w$.
	Hence $\neg\phi\in w$ because $w$ is maximal, and 
	thus $\phi\limp\false\in w$. 
	Hence $(\proves{M}{\phi}{a}{\community})\limp\proves{M}{\false}{a}{\community}\in w$ by regularity.
	Hence $\proves{M}{\false}{a}{\community}\in w$ by the first supposition and \emph{modus ponens}.
	Hence $\neg(\proves{M}{\false}{a}{\community})\not\in w$ because $w$ is consistent.
	Yet since $w$ is maximal,
		$\neg(\proves{M}{\false}{a}{\community})\in w$ (proof consistency).
	Contradiction.
	Hence $w$ is actually a $w'$ such that $\phi\in w'$.
	
	For (4), 
		suppose that $b\in\community\cup\set{a}$ and 
		let $w,w',w''\in\mathcal{W}$.
	Further suppose that 
			$w\canrel{\sign{M}{a}}{b}{\community\cup\set{a}}w'$ 
				(\ie for all $\phi\in\pFormulas$,
					if $\proves{\sign{M}{a}}{\phi}{b}{\community\cup\set{a}}\in w$ 
					then $\phi\in w'$) and 
			$w'\canrel{M}{a}{\community}w''$
				(\ie for all $\phi\in\pFormulas$,
					if $\proves{M}{\phi}{a}{\community}\in w'$ 
					then $\phi\in w''$).
	Furthermore suppose that
		$\proves{M}{\phi}{a}{\community}\in w$.
	Since $w$ is maximal, 
		$$\text{$(\proves{M}{\phi}{a}{\community})\limp\bigwedge_{b\in\community\cup\set{a}}\proves{\sign{M}{a}}{(\proves{M}{\phi}{a}{\community})}{b}{\community\cup\set{a}}\in w$,}$$
	as a direct consequence of peer review and proof conjunctions.
	Hence, applying \emph{modus ponens} consecutively, 
		$\bigwedge_{b\in\community\cup\set{a}}\proves{\sign{M}{a}}{(\proves{M}{\phi}{a}{\community})}{b}{\community\cup\set{a}}\in w$ by the fourth supposition, 
		$\proves{M}{\phi}{a}{\community}\in w'$ by the second supposition, and finally 
		$\phi\in w''$ by the third supposition.
		
	For (5),
		let $\community'\subseteq\agents$ and
		suppose that $\community\subseteq\community'$.
	That is,
		$\community\cup\community'=\community'$.
	Further,
		let $w,w'\in\mathcal{W}$ and 
		suppose that $w\canrel{M}{a}{\community}w'$.
	That is,
		for all $\phi\in\pFormulas$,
			if $\proves{M}{\phi}{a}{\community}\in w$ 
			then $\phi\in w'$.
	Furthermore,	
		let $\phi\in\pFormulas$ and
		suppose that $\proves{M}{\phi}{a}{\community'}\in w$.
	Thus $\proves{M}{\phi}{a}{\community\cup\community'}\in w$ by the first supposition.
	Since $w$ is maximal,
		$$\text{$(\proves{M}{\phi}{a}{\community\cup\community'})\limp\proves{M}{\phi}{a}{\community}\in w$\quad(group decomposition).}$$
	Hence $\proves{M}{\phi}{a}{\community}\in w$ by \emph{modus ponens}, and 
	thus $\phi\in w'$ by the second supposition.
	
	For (6),
		suppose that 
			$M\leq_{a}M'$.
	That is,
		for all $w\in\mathcal{W}$, 
			if $M\in\clo{a}{w}(\emptyset)$ then $M'\in\clo{a}{w}(\emptyset)$.
	Hence for all $w\in\mathcal{W}$,
		if $\knows{a}{M}\in w$ then $\knows{a}{M'}\in w$
				due to the maximality of $w'$, which 
					contains all the term axioms corresponding to the defining clauses of $\clo{a}{w}$.
	Hence for all $w\in\mathcal{W}$,
			if $\canModel, w\models\knows{a}{M}$ then $\canModel, w\models\knows{a}{M'}$, by the Truth Lemma.
	Thus for all $w\in\mathcal{W}$, $\canModel, w\models\knows{a}{M}\limp\knows{a}{M'}$.
	Hence for all $w\in\mathcal{W}$, 
		$\knows{a}{M}\limp\knows{a}{M'}\in w$ by the Truth Lemma.
	Hence the following intermediate result, called IR, 
			$$\text{for all $w\in\mathcal{W}$ and $\phi\in\pFormulas$, 
			$(\proves{M}{\phi}{a}{\community})\limp\proves{M'}{\phi}{a}{\community}\in w$,}$$
		by 
			epistemic antitonicity.
	Further,
		let $w,w'\in\mathcal{W}$.
	Hence,
		\begin{itemize}
			\item $w\canrel{M}{a}{\community}w'$ by definition if and only if
			\item (for all $\phi\in\pFormulas$,
					if $\proves{M}{\phi}{a}{\community}\in w$ 
					then $\phi\in w'$), which by IR implies 
			\item (for all $\phi\in\pFormulas$,
					if $\proves{M'}{\phi}{a}{\community}\in w$ 
					then $\phi\in w'$) by definition if and only if 
			\item $w\canrel{M'}{a}{\community}w'$.
		\end{itemize}
\end{proof}

\section{Other proofs}
	Let 
		``PT'' abbreviate ``classical propositional tautology'' and
		``PL'' ``classical propositional logic'', and
	let PT and PL refer to the propositional fragment of LiP only.
\subsection{Proof of Theorem~\ref{theorem:SomeUsefulDeducibleStructuralLaws}}\label{appendix:StructuralProofs}
	\begin{enumerate}
		\item	
			\begin{enumerate}
				\item $\LiPded\knows{a}{\pair{M}{M'}}\limp(\knows{a}{M}\land\knows{a}{M'})$\hfill unpairing
				\item $\LiPded(\knows{a}{M}\land\knows{a}{M'})\limp\knows{a}{M}$\hfill PT
				\item $\LiPded\knows{a}{\pair{M}{M'}}\limp\knows{a}{M}$\hfill a, b, PL.
			\end{enumerate}
		\item Symmetrically to 1.
		\item
			\begin{enumerate}
				\item $\LiPded\knows{a}{\pair{M}{M}}\limp\knows{a}{M}$\hfill left or right projection
				\item $\LiPded\knows{a}{M}\limp(\knows{a}{M}\land\knows{a}{M})$\hfill PT
				\item $\LiPded(\knows{a}{M}\land\knows{a}{M})\limp\knows{a}{\pair{M}{M}}$\hfill pairing
				\item $\LiPded\knows{a}{M}\limp\knows{a}{\pair{M}{M}}$\hfill b, c, PL
				\item $\LiPded\knows{a}{\pair{M}{M}}\lequiv\knows{a}{M}$\hfill a, d, PL.
			\end{enumerate}
		\item
			\begin{enumerate}
				\item $\LiPded(\knows{a}{M}\land\knows{a}{M'})\lequiv\knows{a}{\pair{M}{M'}}$\hfill [un]pairing
				\item $\LiPded(\knows{a}{M'}\land\knows{a}{M})\lequiv\knows{a}{\pair{M'}{M}}$\hfill [un]pairing
				\item $\LiPded(\knows{a}{M}\land\knows{a}{M'})\lequiv(\knows{a}{M'}\land\knows{a}{M})$\hfill PT
				\item $\LiPded\knows{a}{\pair{M}{M'}}\lequiv\knows{a}{\pair{M'}{M}}$\hfill a, b, c, PL.
			\end{enumerate}
		\item 
			\begin{enumerate}
				\item $\LiPded(\knows{a}{M}\limp\knows{a}{M'})\lequiv(\knows{a}{M}\limp(\knows{a}{M}\land\knows{a}{M'}))$\hfill PT
				\item $\LiPded(\knows{a}{M}\land\knows{a}{M'})\lequiv\knows{a}{\pair{M}{M'}}$\hfill [un]pairing
				\item $\LiPded(\knows{a}{M}\limp\knows{a}{M'})\lequiv(\knows{a}{M}\limp\knows{a}{\pair{M}{M'}})$\hfill a, b, PL
				\item $\LiPded\knows{a}{\pair{M}{M'}}\limp\knows{a}{M}$\hfill left projection
				\item $\LiPded(\knows{a}{M}\limp\knows{a}{M'})\lequiv(\knows{a}{\pair{M}{M'}}\lequiv\knows{a}{M})$\hfill c, d, PL.
			\end{enumerate}
		\item
			\begin{enumerate}
				\item $\LiPded\knows{a}{a}$\hfill knowledge of one's own name
				\item $\LiPded\knows{a}{a}\limp(\knows{a}{M}\limp\knows{a}{a})$\hfill PT
				\item $\LiPded\knows{a}{M}\limp\knows{a}{a}$\hfill b, c, PL
				\item $\LiPded(\knows{a}{M}\limp\knows{a}{a})\lequiv(\knows{a}{\pair{M}{a}}\lequiv\knows{a}{M})$\hfill neutral pair elements
				\item $\LiPded\knows{a}{\pair{M}{a}}\lequiv\knows{a}{M}$\hfill c, d, PL.
			\end{enumerate}
		\item
			\begin{enumerate}
				\item $\LiPded(\knows{a}{M}\land\knows{a}{\pair{M'}{M''}})\lequiv\knows{a}{\pair{M}{\pair{M'}{M''}}}$\hfill [un]pairing
				\item $\LiPded(\knows{a}{M'}\land\knows{a}{M''})\lequiv\knows{a}{\pair{M'}{M''}}$\hfill [un]pairing
				\item $\LiPded(\knows{a}{M}\land(\knows{a}{M'}\land\knows{a}{M''}))\lequiv\knows{a}{\pair{M}{\pair{M'}{M''}}}$\hfill a, b, PL
				\item $\LiPded(\knows{a}{M}\land(\knows{a}{M'}\land\knows{a}{M''}))\lequiv((\knows{a}{M}\land\knows{a}{M'})\land\knows{a}{M''})$\hfill PT
				\item $\LiPded((\knows{a}{M}\land\knows{a}{M'})\land\knows{a}{M''})\lequiv\knows{a}{\pair{M}{\pair{M'}{M''}}}$\hfill c, d, PL
				\item $\LiPded(\knows{a}{M}\land\knows{a}{M'})\lequiv\knows{a}{\pair{M}{M'}}$\hfill [un]pairing
				\item $\LiPded(\knows{a}{\pair{M}{M'}}\land\knows{a}{M''})\lequiv\knows{a}{\pair{M}{\pair{M'}{M''}}}$\hfill e, f, PL
				\item $\LiPded(\knows{a}{\pair{M}{M'}}\land\knows{a}{M''})\lequiv\knows{a}{\pair{\pair{M}{M'}}{M''}}$\hfill [un]pairing
				\item $\LiPded\knows{a}{\pair{M}{\pair{M'}{M''}}}\lequiv\knows{a}{\pair{\pair{M}{M'}}{M''}}$\hfill g, h, PL.
			\end{enumerate}
		\item By propositional logic and epistemic antitonicity.
		\item[$\bullet$] 9--10 and 17 follow directly from epistemic antitonicity and the corresponding pairing laws and signature synthesis, respectively.
		\item[11.] By propositional logic directly from proof extension left and right.
		\item[$\bullet$] 12--13 and 15--16 follow directly from epistemic bitonicity and the corresponding pairing laws 
							by propositional logic.
		\item[14.] 
			\begin{enumerate}
				\item\quad$\LiPded\knows{a}{M}\limp\knows{a}{M'}$\hfill hypothesis
				\item\quad$\LiPded(\knows{a}{M}\limp\knows{a}{M'})\lequiv(\knows{a}{\pair{M}{M'}}\lequiv\knows{a}{M})$\hfill neutral pair elements
				\item\quad$\LiPded\knows{a}{\pair{M}{M'}}\lequiv\knows{a}{M}$\hfill a, b, PL
				\item\quad$\set{\knows{a}{\pair{M}{M'}}\lequiv\knows{a}{M}}\LiPded\proves{\pair{M}{M'}}{\phi}{a}{\community}\lequiv\proves{M}{\phi}{a}{\community}$\hfill epistemic bitonicity
				\item\quad$\LiPded\proves{\pair{M}{M'}}{\phi}{a}{\community}\lequiv\proves{M}{\phi}{a}{\community}$\hfill c, d, PL
				\item $\set{\knows{a}{M}\limp\knows{a}{M'}}\LiPded\proves{\pair{M}{M'}}{\phi}{a}{\community}\lequiv\proves{M}{\phi}{a}{\community}$\hfill a--e, definition.
			\end{enumerate}
		\item[18.] 
			\begin{enumerate}
				\item $\LiPded((\proves{M}{\phi}{a}{\community})\lor
						\proves{b}{\phi}{a}{\community})\limp\proves{\pair{M}{b}}{\phi}{a}{\community}$\hfill
							proof extension
				\item $\LiPded\knows{a}{\sign{M}{b}}\limp\knows{a}{\pair{M}{b}}$\hfill signature analysis
				\item $\LiPded(\proves{\pair{M}{b}}{\phi}{a}{\community})\limp\proves{\sign{M}{b}}{\phi}{a}{\community}$\hfill b, epistemic antitoncity
				\item $\LiPded((\proves{M}{\phi}{a}{\community})\lor
						\proves{b}{\phi}{a}{\community})\limp\proves{\sign{M}{b}}{\phi}{a}{\community}$\hfill
							a, c, PL.
			\end{enumerate}
		\item[19.] 
			\begin{enumerate}
				\item $\LiPded(\proves{\sign{M}{a}}{\phi}{a}{\community})\limp\proves{M}{\phi}{a}{\community}$\hfill self-signing elimination
				\item $\LiPded((\proves{M}{\phi}{a}{\community})\lor\proves{a}{\phi}{a}{\community})\limp\proves{\sign{M}{a}}{\phi}{a}{\community}$\hfill signing introduction
				\item $\LiPded(\proves{M}{\phi}{a}{\community})\limp\proves{\sign{M}{a}}{\phi}{a}{\community}$\hfill b, PL
				\item $\LiPded(\proves{\sign{M}{a}}{\phi}{a}{\community})\lequiv\proves{M}{\phi}{a}{\community}$\hfill a, c, PL.
			\end{enumerate}
		\item[20.] Suppose that $\agents=\set{a}$ and $\messages\setminus\set{\Kcomb{a},\Scomb{a}}$.
			\begin{enumerate}
				\item Let us proceed by induction over $M\in\messages$.
					\begin{enumerate} 
						\item \emph{base case,} that is,  $M\defeq b$, for $b\in\agents$.
								Hence $b=a$, and 
								thus $\LiPded\knows{a}{b}$ because $\LiPded\knows{a}{a}$.
						\item \emph{inductive step} for $M\defeq\sign{M'}{b}$, for $M'\in\messages$ and $b\in\agents$.
								Hence $b=a$, and
								thus $\LiPded\knows{a}{M'}\limp\knows{a}{\sign{M'}{b}}$ because 
									$\LiPded\knows{a}{M'}\limp\knows{a}{\sign{M'}{a}}$.
								Suppose that $\LiPded\knows{a}{M'}$.
								Hence $\LiPded\knows{a}{\sign{M'}{b}}$, by \emph{modus ponens}.
						\item \emph{inductive step} for $M\defeq\pair{M'}{M''}$, for $M',M''\in\messages$.
								Suppose that $\LiPded\knows{a}{M'}$ and $\LiPded\knows{a}{M''}$.
								Hence $\LiPded\knows{a}{M'}\land\knows{a}{M''}$, by propositional logic.
								Now, $\LiPded(\knows{a}{M'}\land\knows{a}{M''})\limp\knows{a}{\pair{M'}{M''}}$, and
								hence $\LiPded\knows{a}{\pair{M'}{M''}}$, by \emph{modus ponens}.
					\end{enumerate}
				\item \begin{enumerate}
						\item $\LiPded\knows{a}{\pair{M}{M'}}$\hfill total knowledge
						\item $\LiPded\knows{a}{M}\land\knows{a}{M'}$\hfill i, unpairing
						\item $\LiPded\knows{a}{M}\lequiv\knows{a}{M'}$\hfill ii, propositional logic.
						\end{enumerate}
				\item Jointly from b and epistemic bitonicity by propositional logic.
			\end{enumerate}
	\end{enumerate}

\subsection{Proof of Theorem~\ref{theorem:SomeUsefulDeducibleLogicalLaws}}\label{appendix:LogicalProofs}
	\begin{enumerate}			
		\item 	
			\begin{enumerate}
				\item$\LiPded(\proves{M}{(\phi\limp\phi')}{a}{\community})\limp
							\proves{\pair{M}{M'}}{(\phi\limp\phi')}{a}{\community}$\hfill proof extension
				\item$\LiPded(\proves{\pair{M}{M'}}{(\phi\limp\phi')}{a}{\community})\limp
							((\proves{\pair{M}{M'}}{\phi}{a}{\community})\limp\proves{\pair{M}{M'}}{\phi'}{a}{\community})$\hfill K 
				\item$\LiPded(\proves{M}{(\phi\limp\phi')}{a}{\community})\limp
							((\proves{\pair{M}{M'}}{\phi}{a}{\community})\limp\proves{\pair{M}{M'}}{\phi'}{a}{\community})$\hfill a, b, PL
				\item$\LiPded(\proves{M'}{\phi}{a}{\community})\limp
							\proves{\pair{M'}{M}}{\phi}{a}{\community}$\hfill proof extension
				\item$\LiPded(\proves{\pair{M'}{M}}{\phi}{a}{\community})\lequiv
							\proves{\pair{M}{M'}}{\phi}{a}{\community}$\hfill proof commutativity
				\item$\LiPded(\proves{M'}{\phi}{a}{\community})\limp
							\proves{\pair{M}{M'}}{\phi}{a}{\community}$\hfill d, e, PL
				\item$\LiPded(\proves{M}{(\phi\limp\phi')}{a}{\community})\limp
							((\proves{M'}{\phi}{a}{\community})\limp\proves{\pair{M}{M'}}{\phi'}{a}{\community})$\hfill c, f, PL.
			\end{enumerate}
		\item 
			\begin{enumerate}
				\item\quad$\set{\phi\limp\phi'}\subseteq\LiP$\hfill hypothesis
				\item\quad$(\phi\limp\phi')\in\LiP$\hfill a, definition
				\item\quad$\LiPded\phi\limp\phi'$\hfill b, definition
				\item\quad$\LiPded\proves{M}{(\phi\limp\phi')}{a}{\community}$\hfill c, N
				\item\quad$\LiPded(\proves{M}{(\phi\limp\phi')}{a}{\community})\limp
				((\proves{M}{\phi}{a}{\community})\limp\proves{M}{\phi'}{a}{\community})$
					\hfill K
				\item\quad$\LiPded(\proves{M}{\phi}{a}{\community})\limp\proves{M}{\phi'}{a}{\community}$
					\hfill d, e, PL
				\item\quad$((\proves{M}{\phi}{a}{\community})\limp\proves{M}{\phi'}{a}{\community})\in\LiP$
					\hfill f, definition
				\item $\set{\phi\limp\phi'}\LiPded
							(\proves{M}{\phi}{a}{\community})\limp\proves{M}{\phi'}{a}{\community}$
						\hfill a--g, definition.
			\end{enumerate}
		\item 
			\begin{enumerate}
				\item\quad$\set{\phi\lequiv\phi'}\subseteq\LiP$\hfill hypothesis
				\item\quad$(\phi\lequiv\phi')\in\LiP$\hfill a, definition
				\item\quad$\LiPded\phi\lequiv\phi'$\hfill b, definition
				\item\quad$\LiPded\phi\limp\phi'$\hfill c, PL
				\item\quad$\LiPded(\proves{M}{\phi}{a}{\community})\limp\proves{M}{\phi'}{a}{\community}$
					\hfill d, R
				\item\quad$\LiPded\phi'\limp\phi$\hfill c, PL
				\item\quad$\LiPded(\proves{M}{\phi'}{a}{\community})\limp\proves{M}{\phi}{a}{\community}$
					\hfill f, R
				\item\quad$\LiPded(\proves{M}{\phi}{a}{\community})\lequiv\proves{M}{\phi'}{a}{\community}$
					\hfill e, g, PL
				\item\quad$((\proves{M}{\phi}{a}{\community})\lequiv\proves{M}{\phi'}{a}{\community})\in\LiP$
					\hfill h, definition
				\item $\set{\phi\lequiv\phi'}\LiPded
							(\proves{M}{\phi}{a}{\community})\lequiv\proves{M}{\phi'}{a}{\community}$
						\hfill a--i, definition.
			\end{enumerate}
		\item By regularity, epistemic antitonicity, and the transitivity of `$\limp$'.
		\item By epistemic regularity and propositional logic.
		\item 
			\begin{enumerate} 
				\item$\LiPded(\proves{M}{(\phi'\limp(\phi\land\phi'))}{a}{\community})\limp((\proves{M'}{\phi'}{a}{\community})\limp\proves{\pair{M}{M'}}{(\phi\land\phi')}{a}{\community})$\hfill GK
				\item$\LiPded\phi\limp(\phi'\limp(\phi\land\phi'))$\hfill PT
				\item$\LiPded(\proves{M}{\phi}{a}{\community})\limp\proves{M}{(\phi'\limp(\phi\land\phi'))}{a}{\community}$\hfill b, R
				\item$\LiPded(\proves{M}{\phi}{a}{\community})\limp((\proves{M'}{\phi'}{a}{\community})\limp\proves{\pair{M}{M'}}{(\phi\land\phi')}{a}{\community})$\hfill a, c, PL
				\item$\LiPded((\proves{M}{\phi}{a}{\community})\land\proves{M'}{\phi'}{a}{\community})\limp
						\proves{\pair{M}{M'}}{(\phi\land\phi')}{a}{\community}$\hfill d, PL.
			\end{enumerate}
		\item 
			\begin{enumerate} 
				\item$\LiPded((\proves{M}{\phi}{a}{\community})\land\proves{M}{\phi'}{a}{\community})\limp
						\proves{\pair{M}{M}}{(\phi\land\phi')}{a}{\community}$\hfill proof conjunctions
				\item $\LiPded(\proves{\pair{M}{M}}{(\phi\land\phi')}{a}{\community})\lequiv\proves{M}{(\phi\land\phi')}{a}{\community}$\hfill proof idempotency
				\item$\LiPded((\proves{M}{\phi}{a}{\community})\land\proves{M}{\phi'}{a}{\community})\limp
						\proves{M}{(\phi\land\phi')}{a}{\community}$\hfill a, b, PL
				\item$\LiPded(\phi\land\phi')\limp\phi$\hfill PT
				\item$\LiPded\proves{M}{(\phi\land\phi')}{a}{\community}\limp\proves{M}{\phi}{a}{\community}$\hfill d, R
				\item$\LiPded(\phi\land\phi')\limp\phi'$\hfill PT
				\item$\LiPded(\proves{M}{(\phi\land\phi')}{a}{\community})\limp\proves{M}{\phi'}{a}{\community}$\hfill f, R
				\item$\LiPded(\proves{M}{(\phi\land\phi')}{a}{\community})\limp
						((\proves{M}{\phi}{a}{\community})\land\proves{M}{\phi'}{a}{\community})$\hfill e, g, PL
				\item$\LiPded((\proves{M}{\phi}{a}{\community})\land\proves{M}{\phi'}{a}{\community})\lequiv
						\proves{M}{(\phi\land\phi')}{a}{\community}$\hfill c, h, PL.				
			\end{enumerate}
		\item 
			\begin{enumerate}
				\item $\LiPded(\proves{M}{\phi}{a}{\community})\limp\proves{\pair{M}{M'}}{\phi}{a}{\community}$\hfill proof extension, right
				\item $\LiPded\phi\limp(\phi\lor\phi')$\hfill PT
				\item $\LiPded(\proves{\pair{M}{M'}}{\phi}{a}{\community})\limp\proves{\pair{M}{M'}}{(\phi\lor\phi')}{a}{\community}$\hfill b, R
				\item $\LiPded(\proves{M}{\phi}{a}{\community})\limp\proves{\pair{M}{M'}}{(\phi\lor\phi')}{a}{\community}$\hfill
						a, c, PL
				\item $\LiPded(\proves{M'}{\phi'}{a}{\community})\limp\proves{\pair{M}{M'}}{\phi'}{a}{\community}$\hfill proof extension, left
				\item $\LiPded\phi'\limp(\phi\lor\phi')$\hfill PT
				\item $\LiPded(\proves{\pair{M}{M'}}{\phi'}{a}{\community})\limp\proves{\pair{M}{M'}}{(\phi\lor\phi')}{a}{\community}$\hfill f, R
				\item $\LiPded(\proves{M'}{\phi'}{a}{\community})\limp\proves{\pair{M}{M'}}{(\phi\lor\phi')}{a}{\community}$\hfill
						e, g, PL
				\item $\LiPded((\proves{M}{\phi}{a}{\community})\lor(\proves{M'}{\phi'}{a}{\community}))\limp\proves{\pair{M}{M'}}{(\phi\lor\phi')}{a}{\community}$\hfill d, h, PL.
			\end{enumerate}
		\item 
			\begin{enumerate}			
				\item $\LiPded((\proves{M}{\phi}{a}{\community})\lor\proves{M}{\phi'}{a}{\community})\limp\proves{\pair{M}{M}}{(\phi\lor\phi')}{a}{\community}$\hfill proof disjunctions
				\item $\LiPded(\proves{\pair{M}{M}}{(\phi\lor\phi')}{a}{\community})\lequiv\proves{M}{(\phi\lor\phi')}{a}{\community}$\hfill proof idempotency
				\item $\LiPded((\proves{M}{\phi}{a}{\community})\lor\proves{M}{\phi'}{a}{\community})\limp\proves{M}{(\phi\lor\phi')}{a}{\community}$\hfill a, b, PL.
			\end{enumerate}
		\item 
			\begin{enumerate}
				\item$\LiPded\knows{a}{a}$\hfill knowledge of one's own name
				\item$\LiPded\true$\hfill a, definition
				\item$\LiPded\proves{M}{\true}{a}{\community}$\hfill b, N.
			\end{enumerate}
		\item
			\begin{enumerate}
				\item$\LiPded(\proves{a}{\phi}{a}{\community})\limp(\knows{a}{a}\limp\phi)$\hfill epistemic truthfulness
				\item$\LiPded\knows{a}{a}$\hfill knowledge of one's own name
				\item$\LiPded\knows{a}{a}\limp((\knows{a}{a}\limp\phi)\limp\phi)$\hfill PT 
				\item$\LiPded(\knows{a}{a}\limp\phi)\limp\phi$\hfill b, c, PL
				\item$\LiPded(\proves{a}{\phi}{a}{\community})\limp\phi$\hfill a, d, PL.
			\end{enumerate}
		\item 
			\begin{enumerate} 
				\item$\set{\phi}\LiPded\proves{a}{\phi}{a}{\community}$\hfill N
				\item\quad$\set{\proves{a}{\phi}{a}{\community}}\subseteq\LiP$\hfill hypothesis
				\item\quad$(\proves{a}{\phi}{a}{\community})\in\LiP$\hfill b, definition
				\item\quad$\LiPded\proves{a}{\phi}{a}{\community}$\hfill c, definition
				\item\quad$\LiPded(\proves{a}{\phi}{a}{\community})\limp\phi$\hfill self-truthfulness
				\item\quad$\LiPded\phi$\hfill d, e, PL
				\item$\set{\proves{a}{\phi}{a}{\community}}\LiPded\phi$\hfill b--f, definition
				\item$\phi\LiPdedBis\proves{a}{\phi}{a}{\community}$\hfill a, g, definition.
			\end{enumerate}
		\item
			\begin{enumerate} 
				\item$\LiPded\proves{M}{\true}{a}{\community}$\hfill anything can prove tautological truth
				\item$\LiPded(\proves{M}{\true}{a}{\community})\limp\neg(\proves{M}{\neg\true}{a}{\community})$\hfill proof consistency
				\item$\LiPded\neg(\proves{M}{\neg\true}{a}{\community})$\hfill a, b, PL
				\item$\LiPded\neg(\proves{M}{\false}{a}{\community})$\hfill c, definition.
			\end{enumerate}	
		\item 
			\begin{enumerate}
				\item$\LiPded(\proves{M}{\true}{a}{\community})\limp
								\bigwedge_{b\in\community\cup\set{a}}(\proves{\sign{M}{a}}{(\knows{a}{M}\land\proves{M}{\true}{a}{\community})}{b}{\community\cup\set{a}})$\hfill peer review
				\item$\LiPded\proves{M}{\true}{a}{\community}$\hfill anything can prove tautological truth
				\item$\LiPded\bigwedge_{b\in\community\cup\set{a}}(\proves{\sign{M}{a}}{(\knows{a}{M}\land\proves{M}{\true}{a}{\community})}{b}{\community\cup\set{a}})$\hfill a, b, PL
				\item$\LiPded(\proves{\sign{M}{a}}{(\knows{a}{M}\land\proves{M}{\true}{a}{\community})}{b}{\community\cup\set{a}})\limp\proves{\sign{M}{a}}{\knows{a}{M}}{b}{\community\cup\set{a}}$\hfill proof conjunction \emph{bis}
				\item $\LiPded\bigwedge_{b\in\community\cup\set{a}}(\proves{\sign{M}{a}}{\knows{a}{M}}{b}{\community\cup\set{a}})$\hfill c, d, PL.
			\end{enumerate} 
		\item
			\begin{enumerate}
				\item$\LiPded\proves{\sign{M}{a}}{\knows{a}{M}}{a}{\emptyset\cup\set{a}}$\hfill authentic knowledge
				\item$\LiPded(\proves{\sign{M}{a}}{\knows{a}{M}}{a}{\emptyset\cup\set{a}})\limp
						\proves{\sign{M}{a}}{\knows{a}{M}}{a}{\emptyset}$\hfill group decomposition
				\item$\LiPded\proves{\sign{M}{a}}{\knows{a}{M}}{a}{\emptyset}$\hfill a, b, PL
				\item$\LiPded(\proves{\sign{M}{a}}{\knows{a}{M}}{a}{\emptyset})\limp
							\proves{M}{\knows{a}{M}}{a}{\emptyset}$\hfill self-signing elimination
				\item$\LiPded\proves{M}{\knows{a}{M}}{a}{\emptyset}$\hfill c, d, PL.
			\end{enumerate}	
		\item 
			\begin{enumerate}
				\item \quad $\LiPded\knows{a}{M}\limp\phi$\hfill hypothesis 
				\item \quad $\LiPded(\proves{\sign{M}{a}}{\knows{a}{M}}{b}{\community\cup\set{a}})\limp
										\proves{\sign{M}{a}}{\phi}{b}{\community\cup\set{a}}$\hfill a, R
				\item \quad $\LiPded\bigwedge_{b\in\community\cup\set{a}}
								(\proves{\sign{M}{a}}{\knows{a}{M}}{b}{\community\cup\set{a}})$\hfill 
									authentic knowledge
				\item \quad $\LiPded\bigwedge_{b\in\community\cup\set{a}}
								(\proves{\sign{M}{a}}{\phi}{b}{\community\cup\set{a}})$\hfill b, c, PL
				\smallskip
				\item \quad $\LiPded\bigwedge_{b\in\community\cup\set{a}}
								(\proves{\sign{M}{a}}{\phi}{b}{\community\cup\set{a}})$\hfill hypothesis
				\item \quad $\LiPded\proves{\sign{M}{a}}{\phi}{a}{\community\cup\set{a}}$\hfill e, PL
				\item \quad $\LiPded\proves{\sign{M}{a}}{\phi}{a}{\community\cup\set{a}}\limp(\knows{a}{\sign{M}{a}}\limp\phi)$\hfill epistemic truthfulness
				\item \quad $\LiPded\knows{a}{\sign{M}{a}}\limp\phi$\hfill f, g, PL
				\item \quad $\LiPded\knows{a}{M}\limp\knows{a}{\sign{M}{a}}$\hfill signature synthesis
				\item \quad $\LiPded\knows{a}{M}\limp\phi$\hfill h, i, PL
				\smallskip
				\item $\LiPded\knows{a}{M}\limp\phi$ iff 
							$\LiPded\bigwedge_{b\in\community\cup\set{a}}
								(\proves{\sign{M}{a}}{\phi}{b}{\community\cup\set{a}})$\hfill a--d, e--j, PL
				\item $\knows{a}{M}\limp\phi\LiPdedBis
							\bigwedge_{b\in\community\cup\set{a}}
								(\proves{\sign{M}{a}}{\phi}{b}{\community\cup\set{a}})$\hfill k, definition.
			\end{enumerate}
		\item By instantiating $\phi$ with $\knows{a}{M'}$ in AEN.
		\item Like AEN but by invoking self-knowledge instead of authentic knowledge.
		\item By instantiating $\phi$ with $\knows{a}{M'}$ in SEN.
		\item 
			\begin{enumerate}
				\item $\LiPded\bigwedge_{a\in\community\cup\set{b}}
							(\proves{\sign{M}{b}}{\knows{b}{M}}{a}{\community\cup\set{b}})$\hfill authentic knowledge
				\item $\LiPded\knows{a}{\sign{M}{b}}\limp\knows{b}{M}$\hfill a, AEN \emph{bis}.
			\end{enumerate}
		\item 
			\begin{enumerate}
				\item $\LiPded(\proves{M}{\phi}{a}{\community})\limp\bigwedge_{b\in\community\cup\set{a}}(\proves{\sign{M}{a}}{(\knows{a}{M}\land\proves{M}{\phi}{a}{\community})}{b}{\community\cup\set{a}})$\hfill peer review
				\item$\LiPded(\proves{M}{\phi}{a}{\community})\limp(\knows{a}{M}\limp\phi)$\hfill epistemic truthfulness 
				\item$\LiPded(\knows{a}{M}\land\proves{M}{\phi}{a}{\community})\limp\phi$\hfill b, PL
				\item$\LiPded(\proves{\sign{M}{a}}{(\knows{a}{M}\land\proves{M}{\phi}{a}{\community})}{b}{\community\cup\set{a}})\limp\proves{\sign{M}{a}}{\phi}{b}{\community\cup\set{a}}$\hfill c, R
				\item $\LiPded(\proves{M}{\phi}{a}{\community})\limp\bigwedge_{b\in\community\cup\set{a}}(\proves{\sign{M}{a}}{\phi}{b}{\community\cup\set{a}})$\hfill a, d, PL.
			\end{enumerate}
		\item 
			\begin{enumerate}
				\item $\LiPded(\proves{M}{\phi}{a}{\community\cup\community'})\limp(\proves{M}{\phi}{a}{\community})$\hfill
						group decomposition
				\item $\LiPded(\proves{M}{\phi}{a}{\community\cup\community'})\lequiv
								(\proves{M}{\phi}{a}{\community'\cup\community})$\hfill $\community\cup\community'=\community'\cup\community$
				\item $\LiPded(\proves{M}{\phi}{a}{\community'\cup\community})\limp(\proves{M}{\phi}{a}{\community'})$\hfill
						group decomposition
				\item $\LiPded(\proves{M}{\phi}{a}{\community\cup\community'})\limp(\proves{M}{\phi}{a}{\community'})$\hfill
						b, c, PL
				\item $\LiPded(\proves{M}{\phi}{a}{\community\cup\community'})\limp
						((\proves{M}{\phi}{a}{\community})\land\proves{M}{\phi}{a}{\community'})$\hfill a, d, PL.
			\end{enumerate}
		\item 
			\begin{enumerate}
				\item $\LiPded(\proves{M}{\phi}{a}{\community\cup\set{a}})\limp(\proves{M}{\phi}{a}{\community})$\hfill group decomposition
				\item $\LiPded(\proves{M}{\phi}{a}{\community})\limp\bigwedge_{b\in\community\cup\set{a}}(\proves{\sign{M}{a}}{\phi}{b}{\community\cup\set{a}})$\hfill simple peer review
				\item $\LiPded(\proves{M}{\phi}{a}{\community})\limp\proves{\sign{M}{a}}{\phi}{a}{\community\cup\set{a}}$\hfill $a\in\community\cup\set{a}$, b, PL
				\item $\LiPded(\proves{\sign{M}{a}}{\phi}{a}{\community\cup\set{a}})\lequiv\proves{M}{\phi}{a}{\community\cup\set{a}}$\hfill self-signing idempotency
				\item $\LiPded(\proves{M}{\phi}{a}{\community})\limp\proves{M}{\phi}{a}{\community\cup\set{a}}$\hfill c, d, PL
				\item $\LiPded(\proves{M}{\phi}{a}{\community\cup\set{a}})\lequiv\proves{M}{\phi}{a}{\community}$\hfill a, e, PL.
			\end{enumerate}
		\item 
			\begin{enumerate}
				\item $\LiPded\knows{a}{M}\limp((\proves{M}{\phi}{a}{\community})\limp\phi)$\hfill epistemic truthfulness, PL
				\item $\LiPded(\proves{\sign{M}{a}}{\knows{a}{M}}{a}{\community})\limp\proves{\sign{M}{a}}{((\proves{M}{\phi}{a}{\community})\limp\phi)}{a}{\community}$\hfill a, R
				\item $\LiPded\proves{\sign{M}{a}}{\knows{a}{M}}{a}{\community\cup\set{a}}$\hfill authentic knowledge
				\item $\LiPded(\proves{\sign{M}{a}}{\knows{a}{M}}{a}{\community\cup\set{a}})\lequiv
								\proves{\sign{M}{a}}{\knows{a}{M}}{a}{\community}$\hfill self-neutral group element
				\item $\LiPded\proves{\sign{M}{a}}{\knows{a}{M}}{a}{\community}$\hfill c, d, PL
				\item $\LiPded\proves{\sign{M}{a}}{((\proves{M}{\phi}{a}{\community})\limp\phi)}{a}{\community}$\hfill b, e, PL
				\item $\LiPded(\proves{\sign{M}{a}}{((\proves{M}{\phi}{a}{\community})\limp\phi)}{a}{\community})\lequiv\proves{M}{((\proves{M}{\phi}{a}{\community})\limp\phi)}{a}{\community}$\hfill self-signing idempotency 
				\item $\LiPded\proves{M}{((\proves{M}{\phi}{a}{\community})\limp\phi)}{a}{\community}$\hfill f, g, PL.
			\end{enumerate}
		\item By the necessitation of the law that nothing can prove falsehood.
		\item
			\begin{enumerate}
				\item $\LiPded(\proves{M}{\phi}{a}{\community})\limp\bigwedge_{b\in\community\cup\set{a}}(\proves{\sign{M}{a}}{(\knows{a}{M}\land\proves{M}{\phi}{a}{\community})}{b}{\community\cup\set{a}})$\hfill peer review
				\item\quad$b\in\community\cup\set{a}$\hfill hypothesis
				\item\quad$\LiPded(\proves{M}{\phi}{a}{\community})\limp\bigwedge_{b\in\community\cup\set{a}}(\proves{\sign{M}{a}}{\phi}{b}{\community\cup\set{a}})$\hfill simple peer review
				\item\quad$\LiPded(\knows{a}{M}\land\proves{M}{\phi}{a}{\community})\limp\bigwedge_{b\in\community\cup\set{a}}(\proves{\sign{M}{a}}{\phi}{b}{\community\cup\set{a}})$\hfill c, PL
				\item\quad$\LiPded(\begin{array}[t]{@{}l@{}}
							\proves{\sign{M}{a}}{(\knows{a}{M}\land\proves{M}{\phi}{a}{\community})}{b}{\community\cup\set{a}})\limp\\
							\proves{\sign{M}{a}}{(\bigwedge_{b\in\community\cup\set{a}}(\proves{\sign{M}{a}}{\phi}{b}{\community\cup\set{a}}))}{b}{\community\cup\set{a}}
							\end{array}$\hfill d, R
				\item $\LiPded\bigwedge_{b\in\community\cup\set{a}}(\begin{array}[t]{@{}l@{}}
							(\proves{\sign{M}{a}}{(\knows{a}{M}\land\proves{M}{\phi}{a}{\community})}{b}{\community\cup\set{a}})\limp\\
							\proves{\sign{M}{a}}{(\bigwedge_{b\in\community\cup\set{a}}(\proves{\sign{M}{a}}{\phi}{b}{\community\cup\set{a}}))}{b}{\community\cup\set{a}})
							\end{array}$\hfill b--e, PL
				\item $\LiPded(\begin{array}[t]{@{}l@{}}
								\proves{M}{\phi}{a}{\community})\limp\\
								\bigwedge_{b\in\community\cup\set{a}}(\proves{\sign{M}{a}}{(\bigwedge_{b\in\community\cup\set{a}}(\proves{\sign{M}{a}}{\phi}{b}{\community\cup\set{a}}))}{b}{\community\cup\set{a}})
							\end{array}$\hfill a, f, PL
				\item $\LiPded(\begin{array}[t]{@{}l@{}}
								\proves{M}{\phi}{a}{\community})\limp\\
								\proves{\sign{M}{a}}{(\bigwedge_{b\in\community\cup\set{a}}(\proves{\sign{M}{a}}{\phi}{b}{\community\cup\set{a}}))}{a}{\community\cup\set{a}}
								\end{array}$\hfill $a\in\community\cup\set{a}$, g, PL
				\item $\LiPded\begin{array}[t]{@{}l@{}}
						\proves{\sign{M}{a}}{(\bigwedge_{b\in\community\cup\set{a}}(\proves{\sign{M}{a}}{\phi}{b}{\community\cup\set{a}}))}{a}{\community\cup\set{a}}\lequiv\\
						\proves{\sign{M}{a}}{(\bigwedge_{b\in\community\cup\set{a}}(\proves{\sign{M}{a}}{\phi}{b}{\community\cup\set{a}}))}{a}{\community}
						\end{array}$\hfill self-neutral group element
				\item $\LiPded(\proves{M}{\phi}{a}{\community})\limp\proves{\sign{M}{a}}{(\bigwedge_{b\in\community\cup\set{a}}(\proves{\sign{M}{a}}{\phi}{b}{\community\cup\set{a}}))}{a}{\community}$\hfill h, i, PL
				\item $\LiPded(\begin{array}[t]{@{}l@{}}\proves{\sign{M}{a}}{(\bigwedge_{b\in\community\cup\set{a}}(\proves{\sign{M}{a}}{\phi}{b}{\community\cup\set{a}}))}{a}{\community})\lequiv\\
							\proves{M}{(\bigwedge_{b\in\community\cup\set{a}}(\proves{\sign{M}{a}}{\phi}{b}{\community\cup\set{a}}))}{a}{\community}\end{array}$\hfill self-signing idempotency
				\item $\LiPded(\proves{M}{\phi}{a}{\community})\limp\proves{M}{(\bigwedge_{b\in\community\cup\set{a}}(\proves{\sign{M}{a}}{\phi}{b}{\community\cup\set{a}}))}{a}{\community}$\hfill j, k, PL.
			\end{enumerate}
		\item \begin{enumerate}
				\item $\LiPded(\proves{M}{\phi}{a}{\community})\limp\proves{M}{(\proves{\sign{M}{a}}{\phi}{a}{\community\cup\set{a}})}{a}{\community}$\hfill simple peer review \emph{bis}
				\item $\LiPded(\proves{\sign{M}{a}}{\phi}{a}{\community\cup\set{a}})\lequiv\proves{\sign{M}{a}}{\phi}{a}{\community}$\hfill self-neutral group element
				\item $\LiPded(\proves{M}{(\proves{\sign{M}{a}}{\phi}{a}{\community\cup\set{a}})}{a}{\community})\lequiv\proves{M}{(\proves{\sign{M}{a}}{\phi}{a}{\community})}{a}{\community}$\hfill b, R \emph{bis}
				\item $\LiPded(\proves{M}{\phi}{a}{\community})\limp\proves{M}{(\proves{\sign{M}{a}}{\phi}{a}{\community})}{a}{\community}$\hfill a, c, PL
				\item $\LiPded(\proves{\sign{M}{a}}{\phi}{a}{\community})\lequiv\proves{M}{\phi}{a}{\community}$\hfill self-signing idempotency
				\item $\LiPded(\proves{M}{(\proves{\sign{M}{a}}{\phi}{a}{\community})}{a}{\community})\lequiv\proves{M}{(\proves{M}{\phi}{a}{\community})}{a}{\community}$\hfill e, R \emph{bis}
				\item $\LiPded(\proves{M}{\phi}{a}{\community})\limp\proves{M}{(\proves{M}{\phi}{a}{\community})}{a}{\community}$\hfill d, f, PL
				\item $\LiPded\proves{M}{((\proves{M}{\phi}{a}{\community})\limp\phi)}{a}{\community}$\hfill 
						self-proof of truthfulness
				\item $\LiPded(\proves{M}{((\proves{M}{\phi}{a}{\community})\limp\phi)}{a}{\community})\limp((\proves{M}{(\proves{M}{\phi}{a}{\community})}{a}{\community})\limp
								\proves{M}{\phi}{a}{\community})$\hfill K 
				\item $\LiPded(\proves{M}{(\proves{M}{\phi}{a}{\community})}{a}{\community})\limp
								\proves{M}{\phi}{a}{\community}$\hfill h, i, PL
				\item $\LiPded(\proves{M}{(\proves{M}{\phi}{a}{\community})}{a}{\community})\lequiv
				\proves{M}{\phi}{a}{\community}$\hfill g, j, PL.
			\end{enumerate}
		\item Like the proofs of self-truthfulness and self-truthfulness \emph{bis,} 
				but by invoking the law of total knowledge.
	\end{enumerate}

\subsection{Proof of Theorem~\ref{theorem:ChurchRosser}}
	By induction on the structure of $T:$
	\begin{itemize}
		\item \emph{Base cases (\,$T\in\set{x}_{x\in\mathcal{X}}\cup\set{\Kcomb{a},\Scomb{a}}_{a\in\agents}$\,):}
				let $T_{1},T_{2}\in\mathcal{T}$ and $\alpha_{1},\alpha_{2}\in\agents^{*}$, and
				suppose that $T\reduc{\alpha_{1}}T_{1}$ and $T\reduc{\alpha_{2}}T_{2}$.
				Hence, 
					$\alpha_{1}=\alpha_{2}=\epsilon$ and
					$T_{1}=T_{2}=T$, because
						term variables $x$ and basic combinators $\Kcomb{a}$ and $\Scomb{a}$ 
					reduce only through the empty word $\epsilon$ --- to themselves only.
				Thus $T_{1}\reduc{\epsilon}T$ and $T_{2}\reduc{\epsilon}T$.
				Hence,
					there are $T'\in\mathcal{T}$ and $\alpha_{1}',\alpha_{2}'\in\agents^{*}$ such that
						$T_{1}\reduc{\alpha_{1}'}T'$ and 
						$T_{2}\reduc{\alpha_{2}'}T'$.
				
		\item \emph{Inductive step (\,$T=\pair{U}{V}$ for $U,V\in\mathcal{T}$):}
				let us make the induction hypotheses about $U$ and $V$ that 
					\begin{enumerate}
						\item for all $U_{1},U_{2}\in\mathcal{T}$ and $\beta_{1},\beta_{2}\in\agents^{*}$,
								if $U\reduc{\beta_{1}}U_{1}$ and $U\reduc{\beta_{2}}U_{2}$ 
								then there are 
									$U'\in\mathcal{T}$ and 
									$\beta_{1}',\beta_{2}'\in\agents^{*}$ such that
								$U_{1}\reduc{\beta_{1}'}U'$ and 
								$U_{2}\reduc{\beta_{2}'}U'$ 
						\item for all $V_{1},V_{2}\in\mathcal{T}$ and $\gamma_{1},\gamma_{2}\in\agents^{*}$,
								if $V\reduc{\gamma_{1}}V_{1}$ and $V\reduc{\gamma_{2}}V_{2}$ 
								then there are 
									$V'\in\mathcal{T}$ and 
									$\gamma_{1}',\gamma_{2}'\in\agents^{*}$ such that
								$V_{1}\reduc{\gamma_{1}'}V'$ and 
								$V_{2}\reduc{\gamma_{2}'}V'$.
					\end{enumerate}
				Now, 
					let $T_{1},T_{2}\in\mathcal{T}$ and $\alpha_{1},\alpha_{2}\in\agents^{*}$ and
					suppose that 
						$\pair{U}{V}\reduc{\alpha_{1}}T_{1}$ and 
						$\pair{U}{V}\reduc{\alpha_{2}}T_{2}$.
				Let us proceed by disjunctive case analysis and 
					\begin{enumerate}
						\item suppose 
							that $V$ is not consumed by $U$ in $\pair{U}{V}\reduc{\alpha_{1}}T_{1}$ and 
							that $V$ is not consumed by $U$ in $\pair{U}{V}\reduc{\alpha_{2}}T_{2}$.
							That is, $U$ and $V$ reduce independently in both cases.
							Hence,
								there are 
									$U_{1},V_{1},U_{2},V_{2}\in\mathcal{T}$ such that 
										$T_{1}=\pair{U_{1}}{V_{1}}$ and 
										$T_{2}=\pair{U_{2}}{V_{2}}$ as well as
									$\beta_{1},\gamma_{1},\beta_{2},\gamma_{2}\in\agents^{*}$ such that 
										$U\reduc{\beta_{1}}U_{1}$ and $V\reduc{\gamma_{1}}V_{1}$ and
										$U\reduc{\beta_{2}}U_{2}$ and $V\reduc{\gamma_{2}}V_{2}$.
						\item suppose 
							that $V$ is not consumed by $U$ in $\pair{U}{V}\reduc{\alpha_{1}}T_{1}$ but  
							that $V$ is consumed by $U$ in $\pair{U}{V}\reduc{\alpha_{2}}T_{2}$.
							That is, $U$ and $V$ reduce independently in the first but not in the second case.
							Hence on the one hand, there are 
									$U_{1},V_{1}\in\mathcal{T}$ such that 
										$T_{1}=\pair{U_{1}}{V_{1}}$ as well as
									$\beta_{1},\gamma_{1}\in\agents^{*}$ such that 
										$U\reduc{\beta_{1}}U_{1}$ and $V\reduc{\gamma_{1}}V_{1}$.
							On the other hand, $V$ can be consumed by $U$ 
								either as part of a term $\pair{\pair{\Kcomb{a}}{W}}{V}$
								or as part of a term $\pair{\pair{\pair{\Scomb{a}}{W}}{X}}{V}$.
							In both cases,
								there are 
									$U_{2}\in\mathcal{T}$ ($\pair{\Kcomb{a}}{W}$ and $\pair{\pair{\Scomb{a}}{W}}{X}$, respectively) and $\beta_{2}\in\agents^{*}$ such that 
										$U\reduc{\beta_{2}}U_{2}$ as well as
									$V_{2}=V\in\mathcal{T}$ and $\gamma_{2}=\epsilon\in\agents^{*}$ such that 
										$V\reduc{\gamma_{2}}V_{2}$.
						\item suppose 
									that $V$ is consumed by $U$ in $\pair{U}{V}\reduc{\alpha_{1}}T_{1}$ but 
									that $V$ is not consumed by $U$ in $\pair{U}{V}\reduc{\alpha_{2}}T_{2}$, and
							  proceed symmetrically to Case~2.
						\item suppose 
							that $V$ is consumed by $U$ in $\pair{U}{V}\reduc{\alpha_{1}}T_{1}$ and 
							that $V$ is consumed by $U$ in $\pair{U}{V}\reduc{\alpha_{2}}T_{2}$.
							On both hands,
								$V$ can be consumed by $U$ 
								either as part of a term $\pair{\pair{\Kcomb{a}}{W}}{V}$
								or as part of a term $\pair{\pair{\pair{\Scomb{a}}{W}}{X}}{V}$.
							In both cases of the first hand,
								there are 
									$U_{1}\in\mathcal{T}$ ($\pair{\Kcomb{a}}{W}$ and $\pair{\pair{\Scomb{a}}{W}}{X}$, respectively) and $\beta_{1}\in\agents^{*}$ such that 
										$U\reduc{\beta_{1}}U_{1}$ as well as
									$V_{1}=V\in\mathcal{T}$ and $\gamma_{1}=\epsilon\in\agents^{*}$ such that 
										$V\reduc{\gamma_{1}}V_{1}$.
							Similarly in both cases of the second hand,
								there are 
									$U_{2}\in\mathcal{T}$ ($\pair{\Kcomb{a}}{W}$ and $\pair{\pair{\Scomb{a}}{W}}{X}$, respectively) and $\beta_{2}\in\agents^{*}$ such that 
										$U\reduc{\beta_{2}}U_{2}$ as well as
									$V_{2}=V\in\mathcal{T}$ and $\gamma_{2}=\epsilon\in\agents^{*}$ such that 
										$V\reduc{\gamma_{2}}V_{2}$.
					\end{enumerate}
					Hence in all four cases,
						there are 
							$U'\in\mathcal{T}$ and 
							$\beta_{1}',\beta_{2}'\in\agents^{*}$ such that
								$U_{1}\reduc{\beta_{1}'}U'$ and 
								$U_{2}\reduc{\beta_{2}'}U'$, and
						there are 
							$V'\in\mathcal{T}$ and 
							$\gamma_{1}',\gamma_{2}'\in\agents^{*}$ such that
								$V_{1}\reduc{\gamma_{1}'}V'$ and 
								$V_{2}\reduc{\gamma_{2}'}V'$, by the induction hypothesis.
					Hence,
						$\pair{U'}{V'}\in\mathcal{T}$ and
						there are $\alpha_{1}',\alpha_{2}'\in\agents^{*}$ such that
							$\pair{U_{1}}{V_{1}}\reduc{\alpha_{1}'}\pair{U'}{V'}$ and 
							$\pair{U_{2}}{V_{2}}\reduc{\alpha_{2}'}\pair{U'}{V'}$.
					Thus,
						there are 
							$T'\in\mathcal{T}$ and 
							$\alpha_{1}',\alpha_{2}'\in\agents^{*}$ such that
								$T_{1}\reduc{\alpha_{1}'}T'$ and 
								$T_{2}\reduc{\alpha_{2}'}T'$,
									as required.
				(Note that in any case, $U$ cannot be consumed by $V$, by the definition of reduction.)
	\end{itemize}

\subsection{Proof of Theorem~\ref{theorem:TerminationSubjectReduction}}
	By induction on the (staged) definition of iCL-reduction $\reduc{\alpha}$\;:
	\begin{description}
		\item[(Stage $\alpha\in\agents$)] 
			thus there is $a\in\agents$ such that $a=\alpha$\,;
				thus $\underline{\alpha\in\agents^{*}}$. Further, 
			\begin{description}
				\item[(Base case $\Kcomb{a}$)] suppose that 
					$\pair{\pair{\Kcomb{a}}{T}}{T'}\reduc{a}T$ and 
					$\Gamma\vdash_{\mathrm{TiCL}}\linebreak\pair{\pair{\Kcomb{a}}{T}}{T'}:\varphi$.
					Further suppose that $\Gamma\vdash_{\mathrm{TiCL}}T':\varphi'$ for 
						some $\varphi'$ (without any restriction).
					Hence,
						$\Gamma\vdash_{\mathrm{TiCL}}\Kcomb{a}:\varphi\limp(\varphi'\limp\varphi)$ and
							$\underline{\Gamma\vdash_{\mathrm{TiCL}}T:\varphi}$, by TiCL-typing.
				\item[(Base case $\Scomb{a}$)] suppose that 
					$\pair{\pair{\pair{\Scomb{a}}{T}}{T'}}{T''}\reduc{a}\pair{\pair{T}{T''}}{\pair{T'}{T''}}$ 
					and 
					$\Gamma\vdash_{\mathrm{TiCL}}\pair{\pair{\pair{\Scomb{a}}{T}}{T'}}{T''}:\varphi$.
					Further suppose that 
						$\Gamma\vdash_{\mathrm{TiCL}}T'':\varphi''$ for 
							some $\varphi''$ (without any restriction).
					Hence:
						\begin{itemize}
							\item $\Gamma\vdash_{\mathrm{TiCL}}\mathtt{S}_{a}:((\varphi''\limp(\varphi'\limp\varphi))\limp((\varphi''\limp\varphi')\limp(\varphi''\limp\varphi)))$,
							\item $\Gamma\vdash_{\mathrm{TiCL}}T':\varphi''\limp\varphi'$, 
							\item $\Gamma\vdash_{\mathrm{TiCL}}T:\varphi''\limp(\varphi'\limp\varphi)$, and
							\item $\underline{\Gamma\vdash_{\mathrm{TiCL}}\pair{\pair{T}{T''}}{\pair{T'}{T''}}:\varphi}$, 
						\end{itemize} 
						by TiCL-typing.
				\item[(Inductive step)] suppose that for all $\varphi$, 
					if $T\reduc{a}T'$ and $\Gamma\vdash_{\mathrm{TiCL}}T:\varphi$ 
					then $\Gamma\vdash_{\mathrm{TiCL}}T':\varphi$ (induction hypothesis).
					Further suppose that 
						\begin{itemize}
							\item $\pair{T}{T''}\reduc{a}\pair{T'}{T''}$ and 
									$\Gamma\vdash_{\mathrm{TiCL}}\pair{T}{T''}:\varphi'$.
									Hence, 
										first,
											$T\reduc{a}T'$ by the definition of iCL-reduction, and,
										second, 
											$\Gamma\vdash_{\mathrm{TiCL}}T:\varphi''\limp\varphi'$ and
											$\underline{\Gamma\vdash_{\mathrm{TiCL}}T'':\varphi''}$ for 
												some $\varphi''$ (without any restriction), by
													TiCL-typing.
									Hence $\underline{\Gamma\vdash_{\mathrm{TiCL}}T':\varphi''\limp\varphi'}$, 
										by the induction hypothesis.
									Hence $\Gamma\vdash_{\mathrm{TiCL}}\pair{T'}{T''}:\varphi'$,
										by TiCL-typing.
							\item $\pair{T''}{T}\reduc{a}\pair{T''}{T'}$ and 
									$\Gamma\vdash_{\mathrm{TiCL}}\pair{T''}{T}:\varphi'$.
									Hence, 
										first,
											$T\reduc{a}T'$ by the definition of iCL-reduction, and,
										second,\linebreak 
											$\underline{\Gamma\vdash_{\mathrm{TiCL}}T'':\varphi''\limp\varphi'}$ and
											$\Gamma\vdash_{\mathrm{TiCL}}T:\varphi''$ for 
												some $\varphi''$ (without any restriction), by
													TiCL-typing.
									Hence $\underline{\Gamma\vdash_{\mathrm{TiCL}}T':\varphi''}$, 
										by the induction hypothesis.
									Hence $\Gamma\vdash_{\mathrm{TiCL}}\pair{T''}{T'}:\varphi'$,
										by TiCL-typing.
									
						\end{itemize}
			\end{description}
		\item[(Stage $\underline{\alpha\in\agents^{*}}$)]\ 
			\begin{description}
				\item[Base case ($\alpha=\epsilon$)] Suppose that 
					$T\reduc{\epsilon}T'$ and $\Gamma\vdash_{\mathrm{TiCL}}T:\varphi$.
					Hence $T=T'$, and thus $\underline{\Gamma\vdash_{\mathrm{TiCL}}T':\varphi}$.
				\item[Inductive step] Suppose that 
					for all $T,T'\in\mathcal{T}$, 
						if $T\reduc{\alpha}T'$ and $\Gamma\vdash_{\mathrm{TiCL}}T:\varphi$ then 
						$\Gamma\vdash_{\mathrm{TiCL}}T':\varphi$ (induction hypothesis).
					Further suppose that $T\reduc{a\cdot\alpha}T'$ and $\Gamma\vdash_{\mathrm{TiCL}}T:\varphi$.
					Hence there is $T''\in\mathcal{T}$ such that 
						$T\reduc{a}T''$ and 
						$T''\reduc{\alpha}T'$.
					Hence $\Gamma\vdash_{\mathrm{TiCL}}T'':\varphi$, by the previous stage.
					Hence $\underline{\Gamma\vdash_{\mathrm{TiCL}}T':\varphi}$, by the induction hypothesis.
			\end{description}
		\item[(Stage $\alpha\in\agents^{\omega}$)] 
			For proving our statement about infinite words $\alpha\in\agents^{\omega}$, say $P(\alpha)$, 
				it suffices to prove $P$ for 
					all finite approximations $\alpha_{n}$---the 
						(observable) finite prefixes of length $n\in\mathbb{N}$---of $\alpha$.
			Of course, 
				finite prefixes are finite words, and
				in the previous stage we have just proved that 
					$P$ is true for all finite words.
			Hence, $P$ is true for all infinite words $\alpha$ too.
	\end{description}

\section{Singleton society}\label{appendix:SingletonSociety}
For $\agents=\set{a}$ and $\messages\setminus\set{\Kcomb{a},\Scomb{a}}$, 
	bear in mind 
		that $\messages$ is a function of $\agents$, and
		that LiP has been designed for truly interactive cases, that is, cases where $|\agents|>1$ and 
			\emph{not} for non-interactive or degenerately interactive cases, that is, cases where $|\agents|=1$.	
So \textbf{\emph{when $\agents=\set{a}$ and $\messages\setminus\set{\Kcomb{a},\Scomb{a}}$,
	$\messages$ is actually strictly smaller than when $\agents\supsetneq\set{a}!$}}
In particular when $|\agents|>1$,
	obviously neither total knowledge nor epistemic indifference holds,
	nor does proof indifference hold.
For the fortunate failure of proof indifference when $|\agents|>1$, 
	consider the following doubly minimal counter-example.
Without loss of generality, 
	let $\agents\defeq\set{a,b}$ such that $a\neq b$.
Then $\LiPded\proves{b}{\knows{a}{b}}{a}{\emptyset}$ 
	(instance of self-knowledge, \cf Theorem~\ref{theorem:SomeUsefulDeducibleLogicalLaws}), but
		$\not\LiPded\proves{a}{\knows{a}{b}}{a}{\emptyset}$ intuitively, and
			also formally.
Just imagine a state in which $a$ does not know $b$'s name string, \cf 
	Definition~\ref{definition:SemanticIngredients} and \ref{definition:KripkeModel};
		then giving her her own name string, which 
			she already knows anyway, 
				will not make her know $b$'s; and
					then apply the contraposition of axiomatic soundness, \cf
						Theorem~\ref{theorem:Adequacy}.
The counter-example is doubly minimal in the sense that both 
	the involved proof terms ($a$ and $b$ are atomic terms) as well as  
	the involved proof goal ($\knows{a}{b}$ is an atomic proposition about atomic terms) are minimal. 
Note that 
	we could of course conceive LiP 
		without the $\knows{a}{a}$-axiom 
			for some or even all $a\in\agents$ and arbitrary $\agents$.
In particular when $\agents=\set{a}$,
	excluding $\knows{a}{a}$ from $\Gamma_{1}$ definitely makes sense,
		since agent names really make sense only for non-empty non-singleton societies.
In such a system, say $\LiP^{-}$,
	obviously none of the singleton-society laws of LiP would hold for $a$, and
	thus also non-interactive, singleton-society examples 
		(\eg Kripke's Red Barn Example in \cite{JustificationLogic}) could be faithfully formalised.
The price to pay for $\LiP^{-}$ would be, 
	first, the failure of the following laws:
		self-neutral pair element,
		self-neutral proof element, and, \cf Theorem~\ref{theorem:SomeUsefulDeducibleLogicalLaws}, 
		self-truthfulness, and
		the left implication of self-truthfulness \emph{bis}; and thus, 
	second, the impoverishment of the proof-term structure 
		from an idempotent commutative monoid (\cf Corollary~\ref{corollary:ProofEquality})
		to an idempotent commutative semigroup (loss of the neutral element).
(The failure of these laws does \emph{not} imply that their negation succeeds.)
However, how this price is appreciated    
	eventually depends on the considered application.
For example, 
	the failure of self-truthfulness  
		could even be considered desirable:
if we were to exclude $\knows{a}{a}$ from $\Gamma_{1}$, 
		we would actually exclude $(\proves{M}{\phi}{a}{\community})\limp\phi$ from being a theorem in 
			the resulting logical system $\LiP^{-}$ for \emph{all} $M\in\messages$, like in
					the G\"odel-L\"ob Logic of (non-interactive) Provability GL 
						\cite{LogicProvability,ProvabilityLogic}.

\section{Atomic propositions}\label{appendix:AtomicPropositions}
We partially 
	(pre-)instantiate the set of atomic propositions of our logic with our chosen propositional constants (primitives), and (as usual)  
	give the possibility of (further, post-)instantiating this set by means of propositional variables to the users of our logic.
Such a pre-instantiation (at definition time of the logic) 
	is not unusual (many authors choose at least `true' as a predefined primitive), and 
	is of course compatible with substitution invariance (Proposition~\ref{proposition:Hilbert}):
in particular, it does not matter 
	(1) when, at definition time (pre) or at use time (post),
	(2) how many (none, one, or more), nor 
	(3) by what naming mechanism (for example by means of our term constructors, 
			\emph{nota bene} even without term variables)
propositional variables of a logic are instantiated and thus fixed (made constant).
Otherwise, 
	(standard) propositional substitution would be a flawed abstraction mechanism.
What matters is that propositional constants and variables are always treated as such, that is, 
	not substitutable and substitutable, respectively, which 
		is of course what we do and our users must do.
By definition,
	the `substitution' in `substitution invariance' of course only applies to 
		what is propositionally substitutable (and properly changing this base 
			does of course not change the invariance property).
Violating this standard precept by mistreating propositional constants as (substitutable) variables 
	can be misused to wrongly prove the inconsistency of logical theories. 
For such a mistreatment example, 
	simply substitute 
		the propositional constant `false' for 
		the (whole) propositional constant and theorem `$\emptyset\in\{\emptyset\}$' in Set Theory, where 
			$\emptyset$ designates the empty set and 
			the curly brackets are set term constructors.
Similarly but even more simply and wrongly,
	substitute `false' for `true' in any logical theory.

The present, still one-sentence formulation of the fourth bullet in Definition~\ref{definition:LiPLanguage}
	replaces an earlier one that suffers from an accidental artefact that can create confusion.
The reader should mentally replace the corresponding formulation in 
	\cite{KramerICLA2013}, \cite{KramerIMLA2013}, and \cite{LIiP:ACMTOCL}.
We stress that 
	none of these papers nor the present one relies on 
		this unintended artefact, which thus is clearly accidental.
Furthermore, 
	we highlight for 
		all four papers (as already done for \cite{LIiP:ACMTOCL} therein)
			that they describe independent (though of course related) constructions in the sense of 
				not relying on each other for their technical existence.

\end{document}